\documentclass[11pt]{article}

\usepackage[margin=1in]{geometry}

\usepackage{amsmath,amssymb,amsthm}
\usepackage{hyperref}
\usepackage[capitalize]{cleveref}
\newtheorem{theorem}{Theorem}
\newtheorem{definition}{Definition}[section]
\newtheorem{claim}[definition]{Claim}
\newtheorem{lemma}[definition]{Lemma}
\newtheorem{corollary}[definition]{Corollary}
\newtheorem{remark}[definition]{Remark}
\usepackage{graphicx}
\usepackage{xcolor}
\usepackage{thm-restate}
\usepackage{paralist}

\usepackage[ruled,linesnumbered,lined,boxed,commentsnumbered]{algorithm2e}

\newcommand{\sgn}{\mathrm{sign}}
\renewcommand{\sgn}{\mathrm{sgn}}
\newcommand{\bx}{\ensuremath{\mathbf{x}}\xspace}
\newcommand{\bb}{\ensuremath{\mathbf{b}}\xspace}
\newcommand{\bzero}{\ensuremath{\mathbf{0}}\xspace}
\newcommand{\poly}{\ensuremath{\textrm{poly}}}
\newcommand{\polylog}{\text{polylog}}

\newcommand{\eps}{\ensuremath{\varepsilon}\xspace}
\renewcommand{\epsilon}{\ensuremath{\eps}}
\newcommand{\EPS}{\ensuremath{\mathcal{E}}}
\newcommand{\Ex}{\ensuremath{\mathbf{E}}}
\newcommand{\Artur}[1]{\footnote{\textbf{Artur:} \textcolor[rgb]{0.00,0.07,1.00}{#1}}}
\newcommand{\Christian}[1]{\footnote{\textbf{Christian:} \textcolor[rgb]{0.00,0.07,1.00}{#1}}}

            \renewcommand{\Artur}[1]{}\renewcommand{\Christian}[1]{}

\newcommand{\junk}[1]{}

\DeclareMathOperator*{\relu}{ReLU}

\def\RR{{\mathbb{R}}}
\def\NN{{\mathbb{N}}}

\newcommand{\val}{\ensuremath{\mathrm{val}}}
\newcommand{\out}{\ensuremath{\mathrm{output}}}
\newcommand{\sset}{\ensuremath{\mathbb{S}}}

\newcommand{\PP}{\ensuremath{\mathcal{P}}\xspace}   
\newcommand{\NE}{\ensuremath{\mathfrak{N}}\xspace}  
\renewcommand{\NE}{\ensuremath{\mathcal{N}}\xspace} 
\newcommand{\DD}{\ensuremath{\mathcal{D}}\xspace}   
\newcommand{\FN}{\ensuremath{\mathfrak{F}}\xspace}  
\renewcommand{\FN}{\ensuremath{\mathbb{F}}\xspace}  
\newcommand{\OV}{\ensuremath{\mathfrak{B}}\xspace}  
\renewcommand{\OV}{\ensuremath{\mathcal{B}}\xspace} 
\newcommand{\ob}{\ensuremath{\mathfrak{b}}\xspace}  
\newcommand{\IN}{\ensuremath{\mathfrak{I}}\xspace}  
\renewcommand{\IN}{\ensuremath{\mathbb{I}}\xspace}  

\newcommand{\probp}{\ensuremath{\lambda^*}}
\renewcommand{\probp}{\ensuremath{\Xi}}
\renewcommand{\probp}{\ensuremath{\xi}}

\title{\textbf{Testing Neural Networks}}
\title{Property Testing for Networks that Compute}
\title{Property Testing for Computational Networks}
\title{\textbf{Property Testing of Computational Networks}}
\author{Artur Czumaj\thanks{Department of Computer Science and DIMAP, University of Warwick, and University of Cologne. URL: https://www.dcs.warwick.ac.uk/$\sim$czumaj/. Research partially supported by the Key Profile Area (KPA): ``Intelligent Methods
for Earth System Sciences'' at the University of Cologne and by the Centre for Discrete Mathematics and its Applications (DIMAP) at the University of Warwick.} \and Christian Sohler \thanks{Department of Mathematics and Computer Science, University of Cologne.}}
\date{}

\begin{document}


\maketitle

\begin{abstract}
In this paper we initiate the study of \emph{property testing of weighted computational networks 
viewed as computational devices}. Our goal is to design property testing algorithms that for a given computational network with oracle access to the weights of the network, accept (with probability at least $\frac23$) any network that computes a certain function (or a function with a certain property) and reject (with probability at least $\frac23$) any network that is \emph{far} from computing the function (or any function with the given property). We parameterize the notion of being far and want to reject networks that are \emph{$(\epsilon,\delta)$-far}, which means that one needs to change an $\epsilon$-fraction of the description of the network to obtain a network that computes a function that differs in at most a $\delta$-fraction of inputs from the desired function (or any function with a given property).

To exemplify our framework, we present a case study involving simple neural Boolean networks with ReLU activation function. As a highlight, we demonstrate that for such networks, any near constant function is testable in query complexity independent of the network's size. We also show that a similar result cannot be achieved in a natural generalization of the distribution-free model to our setting, and also in a related vanilla testing model.
\end{abstract}


    \thispagestyle{empty}\clearpage\pagenumbering{roman}\tableofcontents\clearpage\pagenumbering{arabic}\setcounter{page}{1}

\section{Introduction}

Since its development in the late 90s \cite{RS96,GGR98}, the area of \emph{property testing} has emerged as an important area of modern theoretical computer science. It considers the relaxation of classical decision problems by distinguishing inputs that belong to a given set (have a given property \PP) from those that are \emph{``far''} from any input that belongs to the set (are \emph{far} from having property \PP). In its most classical setting (see, e.g., survey expositions in \cite{Goldreich17,BY22}), a property \PP is a collection of objects (binary strings, graphs, etc), and being \emph{``far''} is measured by the Hamming distance, namely, in how many places of its representation should an input object be changed so as to have the property \PP. 
As a useful and powerful way of relaxation of the classical decision problems, property testing has found numerous applications and established itself as a significant topic in the study of function properties, graph properties, and beyond. Examples include testing properties of functions such as being a low degree polynomial, being monotone, depending on a specified number of attributes, testing properties of graphs such as being bipartite and being triangle-free, testing properties of strings, and testing properties of geometric objects and visual images such as being well-clustered and being convex (for a more complete picture, see, e.g., \cite{Goldreich17,BY22,CS06,CS10,RS11} and the references therein). Furthermore, property testing is in the center of recent advances in \emph{sublinear algorithms}, as property testing algorithms are frequently ``superfast'' and rely on inspecting small portions of objects and making the evaluation based on such an inspection.

In property testing one assumes an indirect access to the input and its representation that is provided by an \emph{oracle}; e.g., one can query the value of a function on a \emph{specific input}, or 
the existence of a \emph{specific edge} of a graph. The \emph{complexity of the tester} is the number of queries to the oracle.


In this paper, \emph{we study property testing in the context of complex computational devices} (e.g., \emph{computational networks}). We consider a setting in which one wants to study functions that are represented in an indirected way (e.g., as circuits or computational networks, or via complicated algebraic expressions) by computational devices used to determine the value of the functions. One may think about functions that are given \emph{implicitly}, and in order to compute $f$ one has to perform a possibly non-trivial computation on the input. For example, we may have a computational network that \emph{computes} a function $f: \RR^n \rightarrow \{0,1\}$; or
we may have~a function $f: \{0,1\}^n \rightarrow \RR^r$ defined
by a matrix $A \in \RR^{r \times n}$ as $f(x) := Ax$, or $f: \{0,1\}^n \rightarrow \{0,1\}^r$ defined as\footnote{\label{footnote-sign-ReLU}We use standard \textbf{ReLU} and \textbf{signum} functions to ensure that $f$ has range $\{0,1\}^r$; recall that $\relu(z) := \max\{z,0\}$ and $\sgn(z) := 1$ if $z > 0$, $\sgn(0) := 0$, $\sgn(z) := -1$ if $z < 0$. For vectors, the operations are for each coordinate. (Note that the use of these standard definitions make the combination $\sgn(\relu(\cdot))$ \emph{not} redundant: $\sgn(\relu(Ax)) \not\equiv \sgn(Ax)$ for matrices $A$ containing \emph{negative} values --- as it is the case studied in our paper.)} $f(x) := \sgn(\relu(Ax))$.

\junk{
In this paper, \emph{we study property testing in the context of complex computational devices} (e.g., \emph{computational networks}). We consider a setting in which one wants to study functions that are represented in a complicated way (e.g., as circuits or computational networks, or via complicated algebraic expressions) by computational devices used to determine the value of the functions. Here one may think about functions that are given \emph{implicitly}, and so in order to compute $f$ one has to perform some possibly non-trivial computation on the input. For example, we may have a computational network (or a circuit) that \emph{computes} a function $f: \RR^n \rightarrow \{0,1\}$; or
we may have~a function $f: \{0,1\}^n \rightarrow \{0,1\}^r$ defined
by a matrix $A \in \RR^{r \times n}$ as $f(x) := Ax$ or\footnote{We use standard \textbf{ReLU} and \textbf{signum} functions to ensure that $f$ has range $\{0,1\}^r$; recall that $\relu(z) := \max\{z,0\}$ and $\sgn(z) := 1$ if $z > 0$, $\sgn(0) := 0$, $\sgn(z) := -1$ if $z < 0$. For vectors, the operations are for each coordinate. (Note that the use of these standard definitions make the combination $\sgn(\relu(\cdot))$ \emph{not} redundant: $\sgn(\relu(Ax)) \not\equiv \sgn(Ax)$ for matrices $A$ containing \emph{negative} values --- as it is the case studied in our paper.)} $f(x) := \sgn(\relu(Ax))$.
}


Our goal is to determine whether a given computational device computes a function having a given property or if it is \emph{far} from any computational device computing function with a given property.
For example, if for a matrix $A \in \RR^{r \times n}$ we have
    $f(x) := Ax$,
then, informally, we would call the ``device ''
    $Ax$
\emph{far from having property \PP} if there is no matrix $A^* \in \RR^{r \times n}$ that differs from $A$ only in
a small fraction of its entries (i.e., $\|A^*-A\|_0$ is small) such that function
    $f^*(x) := A^*x$
has the property \PP.
Furthermore, to make our study more general and applicable, we will allow to enhance our definition to devices that err on a small number of inputs.\Artur{Can we justify it well? This is an important part to win the argument! I think, the only other work that studied a similar model is testing mixing property of Markov chains due to Batu et al.\ \cite{BFRSW13}, and they use a similar parameterization on both \eps and $\delta$. Maybe one could connect it somehow?}

Observe that while our model can be seen as a special case of the classical property testing framework (see, e.g., references to a similar setting in \cite[Chapter~1]{Goldreich17}), the approach presented in our paper offers a different perspective than that most frequently studied: \emph{our framework focuses on the computational device} --- and its analysis should depend on specific parameters of the device rather than on the inputs or the function $f$. We want to determine if one has to make many changes in the device's representations in order to compute some function having desirable property \PP, or only a small number of changes suffices. And so, for example, if for a matrix $A \in \RR^{r \times n}$ we define $f(x) := Ax$ or $f(x) := \sgn(\relu(Ax))$, then it is natural to define the oracle to access individual entries from $A$ (rather than to allow queries to the value of $f(x)$ for any specific input~$x$).


\begin{figure}[t]
\centerline{\includegraphics[width=.98\textwidth]{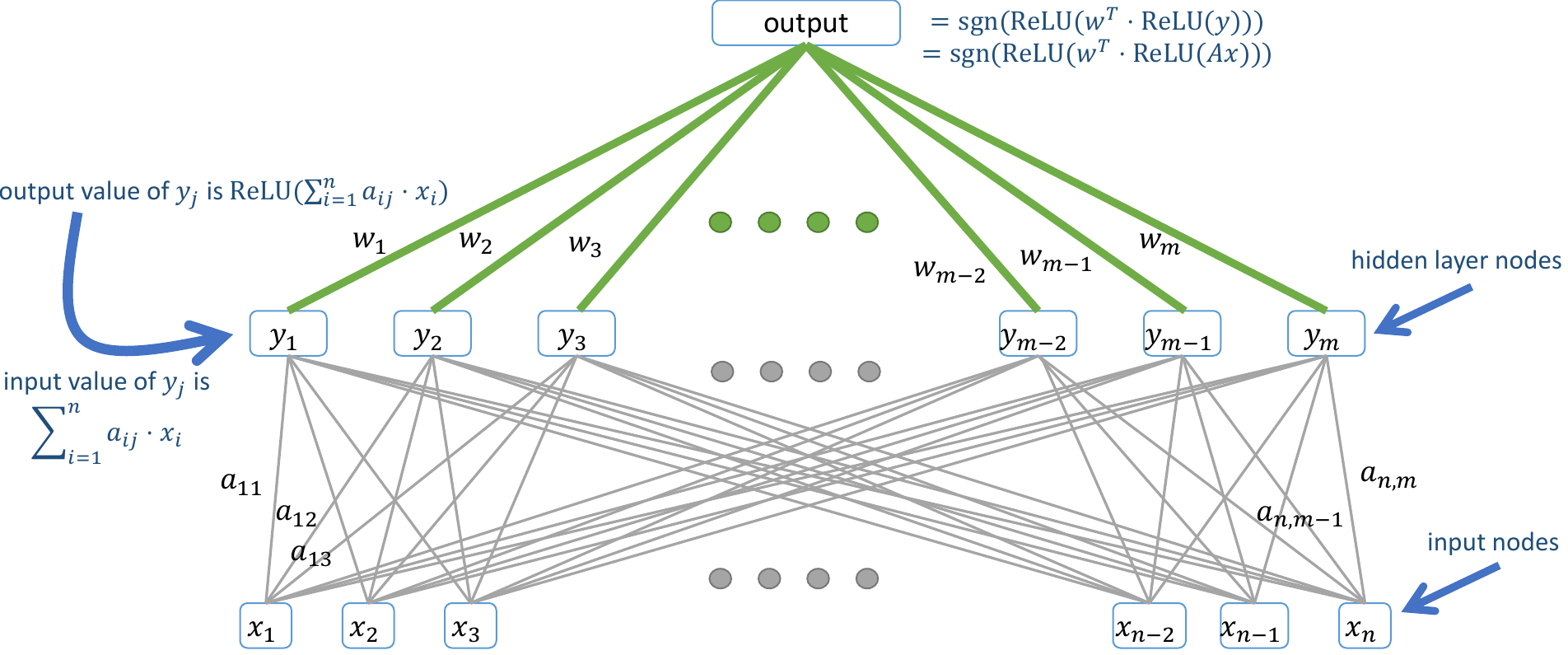}}
\caption{\small
A ReLU network with $n$ input nodes with inputs $x_1, \dots, x_n$, a single hidden layer with $m$ hidden layer nodes with values $y_1, \dots, y_m$, and one output node. Every input node $x_i \in \{0,1\}$ is connected to every hidden layer node $y_j$ by an edge with real weight $a_{ij} \in [-1,1]$ and every hidden layer node $y_j$ is connected to the output node by an edge of real weight $w_j \in [-1,1]$. The value of node $y_j$ is $\sum_{i=1}^na_{ij} x_i$, which after applying the ReLU activation function gives $\relu(\sum_{i=1}^na_{ij} x_i)$.
The Boolean function computed by the network is equal to $\sgn\left(\relu\left(\sum_{j=1}^m \left(w_j \cdot \relu(\sum_{i=1}^na_{ij} x_i)\right)\right)\right) \equiv \sgn\left(\relu(w^T \cdot \relu(Ax))\right)$.
}
\label{fig:ReLU-network-intro}
\end{figure}


\paragraph{Case study: Testing basic neural ReLU networks.}
Our setting is very open-ended, and in order to make our study more concrete, we will consider our framework for a specific example of complex computational networks. Our goal is to take a representative example of computational networks that is generic and simple, yet complex enough to go beyond the study relying on random sampling. We will study an example of the most basic yet highly non-trivial \emph{neural ReLU networks} seen as computational devices from \emph{point of view of property testing}. Consider a trained feedforward network and view it as a circuit that receives an input vector $x \in \{0,1\}^n$ and computes some output. In the most basic case of the network with one output, one hidden layer, and the ReLU activation function (see also Figure~\ref{fig:ReLU-network-intro}), when a network is interpreted as a binary classifier, for a vector $w \in [-1,1]^m$ and matrix $A \in [-1,1]^{m \times n}$ it computes a Boolean function of the form
%
\begin{align*}
    f(x) &:= \sgn\left(\relu(w^T \cdot \relu(Ax))\right) \enspace.
\end{align*}
%

\noindent
(Our setting can be naturally extended to the study of more complex neural networks (functions) with \emph{multiple outputs} and \emph{multiple hidden layers}, where for matrices $W_0, \dots, W_{\ell}$ with $W_i \in [-1,1]^{m_{i+1} \times m_i}$, we compute a function $f:\{0,1\}^{m_0} \rightarrow \{0,1\}^{m_{\ell+1}}$ defined as $f(x) := \sgn(\relu(W_{\ell} \cdot \ \dots \ \cdot \relu(W_1 \cdot \relu(W_0 \cdot \relu(x))) \dots ))$; for more details, see \Cref{def:ReLU-mol} and Figure~\ref{fig:ReLU-network-mL}.)

In order to facilitate our analysis in the property testing framework, we will consider property testing algorithms that are given the oracle access to the weights of the network\footnote{Which for networks with one output and one hidden layer, is to query for the entries of matrix $A$ and vector~$w$.} and the objective is to accept networks that compute functions that have a certain property of interest and reject networks that compute functions that are far from having that property for almost all inputs.


\paragraph{Our motivation.}
Our study is motivated by the modeling of computational networks in the framework of property testing, and to instantiate our framework in the context of basic neural networks. Taking the point of view of property testing and complexity theory, our goal is to obtain a better understanding of the relation between the (local) structure of a computational network
and the function it computes. Our
motivation stems from the fact that property testing has been a successful tool to reveal relations between the local structure of an object (here, a computational device) and its global properties that can provide combinatorial insights about the structure of an object. Examples of such insights
include the characterization of property testing for dense graphs and its relation to the regularity lemma \cite{AS08,AFNS09}, the relation of property testing of hyperfinite graphs to the theory of amenable groups \cite{E12,NS13}, or the line of research that started with testing connectivity \cite{GR02} and then studied sublinear algorithms for approximating MST cost \cite{CRT05,CEFMNRS05,CS09} and whose structural findings formed a building block in finding the currently fastest approximation algorithms for low dimensional MST approximation \cite{AM16}.

To showcase the versatility and power of our framework, we study \emph{networks that are simple}, yet whose \emph{computational properties are complex}. The reliance on \emph{non-linear ReLU operators} is not only something what makes neural deep networks so powerful, but this also makes their study challenging (observe that, say, for $f(x) := Ax$, simple random sampling of entries of matrix $A$ would provide sufficient information to understand many properties of such function; the non-linear ReLU operators make the study and the analysis more complex and interesting).
Feedforward neural networks, as those studied in the simplest form in our paper, form a fundamental part of modern deep networks \cite{LBH15} such as transformers and the attention mechanism  \cite{VSPUJGKP17}, which in turn are the main building block of large language models such as ChatGPT \cite{RNSS18}. Unfortunately, despite their relatively simple setting, we still lack a good understanding of the computation process of these networks, leading to serious questions on the trust and transparency of the underlying learning algorithms. Therefore, even though \emph{our main focus is on the complexity study}, we hope that the lens of property testing will contribute to a better understanding of the relation between the network structure and the computed function or properties of that function. Since the analysis of property testing algorithms relies on establishing close links between the object's local structure and the tested property, we believe that the study of simple yet fundamental models of ReLU networks may provide a new insight into such networks. Additionally, we observe that the basic notion of a relaxed decision problem in property testing has some appealing properties in the context of understanding neural networks: Typically, we would like a neural network to be robust to small changes in the network, as otherwise there is a higher risk of adversarial attacks \cite{SZI13}.
Furthermore, during network training,
\emph{dropout} \cite{SHKSS14} is often used to improve generalization, i.e., a random fraction of input nodes is dropped during training for each training example.


\subsection{Overview of our results}

In this paper, we \emph{introduce a property testing model for computational networks} through a study of \emph{basic neural networks}. We first focus on a simple network structure: a fully connected feedforward network with the ReLU activation function, one hidden layer, and one output node (see Figure~\ref{fig:ReLU-network-intro} and \Cref{def:ReLU}). Later, we generalize our model to a more complex setting with multiple hidden layers and multiple output nodes (see Figure~\ref{fig:ReLU-network-mL} and \Cref{def:ReLU-mol}). We view the network as an \emph{arithmetic circuit} that computes a Boolean function. Our objective is to design property testing algorithms that sample weights of the network and accept, if the network computes a function with a specific property, and reject, if the network is $(\epsilon,\delta)$-far from computing any function with the property, by which we mean that one has to change more than an $\epsilon$-fraction of the network to obtain a function that differs from the target functions on less than a $\delta$-fraction of the inputs.

We believe that one important contribution of this paper is \emph{introduction of a model that allows a non-trivial study of computational networks (e.g., neural networks) from the lens of property testing}.

Next, we provide a \emph{comprehensive study of our model}, focusing on the example of the constant $0$ and the OR-function, for which we design constant-complexity property testers. Further, we consider other natural variants of our model for which we provide super-constant lower bounds demonstrating their complexity limitations. This shows that our model introduces a useful notion of being far that leads to interesting questions and results from the algorithmic and complexity perspective. Our results extend beyond the $0$/OR-functions and we give some positive results showing the testability of some Boolean functions and classes of functions.

Our \emph{main technical contributions} are (see \cref{section:Overview} for more details):
%
\begin{enumerate}[(1)]
\item we study our new model on the problem of \emph{testing some very simple yet fundamental functions}, 
\item we demonstrate limitations of possibility of extending our results to a \emph{distribution-free model},
\item we make \emph{first steps towards understanding of constant time testable properties} in our model, 
\item we extend the model and results to \emph{more complex networks with multiple outputs and layers}.
\end{enumerate}

From the technical perspective, the main and most challenging contributions are (2) the \emph{lower bound for the distribution-free model} and (4) the \emph{extension to multiple layers and outputs}.\Artur{Sounds as overdoing: first we call (1)--(4) as main contributions, and now (2),(4) as main contributions.}

In the following, after briefly reviewing related work, we introduce our property testing model in detail. Then, in \cref{section:Overview}, we summarize our results and provide a detailed technical overview.


\subsection{Related work}
\label{subsec:related-works}


\paragraph{Property testing.}
Since its development in the late 90s \cite{RS96,GGR98}, the area of property testing has emerged as an important area of modern theoretical computer science, see survey expositions in \cite{BY22,Goldreich17,Ron08}. Although a close relationship between property testing and properties of functions, and in the context of learning has been known since the beginning of the area (these topics were explored as early as in \cite{RS96,GGR98}), \emph{the setting presented in this paper has received little attention}. The classical framework of testing properties of functions typically assumes an oracle access to the functions and/or the inputs rather than a computational device, see \cite{BY22,Goldreich17}. (For example,
most function property testers assume oracle access to the values of a given function on a specific input, and typically, graph testers assume access to the edges of the input graph
.)
Similarly, one normally studies networks in the context of their combinatorial properties rather than by analyzing ``functions'' they compute. The only work that we are aware of that considers a model similar to ours is the analysis of testing mixing property of Markov chains by Batu et al.\ \cite{BFRSW13}. And so, we are not aware of any property testing work for basic logical or arithmetic circuits (even a related testing satisfiability in \cite{AS02} uses a different setting), though some works on testing branching programs \cite{FNS04,N02} or whether a function can be represented by a small/shallow circuit/branching program \cite{DLMORSW07,FKRSS04,PRS02}, are of similar flavor. Other related works include property testing of some matrix properties (e.g., \cite{BLW19}, where testing the rank of a matrix $A$ has implications for some properties of $Ax$), but these works use a different setting, rely on linear operators, and do not apply to non-linear operators like those studied in our paper.


\paragraph{Neural networks.}
The complexity of neural networks has received significant attention in the past. In particular, threshold circuits
have been studied as a model of computation and for corresponding complexity class $\textsf{TC}^0$ we know $\textsf{AC}^0 \subsetneq \textsf{TC}^0 \subseteq \textsf{NC}^1$ \cite{V99}. The problem of verifying whether there exist weights such that an input circuit computes an input Boolean function
    (threshold circuit loading)
is known to be \textsf{NP}-complete \cite{J88}, even when the circuit is restricted to three neurons \cite{BR88}. For more details on the computational complexity of
    networks, see, e.g., \cite{O94}.

In general, feedforward networks (acyclic networks) are universal approximators when viewed as computing a continuous function, i.e., they can arbitrarily well approximate any Borel-measurable function between finite dimensional spaces provided a sufficient number of neurons (nodes) in the hidden layer exist \cite{H89}.
Furthermore, width-bounded networks (but generally with a large number of layers) are universal approximators \cite{LPWHW17}. The VC-dimension of neural networks with linear threshold gates has been studied and there is an upper bound of $O(m\log m)$ where $m$ is the number of linear threshold gates \cite{BH88}. These results imply bounds on the number of samples required to train the network. However, for our setting they do not seem to be directly relevant since one cannot evaluate a training example without querying the whole network.

There has been extensive research on the computational complexity of training constant depth ReLU networks. It is known that even for a single ReLU unit, the problem of finding weights minimizing the squared error of a given training set is hard to approximate \cite{GKMR21,MR18,DWX20} (cf.\! \cite{BDL22} for a slightly different proof/result). Further results show that even for a single neuron, the training problem parameterized by the dimension $d$ of the training data is $W[1]$-hard and there is a $n^{\Omega(d)}$ conditional lower bound assuming the exponential time hypothesis \cite{FHN22}. Studies have also explored learning linear combinations of ReLUs from Gaussian distributions (e.g., the first algorithm for a constant number of ReLU units has been recently developed \cite{CDGKM23}).

\section{Property testing model}
\label{sec:property-testing-model}

We begin our study with a feedforward network with one hidden layer and ReLU activation function, see Figure~\ref{fig:ReLU-network-intro}.

The network is fully connected between the layers, and has $n$ input nodes labeled $\{1,\dots, n\}$, a hidden layer with $m$ nodes labeled $\{1,\dots,m\}$, and one output node. (We later generalize our model to multiple hidden layers and output nodes, see
    \cref{remark:1-layer->multiple-layers}.)
We associate a \emph{binary vector} $x = (x_1,\dots,x_n)$ with the input nodes, so that input node $i$ has value $x_i$. Between the input nodes and the hidden layer there is a \emph{complete bipartite graph with real edge weights} from $[-1,1]$. We define $a_{ij}$ to be the weight between input node $i$ and node $j$ of the hidden layer. The \emph{value of an input node is $x_i$}. The \emph{value computed at a hidden layer node $j$ is $\sum_{i} a_{ij} x_i$}. Thus, we can write the values at the hidden layer by a matrix multiplication $Ax$ for an $m \times n$ matrix $A$.

To the values at the hidden layer, we apply the \emph{ReLU activation function}, which is defined as $\relu(y) := \max\{0,y\}$. For a vector $y = (y_1,\dots,y_m)^T$, we define $\relu(y)$ to be $\relu(y) := (\relu(y_1), \dots, \relu(y_m))^T$. We can now write the output of the hidden layer nodes as $\relu(Ax)$.

Between the hidden layer and the output node there is a complete bipartite graph with real weights from $[-1,1]$. Let $w_j$ denote the weight of the edge connecting hidden layer node $j$ to the output node. If $z = \relu(Ax)$ is the vector of hidden layer outputs, then $w^Tz$ is the \emph{value of output node}, where $w = (w_1, \dots, w_m)^T$. Thus, on input $x$ the value at the output node is $w^T \cdot \relu(Ax)$.


\paragraph{ReLU network as a binary classifier.}

We interpret the network as a \emph{binary classifier}: the network assigns each input vector to one of two classes, defined as $0$ and $1$. If the value at the output node is at most $0$ we assign class $0$, if it is larger than $0$ we assign class~1. In this paper we focus on the setting when $x \in \{0,1\}^n$ is a binary vector. Thus, we may view the neural network as a function $f: \{0,1\}^n \rightarrow \{0,1\}$ that computes the value $1$ if the output value of the network is greater than $0$, and $0$ if the output value is at most $0$. We can write\footnote{Note that the outer $\relu$ operation is \emph{not redundant} in $\sgn(\relu(\cdot))$ because the argument $w^T \cdot \relu(Ax)$ can be negative, in which case $\sgn(\cdot)$ would return $-1 \notin \{0,1\}$; see also the discussion in footnote~\ref{footnote-sign-ReLU}.} $f(x) = \sgn(\relu(w^T \cdot \relu(Ax)))$. 
For a given input $x \in \{0,1\}^n$, we call $f(x)$ the \emph{output of the~network}.

Our goal is to understand how the local structure of the network is related to the function $f$ or to properties of that function. We will address this question by studying the function from a \emph{property testing} point of view. With the above discussion, we can make the following definition.

\begin{definition}[\textbf{ReLU network with \emph{one hidden layer} and \emph{single output}}]
\label{def:ReLU}
A ReLU network with $n$ input nodes, $m$ hidden layer nodes, and one output node is a pair $(A,w)$, where $A\in[-1,1]^{m\times n}$ and $w\in[-1,1]^m$. The binary function $f:\{0,1\}^n \rightarrow \{0,1\}$ computed by the ReLU network $(A,w)$~is
%
\begin{align*}
    f(x) &:= \sgn(\relu(w^T \cdot \relu(Ax))).
\end{align*}
%
\end{definition}


\begin{remark}\rm
\label{remark:1-layer->multiple-layers}
For simplicity of presentation, we describe here the model only for ReLU networks with one output and one hidden layer. Our study can be naturally extended to ReLU networks with multiple outputs (see \cref{sec:multiple-outputs}) and multiple hidden layers (see Sections~\ref{sec:multiple-layers}--\ref{sec:ReLU-testing-constant-function-2-sided-mL-mo}), where for matrices $W_0, \dots, W_{\ell}$ with $W_i \in [-1,1]^{m_{i+1} \times m_i}$, the network computes a function $f: \{0,1\}^{m_0} \rightarrow \{0,1\}^{m_{\ell+1}}$
\begin{align}
    f(x) &:=
    \sgn(\relu(W_{\ell} \cdot \ \dots \ \cdot \relu(W_2 \cdot \relu(W_1 \cdot \relu(W_0 \cdot \relu(x)))) \dots ))
    \enspace.
    \label{eq:def-f-mL}
\end{align}
%

\noindent
For more details, see Figure \ref{fig:ReLU-network-mL} and the 
analysis in Sections \ref{sec:multiple-outputs}--\ref{sec:ReLU-testing-constant-function-2-sided-mL-mo} (see also \cref{subsection:Overview-multiple-outputs+layers}).
\end{remark}


\subsection{Property testing: ReLU networks with \emph{one hidden layer} and \emph{single output}}
\label{subsec:property-testing-model}

We introduce a new property testing model for computational networks on an example of neural ReLU networks. We view (see Figure~\ref{fig:ReLU-network-intro}) the network as a weighted graph $G = (V, E)$, where $V = I \cup H \cup O$ with $I$ being the set of \emph{input nodes}, $H$ the set of \emph{hidden layer nodes}, and $O$ the set of \emph{output nodes} (frequently $|O| = 1$). The nodes in $I, H$ and $O$ are labelled from $1$ to $|I|, |H|$ and $|O|$, respectively. Furthermore, $E$ contains all edges between $I$ and $H$, and between $H$ and $O$. The edges between $I$ and $H$ are also referred to as \emph{first layer edges} and the edges from $H$ to $O$ as \emph{second layer edges}. We assume that we have \emph{query access to the edge weights}: we can specify nodes number $i$ in $I$ and number $j$ from $H$, and query the value $a_{ij}$ of that edge in $O(1)$ time. Similarly, we can access the entries of $w$ by specifying the corresponding hidden layer node.

\junk{
Our objective is to design (randomized) sampling algorithms that decide whether a given neural network has a certain property or is far away from the property. To achieve this goal, we must first define a distance measure of neural networks. We view the network as a weighted graph and use the (relative) edit distance. However, since we want to take into consideration that there is a different number of edges on each layer, \emph{for two networks $(A_1,w_1)$ and $(A_2,w_2)$} with $n$ input nodes, $m$ hidden layer nodes, and one output node, we \emph{define their distance} as
$
    \max \left\{ \tfrac{\|A_1-A_2\|_0}{mn}, \tfrac{\|w_1-w_2\|_0}{m}\right\},
$
where we use the standard notation $\|\cdot\|_0$ to denote the number of non-zero entries in a matrix/vector.
}

Our objective is to design (randomized) sampling algorithms that decide whether a given neural network has a certain property or is far away from the property. To achieve this goal, we must first define a distance measure of neural networks. We view the network as a weighted graph and use the (relative) edit distance. However, since we want to take into consideration that there is a different number of edges on each layer, \emph{for two networks $(A_1,w_1)$ and $(A_2,w_2)$} with $n$ input nodes, $m$ hidden layer nodes, and one output node, we \emph{define their distance} as
%
\begin{align*}
    \max \left\{ \tfrac{\|A_1-A_2\|_0}{mn}, \tfrac{\|w_1-w_2\|_0}{m}\right\} \enspace,
\end{align*}
where we use the standard notation $\|\cdot\|_0$ to denote the number of non-zero entries in a matrix/vector.

In this paper, we view neural networks as computing Boolean functions and study the question whether a neural network computes 
a function that has a certain \emph{property}. We begin with the study of networks that compute a specific function. Since there are often many networks that compute the same function, the question of whether a given network computes a certain function may be viewed as testing whether an input has a certain specific property in standard property testing.

Given a function $f:\{0,1\} \rightarrow \{0,1\}$ to be tested, a \emph{property tester} is an algorithm with query access to the weights of a network (with $n$ input nodes and $m$ hidden layer nodes) that accepts with probability at least $\frac23$ every network that computes $f$ and rejects with probability at least $\frac23$ every network that is \emph{far} from computing $f$, where we parameterize the notion of being far with parameters $\eps, \delta$. Parameter $\epsilon$ refers to the distance between the networks (defined above) and
$\delta$ refers to the distance between the computed function with respect to the uniform~distribution\footnote{While distance to a property is often parameterized using only the parameter \eps, a relaxation with another parameter has been studied before. In fact, a very similar model has been used in \cite{BFRSW13} to test if a Markov chain is mixing. Similarly to our problem, the graph/Markov chain can be seen as a computational model parameterized in terms of \eps and $\delta$, where \eps refers to the edit distance of the Markov chain and $\delta$ to the $\ell_1$-distance between the $t$-step distributions. We believe that a similar relaxation is useful in our case as the relation between the network and the function is rather complex, and we have little hope to obtain results without this second parameter~$\delta$.}.

Before we proceed with our definitions, let us recall the notion of \emph{property of Boolean functions}, which is any family $\PP = \bigcup_{n \in \NN} \PP_n$, where $\PP_n$ is a set of Boolean functions $f:\{0,1\}^n \rightarrow \{0,1\}$.

\junk{
\begin{definition}[\textbf{ReLU network being far from computing a function}]
\label{def:ReLU-farness-from-function}
Let $(A,w)$ be a ReLU network with $n$ input nodes and $m$ hidden layer nodes. $(A,w)$ is called \textbf{$(\epsilon,\delta)$-close to computing a function} $f:\{0,1\}^n \rightarrow \{0,1\}$, if one can change the matrix $A$ in at most $\epsilon nm$ places and the weight vector $w$ in at most $\epsilon m$ places to obtain a ReLU network that computes a function $g$ such that $\Pr[g(x)\ne f(x))]\le \delta$, where $x$ is chosen uniformly at random from $\{0,1\}^n$. If $(A,w)$ is not $(\epsilon,\delta)$-close to computing $f$ we say that it is \textbf{$(\epsilon,\delta)$-far from computing $f$}.
\end{definition}

We can extend this definition to properties of functions.
\begin{definition}[\textbf{Properties of Boolean functions}]
\label{def:props-of-functions}
A \textbf{property of Boolean functions} is a family $\PP = \bigcup_{n \in \NN} \PP_n$, where $\PP_n$ is a set of Boolean functions $f: \{0,1\}^n \rightarrow \{0,1\}$.
\end{definition}
}

\begin{definition}[\textbf{ReLU network being far from a property of functions}]
\label{def:ReLU-farness-from-property}
Let $(A,w)$ be a ReLU network with $n$ input nodes and $m$ hidden layer nodes. $(A,w)$ is called \textbf{$(\epsilon,\delta)$-close to computing a function $f:\{0,1\}^n \rightarrow \{0,1\}$ with property $\PP = \bigcup_{n \in \NN} \PP_n$}, if one can change the matrix $A$ in at most $\epsilon nm$ places and the weight vector $w$ in at most $\epsilon m$ places to obtain a ReLU network that computes a function $g$, such that there exists a function $f \in \PP_n$ for which $\Pr[g(x)\ne f(x)] \le \delta$, where $x$ is chosen uniformly at random from $\{0,1\}^n$. If $(A,w)$ is not $(\epsilon,\delta)$-close to computing $f$ with property \PP then we say that it is \textbf{$(\epsilon,\delta)$-far from computing $f$ with property \PP}.
%
%
\end{definition}

\noindent
If $|\PP_n| = 1$, $f \in \PP_n$, then in \cref{def:ReLU-farness-from-property} we say that a \emph{network is $(\epsilon,\delta)$-close/far from computing~$f$}.

\begin{remark}\label{footnote:not-triangle-inequality}\rm
One could argue that it is highly undesirable that the notion of $(\eps,\delta)$-distance employed in this paper does not satisfy the triangle inequality. However, our model is a natural relaxation of the model with $\delta = 0$, which is exactly the standard model of property testing with a property \PP being the set of networks computing a given function (like the 0-function) \cite{Goldreich17}. So, \emph{it is an inherent feature of property testing that the distance to a property/function often does not satisfy the triangle inequality}: Just as two graphs $G$ and $H$ being close to 3-colorable does not imply that $G$ and $H$ are close to each other, or two networks close to computing the 0-function are not necessarily close.
\end{remark}

\begin{definition}
\label{def:ReLU-testability}
A property $\PP = \bigcup_{n \in \NN} \PP_n$ of Boolean functions is \textbf{testable} (for ReLU networks) \textbf{with query complexity $q(\epsilon, \delta, n, m)$},
if for every $n, m \in \NN$, every $\epsilon\in (0,1)$, and every $\delta \in (0,1)$
there exists an algorithm that receives parameters $n, m, \epsilon, \delta$ as input and
is given query access to an arbitrary ReLU network with $n$ input nodes, $m$ hidden layer nodes and one output node, and that
\begin{itemize}
\item queries the ReLU network in at most $q(\epsilon, \delta, n ,m)$ many places,
\item accepts with probability at least $\frac23$ every network that computes a function from $\PP_n$, and
\item rejects with probability at least $\frac23$ every network that computes a function 
$(\epsilon,\delta)$-far from \PP.
\end{itemize}
If the function $q$ depends only on $\epsilon$ and $\delta$, then we just say that the property \PP is \textbf{testable}.\footnote{In property testing, traditionally (cf. \cite{BY22,Goldreich17}), the ultimate goal of the analysis is to \emph{obtain testers whose query complexity is independent of the input size}, depending only on \eps, $\delta$. Testers with query complexity $o(mn)$ are of interest, but often not as the main objective. The importance of having the query complexity independent of \eps and $\delta$ is especially amplified in our setting: if we query the ReLU network in $q(\epsilon, \delta, n, m)$ many places, then normally the tester will be computing the function on the subnetwork determined by the tested edges \emph{for all possible $2^n$ inputs}. However, the number of \emph{relevant} input bits is always upper bounded by $q(\epsilon, \delta, n, m)$ (since only input bits adjacent to the queried edges impact the function computed), and therefore if the query complexity is independent of the input size, 
then the number of inputs to be considered is at most $2^{q(\epsilon, \delta)}$, which is independent on $n$ and $m$. In other words, \emph{testers with query complexity independent of $n$ and $m$ have their running time independent of $n$ and $m$} too.}

A tester that accepts every function from \PP with probability $1$ is called \textbf{one-sided tester}.
\end{definition}

\junk{
We remark that the above notion also applies to the setting of testing whether a network computes a certain function: then, the property will simply consist of the studied function.
}


\section{Technical overview}
\label{section:Overview}


\mbox{We initiate the study of property testing of computational networks on an example of neural networks.}


\subsection{Studying variants of the model on a simple function}
\label{subsection:Overview-0-and-OR}

We begin with the study of the constant $0$-function (the same results are true for the OR function) in our basic model of feedforward \emph{networks with ReLU activation function, one hidden layer, and single output}.
We start by showing that the constant $0$-function is testable in such model. This is the case even though (see \cref{thm:NP-completeness-of-0-function} in \cref{sec:ReLU-testing-0-function-NP-completeness}) the problem of verifying if the network computes the constant $0$-function is NP-hard (note that if we had no ReLU we could simply compute $w^TA$ and check whether $w^TA$ has positive entries; so here ReLU makes the problem more difficult). Our result is summarized in the following theorem (for a proof, see \cref{sec:ReLU-testing-0-OR-function-2-sided}):

\begin{restatable}{theorem}{ZeroORFunction}
\label{thm:ReLU-testing-0-OR-function-2-sided}
Let $(A,w)$ be a ReLU network with $n$ input nodes, $m$ hidden layer nodes, and a single output. Let $\delta \ge e^{-n/16}$, $\frac{1}{m} < \epsilon < \frac{1}{2}$, and $0 < \lambda < \frac{1}{2}$.
%
\begin{enumerate}[(1)]
\item There is a tester that queries $O(\frac{\ln^2(1/\epsilon \lambda)}{\epsilon^6})$ entries from $A$ and $w$, and
\begin{inparaenum}[(i)]
\item accepts with probability at least $\ge 1-\lambda$, if the ReLU network $(A,w)$ computes the constant $0$-function, and
\item rejects with probability at least $1-\lambda$, if $(A,w)$ is $(\epsilon,\delta)$-far from computing the constant $0$-function.
\end{inparaenum}

\item There is a tester with identical properties for computing the OR-function.  \Artur{It's a shorted version; earlier we had a long 3-lines long text: \it (2) There is a tester that queries $O(\frac{\ln^2(1/\epsilon \lambda)}{\epsilon^6})$ entries from $A$ and $w$, and
\begin{inparaenum}[(i)]
\item accepts with probability at least $1-\lambda$, if the ReLU network $(A,w)$ computes the OR-function, and
\item rejects with probability at least $1-\lambda$, if $(A,w)$ is $(\epsilon,\delta)$-far from computing the OR-function.
\end{inparaenum}
}
\junk
{
\item There is a tester that queries $O(\frac{\ln^2(1/\epsilon \lambda)}{\epsilon^6})$ entries from $A$ and $w$, and
\begin{inparaenum}[(i)]
\item accepts with probability at least $1-\lambda$, if the ReLU network $(A,w)$ computes the OR-function, and
\item rejects with probability at least $1-\lambda$, if $(A,w)$ is $(\epsilon,\delta)$-far from computing the OR-function.
\end{inparaenum}
}
\end{enumerate}
\end{restatable}

The query complexity of our tester does not depend on $\delta$ and we only require $\delta \ge e^{-\Omega(n)}$.
\emph{Assuming that $n$ and $m$ are polynomially related, one could also remove this dependence} by introducing a factor $\polylog (1/\delta)$ in the complexity, since for $\delta < e^{-n/16}$ we can read the whole network.

Observe that the \emph{query complexity of our tester in \cref{thm:ReLU-testing-0-OR-function-2-sided} is independent of the network size}. Further, for constant \eps and $\lambda$, the query complexity of our is $O(1)$.

Our property tester for the constant $0$-function (see Algorithm~\ref{alg:AllZeroTester} in \cref{sec:ReLU-testing-0-OR-function-2-sided}) samples a subset of $\poly(1/\eps) \cdot \polylog(1/\lambda)$ input and hidden layer nodes and evaluates whether the induced network with a \emph{suitable bias} at the output node computes the constant $0$-function. While this approach seems natural, the analysis and its use of the bias is non-standard and highly non-trivial.

The first step of our analysis shows that if the network is $(\epsilon,\delta)$-far from computing the constant $0$-function then there is an input $x_0 \in \{0,1\}^n$ on which the network outputs value $\Omega(\epsilon nm)$. Then, we show that sampling a constant number of hidden layer nodes approximates the scaled output value for this particular $x_0$ with sufficient accuracy, giving the scaled output value $\Omega(\epsilon n m)$. In the next step, we show that we can sample a constant number of input nodes and approximate the value at the sampled hidden layer nodes so that the value at the output node is positive. Combining all the steps yields that if the network is $(\epsilon,\delta)$-far then our tester rejects with sufficient probability.

In order to show that the algorithm accepts a network computing the constant $0$-function, we will argue that for all inputs $x \in\{0,1\}^n$ the sampled network computes the constant $0$-function. For that, we take a different view and consider the sample from the input nodes as being fixed. We then observe that restricting our network to the sampled input nodes can also be done by restricting the input to vectors $x\in\{0,1\}^n$ that are 0 for all input nodes that do not belong to the sample. This allow us to interpret our sampling approach as only sampling from the hidden layer. Hence, we can apply a concentration bound and take a union bound over all choices of input vectors that are zero on the nodes not in the sample. We also observe that while the analysis considers $n$-dimensional input vectors we know that these are $0$ except for the sampled input nodes. Thus, in order to compute the output value we only need to query edges between sampled nodes.


\paragraph{Comparison to a vanilla testing algorithm.}
The above result may be viewed as saying that a network that is far from computing the constant $0$-function (or from computing the OR-function) is so already on most of small subnetworks. While this may not seem to be too surprising in the first place, it turns out to be a useful addition to a \emph{vanilla testing algorithm} that samples inputs and evaluates the network on these inputs. This vanilla testing algorithm is arguably the most natural and routinely used approach to test the functionality of neural networks. In fact, it corresponds to the most \emph{standard setting in property testing of functions} (see, e.g., \cite{Goldreich17,BY22}) in which \emph{the oracle allows to query the value of the network on a given input $x$}, rather than to query the network's weights.
In \cref{subsec:comparisons-to-vanilla}, we demonstrate the versatility of our property testing model and show that there are networks that \emph{can be recognized as $(\epsilon,\delta)$-far} (for constant $\epsilon$ and exponentially small $\delta$) from computing the $0$-function \emph{by our tester}, but the \emph{vanilla tester in expectation requires an exponential in $n$ number of samples} in order to find an input that does not evaluate to $0$.


\paragraph{Testing with one-sided error.}
We also consider the testability of the constant $0$-function and the OR-function with \emph{one-sided error}. In \cref{sec:ReLU-testing-0-OR-function-1-sided-upper}, we design a tester with one-sided error and query complexity of $\widetilde{O}(m/\epsilon^2)$. The one-sided error property tester samples a subset of $O(\log m/\epsilon^2)$ input nodes $S$ and rejects only if there exists an input that is zero on all nodes outside of $S$ and that evaluates to $1$ or $0$, respectively. If such an input is found, then it is a counter-example to computing the constant $0$-function (or the OR-function, respectively) and the algorithm can reject.

In \cref{sec:ReLU-testing-0-OR-function-1-sided-hardness} we give a near matching lower bound: any tester (for any constant \eps) for the constant 0-function (or for the OR-function) that has one-sided error has query complexity of $\Omega(m)$. The lower bound follows from the fact that if too few weights are known, one can always ''complete'' the network to enforce it to compute the constant $0$-function
. Thus, the tester always accepts, but there are inputs that are $(\epsilon,\delta)$-far from computing the constant $0$-function.


\paragraph{Why testing of the 0-function?}
We consider the constant 0- and the OR-functions to be representative for natural, simple, yet rich, and highly non-trivial functions in our setting. Furthermore, note that for
functions $f$ and $g$, $f \equiv g$ iff the network $\relu(f-g) + \relu(g-f)$ computes the 0 function, and therefore testing the 0-function is closely related to the question of two networks computing the same output values for all inputs. 
Other related problems that might be in reach are dictatorship and juntas (which look like essentially equivalent to 0-testing, but they are not).


\subsection{A lower bound in a \emph{distribution-free model} (or on the choice of our model)}
\label{subsec:distribution-free model}
It is natural to try to extend our framework to a \emph{distribution-free model} (see, e.g., \cite[Section~2]{GGR98}, \cite[Chapter~12.4]{Goldreich17}
), where the notion of being $(\epsilon,\delta)$-far from a property defines the farness with respect to the second parameter $\delta$ in terms of \emph{arbitrary distribution}. In the distribution-free model the distance between Boolean functions is measured with respect to an unknown distribution rather than the uniform distribution, as in our model considered above. It is also assumed that a property tester has sample access to this distribution. In our case, we will assume \emph{query access to bits of the sample} since our goal is to be sublinear in the input size $n$ or even to have a constant query complexity. While the study of this model is very tempting, unfortunately, we show a \emph{lower bound that rules out sublinear-time testers even for the constant $0$ function} (see \cref{sec:distribution-free-lower-bound}).

\begin{restatable}{theorem}{LowerBound}
\label{thm:testing-distribution-free-0}
For every constant $k>1$, every property tester for the constant $0$-function in the distribution-free model has a query complexity of $\Omega(n^{1-1/k})$.
\end{restatable}

The proof of \cref{thm:testing-distribution-free-0} is very technical and elaborate, but we will sketch here the main idea of the proof for $k=2$ (where for purpose of exposition we use different constants than in the general proof). We define two distributions $\DD_1$ and $\DD_2$ on pairs of a network and an input distribution (see also Figure~\ref{fig:ReLU-network-distribution-free}) such that the networks from the first distribution are (whp.) computing the constant $0$-function and such that the pairs of network and input distribution $\DD_I$ from the second distribution are (whp.) $(\epsilon,\delta)$-far from computing the constant $0$-function with respect to $\DD_I$.

For $\DD_1$ the network will (whp.) compute the constant $0$-function, so the underlying distribution is not relevant. To construct the network we partition the hidden layer nodes into sets $N$ and $P$ such that the nodes in $N$ are connected to the output by edges of weight $-1$ and the nodes in $P$ by edges of weight $1$. In the first layer, the edges to nodes from $N$ are $1$ with probability $\frac23$ and $-1$, otherwise, and the edges to nodes from $P$ are $1$ with probability $\frac12$. Ignoring the effect of the ReLU it is not too hard to see that the network will evaluate (whp.) to $0$ on any fixed input and one can apply a union bound to show that this holds also for all inputs. As we show in the proof, this intuition still holds under the ReLU function (however, it is technically more challenging).

Networks from $\DD_2$ are similar, except that there is a random matching of input nodes and the connections to the nodes from $P$ are sampled for one node of the matching pair and copied for the other. This implies that with probability $\frac12$ the pair of nodes contributes $2$ to the hidden layer node, while the expected contribution to $N$ is $\frac89$ (here the edge weights are sampled independently and the matched pair of input nodes contributes only if both first layer edges are $1$, which happens with probability $\frac49$ and the contribution is $2$). Thus, for any matching pair the expected contribution to the output is linear in the size of the hidden layer. We now define the unknown input distribution to be the uniform distribution on the pairs of matched inputs. The network is far from computing the $0$-function since ''fixing'' one of the $\Omega(n)$ matched pairs requires $\Omega(m)$ changes in the network.

Finally, we observe that distinguishing between the two distributions is expensive as the answers to the queries are identically distributed as long as no matched pair is queried. Having access to samples from the distribution does not help as it is impossible to ''find'' input bits that are 1. Since the matching of input bits is random, we prove that one needs $\Omega(\sqrt{n})$ samples to find a matching pair. We formalize this intuition and parameterize it by replacing pairs of inputs by $k$-sets, where the first $k-1$ inputs are chosen randomly with probability $\frac12$ to be $1$ (and $-1$, otherwise) and the last one is $1$, if there is an odd number of 1s among the first $k-1$. We then modify our construction to get a lower bound for any $k$. The lower bound is one of our main technical contributions, relying on a novel controlling of the expectation of ReLU of sums of binomially distributed random~variables.


\subsection{First steps towards a classification of testable properties?}
\label{subsection:Overview-classification}

Our next goal is to try to understand \emph{which properties of ReLU networks are testable}. Our first result shows that some classes of functions,
\emph{symmetric Boolean functions} or \emph{monotone Boolean functions} have trivial testers
    (that always accept)
  
We prove this by showing that every network is close to computing the constant $0$ function or is close to the OR-function (with $\epsilon, \delta$ not too small).
\footnote{We would like to stress that if two networks are close to computing a given function, then \emph{they are not necessarily} pairwise close to each other (cf. \cref{footnote:not-triangle-inequality}). Similarly, notice that while the closeness result in \cref{theorem:0and1} may on the first glance seem to be unnatural, similar and even stronger claims are known in various scenarios in property testing. For example, in the classical dense graph model (see \cite{Goldreich17}), for $\eps > \frac1n$, every graph is \emph{both}, \eps-close to be connected and \eps-close to be disconnected (
every graph on $n$ vertices can be made connected using at most $n$ edge changes, and can be made disconnected using at most $n$ edge changes). Therefore, the fact that in our property testing framework all properties are close to computing the constant $0$-function or close to computing the OR-function, does not mean that testing of these two properties is trivial, nor it shows that the framework considered here is unnatural.}

\begin{restatable}{theorem}{Closeness}
\label{theorem:0and1}
Let $(A,w)$ be a ReLU network with $n$ input nodes, $m$ hidden layer nodes, and a single output. Let $0 < \delta \le \frac12$ and $\max\{\frac1m, 2 \sqrt{\log(2/\delta)/n}\} \le \epsilon$. Then $(A,w)$ is $(\epsilon,\delta)$-close to computing the constant $0$-function or it is $(\epsilon,\delta)$-close to computing the OR-function.

\end{restatable}

The proof of \cref{theorem:0and1} (see \cref{subsec:everything-close}) first argues that the value at the output (before we apply $\sgn$ and ReLU) typically deviates from its expectation by at most $O(\sqrt{nm})$. Then we show that we can modify the network in an $\epsilon$-fraction of its inputs and increase (or decrease) the expected value at the output bit by more than $\frac12 \epsilon mn$. These two facts can be shown to yield \cref{theorem:0and1}.

Note that the property of monotone functions includes both the 
$0$-function \emph{and} the OR-function. Thus an immediate consequence of \cref{theorem:0and1} is that \emph{monotonicity has a trivial tester} that always accepts if $\epsilon \ge \max\{\frac1m, 2 \sqrt{\log(2/\delta)/n}\}$, and reads the whole input if $\epsilon < \max\{\frac1m, 2 \sqrt{\log(2/\delta)/n}\}$. When $n$ and $m$ are polynomially related, the query complexity of the tester is $\poly(1/\epsilon) \cdot \polylog (1/\delta)$.

\medskip

\Artur{Does the title of \cref{subsection:Overview-classification} fits well the discussion about monotone properties and \cref{thm:monotone-properties-general}?}%
Finally, we study \emph{monotone properties} (which are different objects than monotone functions). A property \PP is \emph{monotone}, if it is closed under flipping bits in the truth table of any function having the property \PP from $0$ to $1$. One can think of monotone properties as being defined by the set of minimal functions under the operation of flipping zeros to ones. We call such a set a \emph{generator}.

We show that every monotone property is testable with query complexity logarithmically on the size of the generator (we assume that the algorithm gets the generator as input; this assumption is justified as this depends on the property only).
In \cref{sec:monotone-properties} we prove the following result.

\begin{restatable}{theorem}{MonotoneFull}
\label{thm:monotone-properties-general}
Let $\frac1m < \epsilon < 1$, $e^{-n/16} \le \delta \le 1$. Let $\PP = \bigcup_n \PP_n$ be an $f(n)$-generatible monotone property with generator $G = \bigcup_n G_n$ with $|G_n| \le f(n)$ for all $n$. There is a property tester with query access to a ReLU network \NE with $n$ input nodes and $m$ hidden layer nodes that
\begin{inparaenum}[(i)]
\item accepts with probability at least $1-\lambda$, if \NE computes a function in \PP,
\item rejects with probability at least $1-\lambda$, if \NE is $(\epsilon,\delta)$-far from computing a function in \PP.
\end{inparaenum}
The tester has query complexity $O(\ln(\frac{f(n)}{\delta\epsilon\lambda})/(\delta \epsilon^4))$.
\end{restatable}


\subsection{Extension to multiple outputs and multiple hidden layers}
\label{subsection:Overview-multiple-outputs+layers}

We can extend our framework to more complex neural networks and demonstrate that our methodology can be applied to a natural generalization of our setting allowing \emph{multiple outputs} and \emph{multiple hidden layers}. We introduce this model in Sections~\ref{sec:multiple-outputs}--\ref{sec:ReLU-testing-constant-function-2-sided-mL-mo} and show the applicability of our techniques by extending the approach for testing the constant $0$-function (and the OR-function) to
testing near constant functions in ReLU networks with multiple hidden layers and outputs: testing for $f(x) = \sgn(\relu(W_{\ell} \cdot \ \dots \ \cdot \relu(W_1 \cdot \relu(W_0 \cdot \relu(x))) \dots ))$, where $W_i \in [-1,1]^{m_{i+1} \times m_i}$.

We begin with the analysis of ReLU networks with multiple outputs and \emph{one hidden layer}. Before we proceed, let us introduce the notion of \emph{near constant functions} with multiple outputs\footnote{We need this notion to exclude trivial cases; notice that since for any ReLU network $f(\bzero) = \bzero$ (and so there is no ReLU network with $\forall_x f(x) = 1$), we consider only \emph{near constant functions}; see also \cref{def:constant-functions-mo} and \cref{remark:constant-functions-and-ReLU}.}, which are constant functions on every input $x$ except $x = \bzero$. We first present a general reduction to single layer networks. The following theorem, central for our analysis, demonstrates that if a ReLU network with $r$ outputs is $(\epsilon, \delta)$-far from computing near constant function \bb, then for at least an $\frac{\epsilon}{2}$-fraction of the output bits $i$, the ReLU network restricted to the output node $i$ is $(\epsilon',\delta')$-far from computing an appropriate near constant function, where $\delta' \sim \delta/r$ and $\epsilon' = \epsilon^2/1025$.

\begin{restatable}{theorem}{TestingConstantFunctionReduction}
\label{thm:reduction-constant-ReLU-mo}
Let \NE be a ReLU network $(A,W)$ with $n$ input nodes, $m$ hidden layer nodes, and $r$ output nodes.
Let $e^{-n/16} \le \delta < 1$ and $16 \cdot \sqrt{\frac{\ln(2m)}{m}} \le \epsilon < \frac12$. Let $\bb = (\ob_1, \dots, \ob_r) \in \{0,1\}^r$. If \NE is $(\epsilon, \delta)$-far from computing a near constant function \bb then there are more than $\frac12 \epsilon r$ output nodes such that for any such output node $i$, the ReLU network \NE restricted to the output node $i$ is $(\frac{\epsilon^2}{1025}, \frac{\delta - e^{-n/16}}{r})$-far from computing near constant function $\ob_i$.
\end{restatable}

One can then combine \cref{thm:ReLU-testing-0-OR-function-2-sided} with \cref{thm:reduction-constant-ReLU-mo} to immediately obtain the following theorem for testing near constant functions in ReLU networks with \emph{multiple outputs and one hidden layer}.

\begin{restatable}{theorem}{TestingConstantFunctionShl}
\label{ReLU-testing-constant-function-2-sided-mo}
Let \NE be a ReLU network $(A,W)$ with $n$ input nodes, $m$ hidden layer nodes, and $r$ output nodes. Let $\bb \in \{0,1\}^r$, $(r+1) e^{-n/16} \le \delta < 1$, and $c \cdot \sqrt{\ln(2m)/m} < \epsilon < \frac12$ for a sufficiently large constant $c$. There is a property tester that queries $O(\frac{\ln^2(1/\epsilon \lambda)}{\epsilon^{13}})$ entries from $A$ and $W$ and
\begin{itemize}
\item accepts with probability at least $\frac23$, if \NE computes near constant function \bb, and
\item rejects with probability at least $\frac23$, if \NE is $(\epsilon,\delta)$-far from computing near constant function \bb.
\end{itemize}
\end{restatable}

Observe that the \emph{query complexity of our tester in \cref{ReLU-testing-constant-function-2-sided-mo} is independent of the network size}. Further, for constant \eps and $\lambda$, the query complexity is $O(1)$, and similarly as in \cref{thm:ReLU-testing-0-OR-function-2-sided}, one could remove the dependency for $\delta$ by introducing a small factor term in the query complexity.

\medskip

Our analysis so far has studied the model of ReLU networks with a single hidden layer. We can extend our analysis of this setting to allow \emph{multiple hidden layers}. A simplified claim can be stated in the following theorem.

\begin{restatable}{theorem}{MultipleLayers}
\label{thm:ReLU-testing-0-function-2-sided-mL-n}
Let $0 < \lambda < \frac{1}{\ell+1}$, $\frac{1}{n} < \epsilon < \frac12$, and $\delta \ge e^{-n/16} + e^{- (\epsilon/2)^{\ell} n /(342 (\ell+1)^2)}$.
Let \NE be a ReLU network $(W_0, \dots, W_{\ell})$ with $n$ input nodes, $\ell \ge 1$ hidden layers with $n$ hidden layer nodes each, and a single output node. Further, suppose that $n \ge 171 \cdot (\ell+1)^2 \cdot (2/\epsilon)^{2\ell} \cdot (\ln(2/\delta) + \ell \ln n)$.
Assuming that~$\ell$ is constant, there is a tester 
that queries $\Theta(\epsilon^{-8\ell} \cdot \ln^4(1/\lambda\epsilon))$ entries from $W_0, \dots, W_{\ell}$, and
%
\begin{enumerate}[(i)]
\item rejects with probability at least $1 - \lambda$, if \NE is $(\epsilon,\delta)$-far from computing the constant $0$-function,
\item accepts with probability at least $1 - \lambda$, if \NE computes the constant $0$-function.
\end{enumerate}
\end{restatable}

The approach used in the proof of \cref{thm:ReLU-testing-0-function-2-sided-mL-n}, on a high level, follows the method presented in the analysis of a single hidden layer (see \cref{thm:ReLU-testing-0-OR-function-2-sided}), but the need of dealing with multiple layers makes the analysis significantly more complex and challenging (and much longer).
The approach relies on two steps of the analysis, first proving part
\begin{inparaenum}[(i)]
\item about ReLU networks that are $(\epsilon,\delta)$-far from computing the constant $0$-function (\cref{subsubsec:output-mL,subsubsec:ReLU-testing-0-function-2-sided-far-mL}) and then part
\item about ReLU networks that compute the constant $0$-function (\cref{subsubsec:ReLU-testing-0-function-2-sided-compute-mL}).
\end{inparaenum}

The analysis of part (i) is in four steps,
    see \cref{sec:ReLU-testing-0-function-2-sided-mL}.
We first show (\cref{subsubsec:sampling-mL}) that the use of randomly sampled nodes in each layer of the network results in a modified ReLU network that, after scaling, for any fixed input returns the value close to the value returned by the original network.
Next, (\cref{subsubsec:sampling-mL-delta}) we extend this claim to show that there is always a modified ReLU network that, after scaling, for all but a small fraction of the inputs, returns the value close to the value returned by the original network.
We then study ReLU networks $(W_0, \dots, W_{\ell})$ that are $(\epsilon, \delta)$-far from computing the constant $0$-function.
In that case, we first show
that there is always an input $x_0 \in \{0,1\}^{m_0}$ on which the $(\epsilon, \delta)$-far network returns a value which is (sufficiently) large.
Then, using the result from \cref{lemma:sampling-mL-existance} that the sampled network approximates the output of the original network, we will argue in \cref{subsubsec:ReLU-testing-0-function-2-sided-far-mL} that if $(W_0, \dots, W_{\ell})$ is $(\epsilon, \delta)$-far from computing the constant $0$-function then randomly sampled network on input $x_0$ returns a (sufficiently) large.

The study of part (ii), as described in \cref{lem:ReLU-testing-0-function-2-sided-compute-mL}, relies on a similar approach as that in part (i). We use the fact that the sampled network computes a function which on a fixed input is (after scaling) close to the original function on that input (\cref{lemma:sampling-mL}), and hence that value cannot be too large. An important difference with part (i) is that one must prove the claim to \emph{hold for all inputs}, making the arguments more challenging (the straightforward arguments would require the size of the sample to depend on $m_0$, giving a tester with super-constant query complexity).

A consequence of \cref{thm:ReLU-testing-0-function-2-sided-mL-n} is that it gives a tester with query complexity $1/\poly(\epsilon)$ for any constant number of layers.
Furthermore, as in our study for one hidden layer, essentially identical analysis as for the constant $0$-function can be applied to test the OR-function (see \cref{thm:ReLU-testing-OR-function-2-sided-mL} and its simpler form \cref{thm:ReLU-testing-OR-function-2-sided-mL-m0=n} in \cref{sec:ReLU-testing-OR-function-2-sided-mL}).

\medskip

Finally, we can extend the analysis above (\cref{ReLU-testing-constant-function-2-sided-mo,thm:ReLU-testing-0-function-2-sided-mL-n}) to networks with \emph{multiple hidden layers} and \emph{multiple outputs}. As in \cref{ReLU-testing-constant-function-2-sided-mo} (for ReLU networks with a single hidden layer and multiple outputs) the proof of the following \cref{ReLU-testing-constant-function-2-sided-mL-mo} relies on a reduction of the problem to testing a near constant function to testing the constant $0$-function and the OR-function.

\begin{restatable}{theorem}{MultipleOutputsLayers}\textbf{\emph{(Testing a near constant function in feedforward networks with ReLU activation function, multiple hidden layers, and multiple outputs)}}
\label{ReLU-testing-constant-function-2-sided-mL-mo}
Let \NE be a ReLU network $(W_0, \dots, W_{\ell})$ with $m_0$ input nodes, $\ell \ge 1$ hidden layers with $m_1, \dots, m_{\ell}$ hidden layer nodes each, and $m_{\ell+1} \ge 2$ output nodes.
Let $0 < \lambda < \frac{1}{\ell+1}$, $2 \sqrt[\ell]{\frac{17(\ell+1)}{\max_{0 \le k \le \ell} \{m_k\}}} < \epsilon < \frac12$, and $\delta \ge 2 m_{\ell+1} (\probp + e^{-m_0/16})$.
Let $\bb \in \{0,1\}^{m_{\ell+1}}$ and $c = c(\ell)$ be a sufficiently large positive constant.
Suppose that for an arbitrary parameter $0 < \probp < \frac{1}{\ell+1}$ it holds that $m_k \ge c \cdot (2/\epsilon)^{2\ell^2} \cdot \ln(\frac{1}{\probp} \prod_{i=1}^{\ell} m_i)$ for every $0 \le k \le \ell$.
If $\ell$ is constant then there is a tester that queries $\Theta(\epsilon^{-(8\ell^2+1)} \cdot \ln^4(1/\lambda \epsilon))$ entries from $W_0, \dots, W_{\ell}$, and
\begin{itemize}
\item accepts with probability at least $1-\lambda$, if \NE computes near constant function \bb, and
\item rejects with probability at least $1-\lambda$, if \NE is $(\epsilon,\delta)$-far from computing near constant function~\bb.
\end{itemize}
\end{restatable}

\noindent
As with our earlier results in \cref{subsection:Overview-multiple-outputs+layers}, the \emph{query complexity of our tester in \cref{ReLU-testing-constant-function-2-sided-mL-mo} is independent of the network size}, and for constant \eps, $\lambda$, $\ell$, and $m_{\ell+1}$, the query complexity is $O(1)$.










%

\section{Organization of the rest of the paper}
\label{subsec:rest-organization}

The rest of the paper contains technical details, the proofs, the algorithms, and the analysis of the results promised in \cref{section:Overview}.

\cref{sec:ReLU-testing-0-OR-function} presents the analysis of testing constant 0-function and OR-function in our model of \emph{ReLU networks with one hidden layer and a single output}. It begins (\cref{sec:ReLU-testing-0-function-NP-completeness}) with the proof of \cref{thm:NP-completeness-of-0-function}, that deciding if the network computes 0 is \textsf{NP}-hard.  Then, we prove the main result here, \cref{thm:ReLU-testing-0-OR-function-2-sided}, that constant $0$-function and the OR-function are testable in our model.

\cref{subsec:comparisons-to-vanilla} contains the analysis of a \emph{vanilla testing} procedure, concluding with a proof of Corollary~\ref{corollary:hardness-vanilla} which implies that with the right sampling approach we can efficiently analyze the behavior of a neural network in some cases where it would not be detected by the vanilla tester.

\cref{sec:distribution-free-lower-bound} studies in details the \emph{distribution-free model} and \emph{its limitations for testing the constant 0-function}. The main result here, \cref{thm:testing-distribution-free-0}, gives a lower bound on testing the 0-function in the distribution-free model. This is arguably one of the most technically involved parts of the paper.

\cref{sec:ReLU-testing-0-OR-function-1-sided} studies the complexity of testing the constant 0-function and OR-function with \emph{one-sided error}. We show that with one-sided error testing is possible with $O(m \ln n/\eps^2)$ queries (\cref{thm:ReLU-testing-0-OR-function-1-sided}), and at the same time, any tester for the constant $0$-function (or for the OR-function) that has one-sided error has query complexity of $\Omega(m)$ (\cref{thm:ReLU-testing-0-OR-function-1-sided-hardness}).

\cref{sec:monotonicity-symmetry} discusses and proves a structural theorem, \cref{theorem:0and1}, that all ReLU networks are close to the constant 0-function or close to the OR-function, and discusses its consequences for property testing (\cref{corollary:monotonicity,corollary:symmetry}).

\cref{sec:monotone-properties} analyzes \emph{testing of monotone properties}, concluding with the proof of \cref{thm:monotone-properties-general}.

The following \cref{sec:multiple-outputs,sec:multiple-layers,sec:ReLU-testing-constant-function-2-sided-mL-mo} extend our analysis from \cref{sec:ReLU-testing-0-OR-function} of testing constant 0-function and OR-function in ReLU networks with one hidden layer and a single output to \emph{ReLU networks with multiple hidden layers and multiple outputs}. This study, together with the analysis in \cref{sec:distribution-free-lower-bound}, is arguably technically the most involved part of our analysis. First, \cref{sec:multiple-outputs} considers ReLU networks with \emph{multiple outputs and one hidden layer}, beginning with a general reduction to single layer networks in \cref{thm:reduction-constant-ReLU-mo} and summarizing the entire analysis in the proof of \cref{ReLU-testing-constant-function-2-sided-mo}. Next, \cref{sec:multiple-layers} considers ReLU networks with \emph{multiple hidden layers and a single output}, summarizing the analysis in the proofs of \cref{thm:ReLU-testing-0-function-2-sided-mL-n} and \cref{thm:ReLU-testing-0-function-2-sided-mL} for the constant 0-function, and of \cref{thm:ReLU-testing-OR-function-2-sided-mL} and \cref{thm:ReLU-testing-OR-function-2-sided-mL-m0=n} for the OR-functions. Finally, \cref{sec:ReLU-testing-constant-function-2-sided-mL-mo} combines the two previous sections and study testing of ReLU networks with \emph{multiple layers and outputs}, see \cref{ReLU-testing-constant-function-2-sided-mL-mo}.

This analysis is then followed by \cref{sec:conclusions} containing final conclusions and thoughts about our framework and our results.

Finally, for the sake of completeness, we recall in \cref{app:concentration-bounds} three auxiliary concentration bounds used in the paper. 


\section{\emph{One hidden layer}: Testing constant 0-function and OR-function}
\label{sec:ReLU-testing-0-OR-function}


We begin our study with the analysis of ReLU networks with one hidden layer and a single output, deferring the study of more complex neural networks with multiple outputs and multiple hidden layers to Sections~\ref{sec:multiple-outputs}--\ref{sec:ReLU-testing-constant-function-2-sided-mL-mo}. While on a high level, the analysis of the simpler one-output-bit and one-hidden-layer ReLU networks shares many similarities with the analysis of the more general case from Sections~\ref{sec:multiple-outputs}--\ref{sec:ReLU-testing-constant-function-2-sided-mL-mo}, the study of ReLU networks with multiple outputs and multiple hidden layers is more refined and subtle. Therefore we begin with the simpler analysis here, providing useful intuitions and tools which will be used later in Sections~\ref{sec:multiple-outputs}--\ref{sec:ReLU-testing-constant-function-2-sided-mL-mo}.

We consider first two simplest functions we can think of, which are
\begin{itemize}
\item the \textbf{constant 0-function} $f^{(0)}_n: \{0,1\}^n \rightarrow \{0,1\}$ with $f^{(0)}_n(x) = 0$ for all $x\in \{0,1\}^n$, and
\item the \textbf{OR-function}\footnote{Observe that the OR-function $f^{(1)}_n$ can be seen as a complement of the constant $0$-function, since it differs from the constant $1$-function (defined as $f_n(x) = 1$ for all $x \in \{0,1\}^n$) only on a single input $x = \bzero$, on which we know (see \cref{def:ReLU}) that any ReLU network will always return 0 (see \cref{remark:constant-functions-and-ReLU} for a more detailed discussion).} $f^{(1)}_n: \{0,1\}^n \rightarrow \{0,1\}$ with $f^{(1)}_n(x) = \sgn(x) = \sgn(\sum_{i=1}^n x_i)$ for all $x = (x_1, \dots, x_n) \in \{0,1\}^n$
\end{itemize}
and study them from a property testing point of view. We define $P_n = \{f^{(0)}_n\}$ and $Q_n = \{f^{(1)}_n\}$, and the properties $\Pi^{(0)} = \bigcup_{n \in \NN} P_n$ and $\Pi^{(1)} = \bigcup_{n \in \NN} Q_n$. Our objective is to design property testers for $\Pi^{(0)}$ and $\Pi^{(1)}$.


\subsection{Computational complexity: Deciding if the network computes 0 is \textsf{NP}-hard}
\label{sec:ReLU-testing-0-function-NP-completeness}

Before we study the testing task, we first show that the problem of deciding whether a given ReLU network computes the constant $0$-function is \textsf{NP}-hard (under Cook reduction). We remark that similar proofs for related setting are known (see, for example, \cite{BDL22,BR88,DWX20,GKMR21,J88,MR18}).

\begin{theorem}
\label{thm:NP-completeness-of-0-function}
Let $(A,w)$ with $a_{ij} \in \mathbb Q$ and $w_j \in \mathbb Q$ be a ReLU network with $n$ input nodes and $m$ hidden layer nodes
that computes function $f(x) = \sgn(\relu(w^T \relu(Ax))$. Then the problem of deciding on input $(n,m,A,w)$ whether the network computes the constant $0$-function is \textsf{NP}-hard.
\end{theorem}

\begin{proof}
In order to show \textsf{NP}-hardness, we reduce PARTITION to our problem. Recall that in the partition problem we are given a set $I$ with $N$ positive integers our objective is to decide, if $I$ can be partitioned into two sets $S,T$ such that $\sum_{y\in S} y = \sum_{y\in T} y$.

We show that given an instance $I$ of PARTITION we can compute in polynomial time a ReLU network with $n= N+1$ input nodes and $3$ hidden layer nodes such that there exists $x\in\{0,1\}^n$ with $\sgn(\relu(w^T \relu(Ax))=1$ if and only if there exists a partition of $I=\{r_1,\dots,r_N\}$ into $S,T$ such that $\sum_{y\in S} y = \sum_{y\in T} y$.

Let $W=\sum_{y\in I} y$. Let us use $u_1, u_2, u_3$ to denote the hidden layer nodes. We set the weight of the edge connecting the $i$th input node to $u_1$ to be $\frac{r_i}{W}$ and to $u_2$ to be $-\frac{r_i}{W}$. Furthermore, we connect the $N+1$st input layer node via an edge of weight $-\frac12$ to node $u_1$, an edge of weight $\frac12$ to node $u_2$ and an edge of weight $1$ to node $u_3$. In the second layer, nodes $u_1$ and $u_2$ are connected with an edge of weight $-1$ to the output, while node $u_3$ is connected through an edge of weight $\frac1W$. All other edges have weight $0$.

Now assume that a partition of $I$ into $S$ and $T$ with $\sum_{y\in S} y = \sum_{y\in T} y$ exists. In this case, for $1\le i \le N$, we can set $x_i=1$ if $r_i\in S$ and $x_i=0$, otherwise. Furthermore, we set $x_{N+1}=1$. Since the edge weights connecting the $i$th node to $u_1$ and $u_2$ are $\frac{r_i}{W}$ and $-\frac{r_i}{W}$ they sum up to $\frac12$ and $-\frac12$. Taking the last input bit into account, the sum of inputs at $u_1$ and $u_2$ is $0$. However, $u_3$ contributes $\frac1W$ to the output and so we have $\sgn(\relu(w^T \relu(Ax))=1$.

Now consider a vector $x\in\{0,1\}^n$ with $\sgn(\relu(w^T \relu(Ax))=1$. Since $u_1$ and $u_2$ cannot contribute positively to the output and since $u_3$ contributes positively only if $x_{N+1}=1$ we know that $x_{N+1}=1$. Now consider the partition $S=\{r_i: x_i=1\}$ and $T= I \setminus S$ and assume that $\sum_{y\in S} y \ne \sum_{y\in T} y$. Since the numbers in $I$ are integers, it follows that $\sum_{y\in S} y \ne \sum_{y\in T} y$ implies that $|\sum_{y\in S} y - \sum_{y\in T} y|\ge 1$. This means that either $u_1$ or $u_2$ contribute at least $-\frac1W$ to the output value and so $\sgn(\relu(w^T \relu(Ax))=0$. A contradiction! Thus, we have $\sum_{y\in S} y = \sum_{y\in T} y$ and hence the reduction ensures that there exists $x\in\{0,1\}^{N+1}$ with $\sgn(\relu(w^T \relu(Ax))=1$ if and only if there is a partition of $I$ into sets of equal sum. Clearly, the reduction can be computed in polynomial time and so the problem is \textsf{NP}-hard.
\end{proof}

It is also not difficult to see that the arguments in the proof of \cref{thm:NP-completeness-of-0-function} show that also the problem of deciding whether the network computes the OR-function is \textsf{NP} hard.
%


\subsection{One hidden layer: Tester for the constant 0-function and the OR-function}
\label{sec:ReLU-testing-0-OR-function-2-sided}

While \cref{thm:NP-completeness-of-0-function} shows that it is hard to decide if a given ReLU network computes the constant $0$-function or not, in this section we prove that a \emph{relaxed} version of the problem, namely, that of \emph{(property) testing} if a given ReLU network computes the constant $0$-function, has significantly lower complexity. Our result can be easily extended to the OR-function.
The following theorem shows that properties $\Pi^{(0)}$ and $\Pi^{(1)}$ are testable for ReLU networks with one hidden layer.

\ZeroORFunction*

\paragraph{Constant 0-function.}
Our property testing algorithm \textsc{AllZeroTester}$(\epsilon,\lambda,n,m)$ below works by focusing its attention on a small ``subnetwork'' induced by some randomly sampled nodes. The algorithm first samples a subset of $\poly(1/\epsilon) \cdot \ln(1/\lambda)$ nodes from the input layer and from the hidden layer. Then we scale the weights of all edges between the sampled input and the sampled hidden layer nodes by a factor of $\frac{n}{s}$ and all edges from sampled hidden layer nodes to the output node by $\frac{m}{t}$. We then introduce a \emph{bias} $b = -\frac{1}{16}\epsilon nm$ to the output node, i.e., we subtract $\frac{1}{16}\epsilon nm$ before we apply the ReLU function. For the resulting network we check whether it computes the constant $0$-function. If this is the case, we accept. Otherwise, we reject.

In the following, we use matrix notation and write $S$ and $T$ for sampling matrices that select $s$ and $t$ input/hidden layer nodes at random. That is, $S$ is an $n \times n$ diagonal matrix whose diagonal entries $(i,i)$
are equal to 1 iff $x_i$ is in the sample set (and is 0 otherwise), and $T$ is an $m \times m$ diagonal matrix whose diagonal entries $(j,j)$ are equal to 1 iff hidden layer node $j$ is in the second sample set (and is 0 otherwise).
With this notation we can describe our algorithm as follows (we did not optimize constants and we assume that $\epsilon\in (\frac1m,\frac12)$ and $\lambda \in (0,\frac12)$).


\begin{algorithm}[htbp]
\SetAlgoLined\DontPrintSemicolon
\caption{\textsc{AllZeroTester}$(\epsilon,\lambda,n,m)$}
\label{alg:AllZeroTester}

Sample $s = \lceil \frac{2^{20}}{\epsilon^2} \cdot \ln(1/\epsilon\lambda) \rceil$ nodes from the input layer and $t = \lceil \frac{2^{30}}{\epsilon^4} \cdot \ln(1/\epsilon\lambda)\rceil$ nodes  from the hidden layer uniformly at random without replacement.

Let $S \in \NN_0^{n\times n}$ and $T \in \NN_0^{m \times m}$ be the corresponding sampling matrices.

\lIf{\emph{there is $x \in\{0,1\}^n$ with $\frac{nm}{st} \cdot  w^T \cdot \relu(TASx) - \frac{1}{16}\epsilon nm > 0$}}{\textbf{reject};}

\lElse{\textbf{accept}.}
\end{algorithm}


Before we will analyze the properties of the output of \textsc{AllZeroTester}, let us first observe that the query complexity of \textsc{AllZeroTester} is $(s+1) \cdot t = \Theta(st)$, and hence the query complexity is as stated. Indeed, after sampling $s$ input nodes and $t$ nodes from the hidden layer, the diagonal matrices $S$ and $T$ have $s$ and $t$ non-zero entries, respectively. Therefore $TAS$ can be computed by querying only $s \cdot t$ entries from matrix $A$ and it will have at most $t$ non-zero rows and at most $s$ non-zero columns. Hence, in order to compute $TASx$ one has to consider only $s$ bits from $\{0,1\}^n$, and to compute $\frac{nm}{st} \cdot  w^T \cdot \relu(TASx)$ one has to consider only $t$ entries from $w^T$.

\paragraph{The OR-function.}
Not surprisingly, the property tester for the OR-function is similar to the property tester for the constant $0$-function, since essentially, the OR-function is the constant $1$-function, except that for a single input: when $x$ is a $0$-vector $\bzero$ then the network returns 0 (cf. \cref{remark:constant-functions-and-ReLU}). As we will see below, a simple modification of the property tester for the constant $0$-function can be turned into a property tester for the OR-function. Our testing algorithm \textsc{ORTester} is almost identical to Algorithm \textsc{AllZeroTester} except that now we use \emph{bias} $b = \frac{1}{16}\epsilon nm$.


\begin{algorithm}[htbp]
\SetAlgoLined\DontPrintSemicolon
\caption{\textsc{ORTester}$(\epsilon,\lambda,n,m)$}
\label{alg:ORTester}

Sample $s = \lceil \frac{2^{20}}{\epsilon^2} \cdot \ln(1/\epsilon\lambda) \rceil$ nodes from the input layer and $t = \lceil \frac{2^{30}}{\epsilon^4} \cdot \ln(1/\epsilon\lambda)\rceil$ nodes  from the hidden layer uniformly at random without replacement.

Let $S \in \NN_0^{n\times n}$ and $T \in \NN_0^{m \times m}$ be the corresponding sampling matrices.

\lIf{\emph{there is $x \in\{0,1\}^n$ with $\frac{nm}{st} \cdot  w^T \cdot \relu(TASx) + \frac{1}{16} \epsilon nm < 0$}}{\textbf{reject};}

\lElse{\textbf{accept}.}
\end{algorithm}


\paragraph{Properties of the output.}
Next, we will study the properties of the output. First, notice that our sampling network without a bias would compute
\begin{align*}
    \sgn\left(\relu\left(\frac{m}{t} \cdot w^T \cdot T \cdot \relu\left(\frac{n}{s} \cdot ASx\right)\right)\right).
\end{align*}
Since the ReLU function is positive homogeneous, $T$ is a fixed matrix with non-negative entries after the sampling, $\frac{n}{s} >0$, and $\frac{m}{t}>0$, we have that
\begin{align*}
    \sgn\left(\relu\left(\frac{m}{t} \cdot  w^T \cdot T \cdot \relu\left(\frac{n}{s} \cdot ASx\right)\right)\right)
    &=
    \sgn\left(\relu\left(\frac{nm}{st} \cdot  w^T \cdot \relu(TASx)\right)\right).
\end{align*}
Introducing in Algorithm \textsc{AllZeroTester} the bias to the output node leads to a network computing
\begin{align*}
    \sgn\left(\relu\left(\frac{nm}{st} \cdot  w^T \cdot \relu(TASx)- \frac{1}{16} \epsilon nm\right)\right)
\end{align*}
and we notice that one can drop $\sgn$ and the outer $\relu$, if we check whether the computed value is greater than $0$.
Similarly, introducing the bias to the output node leads to the property tester Algorithm \textsc{ORTester} corresponding to a network computing
\begin{align*}
    \sgn\left(\relu\left(\frac{nm}{st} \cdot  w^T \cdot \relu(TASx) + \frac{1}{16}\epsilon nm \right)\right)
\end{align*}
and we notice that one can drop $\sgn$ and the outer $\relu$, if we check whether the computed value is smaller than or equal to $0$.


\subsubsection{One hidden layer: If $(\epsilon, \delta)$-far then there is a bad input}
\label{subsubsec:output-1L}

To analyze the correctness of our algorithms \textsc{AllZeroTester} and \textsc{ORTester}, we first present our structural lemma.

\begin{lemma}
\label{lemma:output}
Let $\delta \ge e^{-n/16}$ and $\frac{1}{m} < \epsilon < \frac{1}{2}$. Let $(A,w)$ be a ReLU network with $n$ input nodes and $m$ hidden layer nodes.
\begin{enumerate}[(1)]
\item If $(A,w)$ is $(\epsilon, \delta)$-far from computing the constant $0$-function then there exists an input $x \in \{0,1\}^n$ such that
\begin{align*}
    w^T \cdot \relu(Ax) &> \tfrac18\epsilon mn \enspace.
\end{align*}
\item If $(A,w)$ is $(\epsilon, \delta)$-far from computing the OR-function then there exists an input $x \in \{0,1\}^n$ such that
\begin{align*}
    w^T \cdot \relu(Ax) &< - \tfrac18 \epsilon mn \enspace.
\end{align*}
\end{enumerate}
\end{lemma}

\begin{proof}
The proofs of both parts of \cref{lemma:output} are by contradiction: we will assume that for every input $x$ the output value is at most $\frac18 \epsilon nm$, or at least $-\frac18 \epsilon nm$, respectively, and then we will show that the ReLU network is $(\epsilon, \delta)$-close to computing the constant $0$-function, or the OR-function, respectively; this would imply contradiction. We will begin with the analysis of the constant $0$-function and then show how the analysis should be modified to extend the proof to the OR-function.
\begin{enumerate}[(1)]
\item For the purpose of contradiction, we will assume that for every input $x$ the output value is at most $\frac18 \epsilon nm$, and then we will show that the ReLU network is $(\epsilon, \delta)$-close to computing the constant $0$-function, implying contradiction.

First, suppose that there are less than $\lfloor \epsilon m \rfloor$ positively weighted edges connecting the hidden layer to the output node. Then we set all of them to $0$ and the resulting network computes the constant $0$-function. Thus the network is $(\epsilon,\delta)$-close, contradicting to the lemma's assumption.

Otherwise, suppose that there are at least $\lfloor \epsilon m \rfloor$ positively weighted edges connecting the hidden layer to the output node. Select an arbitrary subset $E$ of $\lfloor \epsilon m \rfloor$ positively weighted edges and let $U$ denote the set of hidden layer nodes incident to the edges from $E$. Let $E'$ be the subset of first layer edges connecting input nodes to the nodes from $U$. Notice that $|E| = |U| = \lfloor \epsilon m \rfloor$ and $|E'| = n \cdot |U| = \lfloor \epsilon m \rfloor n$. Now we modify the weights of $|E| = \lfloor\epsilon m \rfloor$ edges between the hidden layer nodes and the output, and $|E'| = \lfloor \epsilon m \rfloor n$ edges between input nodes and hidden layer nodes: we assign weight $-1$ to every edge from $E$ and weight $1$ to every edge from $E'$.

Observe that originally, before the modifications, every hidden layer node in $U$ had a non-negative contribution to the output. Furthermore, after the modifications, every hidden layer node from $U$ contributes $-k$ to the output node, where $k$ is the number of 1s in the input vector $x$, i.e., $k = \|x\|_1$. Thus all nodes from $U$ contribute $- |U| \cdot \|x\|_1 = - \lfloor \epsilon m \rfloor \cdot \|x\|_1$, and the changes of the weights of the edges incident to $U$ \emph{reduce the output value} by at least $\lfloor \epsilon m \rfloor \cdot \|x\|_1$.

Next, let us consider a random input vector $x \in \{0,1\}^n$.
Using Chernoff bounds (see auxiliary \cref{claim:NumberOfOnes} for details), with probability at least $1-e^{-n/16} \ge 1 - \delta$, we have $\|x\|_1 > \frac14 n$.
Conditioned on this event, after modifying the weights in the network, the sum of inputs to the output node is \emph{reduced} by at least $\frac14 n \cdot \lfloor \epsilon m \rfloor \ge \frac18 \epsilon nm$ for our range of parameters. However, we have assumed that for every input $x$ the output value is at most $\frac18 \epsilon nm$, and thus the network computes the constant $0$-function. Hence, we modified an $\epsilon$-fraction of the network's edge weights and obtained a function that disagrees with the
constant $0$-function only on inputs with $\|x\|_0 \le n/4$, which is at most a $\delta$-fraction of all inputs. This is a contradiction to the assumption that the network is $(\epsilon,\delta)$-far from that function.
Thus, part (1)
follows for $\delta \ge e^{-n/16}$.

\item The analysis for part (2) is almost identical: we will assume, for the purpose of contradiction, that for every input the output value is at least $- \frac18 \epsilon nm$, and then we will show that the ReLU network is $(\epsilon, \delta)$-close to computing the OR-function; this is a contradiction.

First, suppose that there are fewer than $\lfloor \epsilon m\rfloor $ non-positively weighted edges that connect the hidden layer to the output node. We set all of them to 1 making all edges connecting the hidden layer to the output node to have positive values. Next, we take an arbitrary hidden layer node and make
all its edges incident to the input nodes to become 1. Observe that the resulting network computes the OR-function. Thus, the original network is $(\epsilon,\delta)$-close, contradicting to the lemma's assumption.

Otherwise, proceeding as for part (1), suppose that there are at least $\epsilon m$ non-positively weighted edges connecting the hidden layer to the output node. Take an arbitrary subset $E$ of $\lfloor \epsilon m \rfloor$ non-positively weighted edges and let $U$ denote the set of hidden layer nodes incident to the edges from $E$. Let $E'$ be the subset of first layer edges connecting input nodes to the nodes from $U$. Notice that $|E| = |U| = \lfloor \epsilon m \rfloor$ and $|E'| = n \cdot |U| = \lfloor \epsilon m \rfloor \cdot n$. We modify the weights of $|E| = \lfloor \epsilon m \rfloor$ edges between the hidden layer nodes and the output, and $|E'| = \lfloor \epsilon m\rfloor n$ edges between input nodes and hidden layer nodes: we assign weight 1 to every edge from $E \cup E'$.

Arguing as in part (1), one can show that all nodes from $U$ contribute $|U| \cdot \|x\|_1 = \lfloor \epsilon m\rfloor \cdot \|x\|_1$.

Next, let us consider a random input vector $x \in \{0,1\}^n$
and as above, we can argue that with probability at least $1-e^{-n/16} \ge 1 - \delta$, we have $\|x\|_1 > \frac14 n$.
Conditioned on this event, after modifying the weights in the network, the sum of inputs to the output node is increased by more than $\frac14 n \cdot \lfloor \epsilon m\rfloor\ge \frac18 \epsilon nm$. However, we have assumed that for every input $x$ the output value is at least $-\frac18 \epsilon nm$, and thus the network computes the constant $1$-function; since this is for the inputs with $\|x\|_1 > \frac14 n$, the network computes the OR-function. Hence, we modified an $\epsilon$-fraction of the network's edge weights and obtained a function that may disagree with the OR-function only on inputs $x$ with $\|x\|_0 > n/4$, which is at most a $\delta$-fraction of all inputs. Thus the network is $(\epsilon,\delta)$-close to computing the OR function, which is a contradiction. 
Thus, part (2) follows for $\delta \ge e^{-n/16}$.
\end{enumerate}
\end{proof}


\subsubsection{One hidden layer: Random sampling concentration}
\label{subsubsec:sampling-1L}

Our next lemma considers the impact of the random sampling matrices $S$ and $T$ on the output value of a ReLU network for any fixed $x$ (and hence, for $x$ whose existence is argued in \cref{lemma:output}).

\begin{lemma}
\label{lemma:sampling}
Let $w \in [-1,1]^m$ and matrix $A \in [-1,1]^{m \times n}$. Let $S,T$ be sampling matrices corresponding to samples of size $t \ge \frac{2048 \ln(4/\lambda)}{\epsilon^2}$ and $s \ge \frac{2048 \ln(4t/\lambda)}{\epsilon^2}$. Then for any fixed $x \in \{0,1\}^n$ we have
\begin{align*}
\Pr\left[\left|
    \frac{mn}{st} \cdot w^T \cdot \relu(TASx) - w^T \cdot \relu(Ax) \right| >
    \tfrac{1}{16} \epsilon nm \right]
    &\le \lambda \enspace.
\end{align*}
\end{lemma}

\begin{proof}
Fix an arbitrary $x\in\{0,1\}^n$.
We first show that
\begin{align}
\label{first_ineq}
    \Pr\left[
        \left|\frac{m}{t} \cdot w^T \cdot T \cdot \relu(Ax) - w^T \cdot \relu(Ax) \right|
            > \tfrac{1}{32}\epsilon nm \right]
    &\le \tfrac12 \lambda \enspace.
\end{align}
Observe that $w^T \cdot T \cdot \relu(Ax)$ is the sum of $t$ entries in $w \odot  \relu(Ax)$ (here $\odot$ denotes the Hadamard product) sampled uniformly without replacement. Next, we notice that $\Ex[w^T \cdot T \cdot\relu(Ax)] = \frac{t}{m} \cdot w^T \cdot \relu(Ax)$. Rescaling by $\frac1n$ ensures that the entries in $w \odot\relu(Ax)$ are in $[-1,1]$ and so we can apply Hoeffding bounds \cite{Hoeffding63} (we remark that the Hoeffding bound also holds for sampling without replacement, see \cref{lemma:Hoeffding}) to show that
\begin{align*}
    &\Pr\left[
        \left|\frac{1}{n} w^TT\relu(Ax) - \frac{t}{mn} w^T \relu(Ax)\right| > \tfrac{1}{32} \epsilon t
    \right]
    =\\
    &\Pr\left[
        \left|\frac{1}{n} w^TT\relu(Ax) - \Ex\left[\frac1n w^T \cdot T \cdot\relu(Ax)\right]\right| > \tfrac{1}{32} \epsilon t
    \right]
    \le 2 e^{-2\epsilon^2(t/32)^2/(4t)}
    \le \tfrac12 \lambda
\end{align*}
for our choice of $t \ge \frac{2048 \ln(4/\lambda)}{\epsilon^2}$. Multiplying both sides with $\frac{mn}{t}$ proves inequality (\ref{first_ineq}).

In the following, we view the sampling process as first sampling $t$ nodes from the hidden layer and then $s$ nodes from the input layer. Thus, we can condition on $T$ being fixed and assume that the $s$ nodes from the input layer are sampled uniformly at random. We will now prove the following:
\begin{align}
\label{second_ineq}
    \Pr\left[\left|
        \frac{mn}{ts} \cdot w^T \cdot \relu(TASx)- \frac{m}{t}
        w^T \cdot T \cdot \relu(Ax) \right|
        > \tfrac{1}{32}\epsilon nm \; \Big| \; T \right]
    &\le \tfrac12 \lambda \enspace.
\end{align}

Let $a_i$ be the $i$-th row of matrix $A$. The value $a_i \cdot x$ is the value at the $i$-th hidden layer node before applying ReLU. We may view $a_i S x$ as the sum of $s$ random entries from $a_i \odot x$. Since these entries are from $[-1,1]$ and $\Ex[a_iSx] = \frac{s}{n} \cdot a_i\cdot x$, Hoeffding's bound (\cref{lemma:Hoeffding}) implies that
\begin{align*}
    \Pr\left[
        \left|a_i S x - \frac{s}{n} \cdot a_i x \right|
        > \epsilon s/32 \right]
    \le 2e^{-2\epsilon^2(s/32)^2/(4s)}
    &\le \frac{\lambda}{2t},
\end{align*}
where the last inequality follows from our assumption that $s \ge \frac{2024 \ln(4t/\lambda)}{\epsilon^2}$.

Now we observe that for $w_i\in[-1,1]$ we have
\begin{align*}
    \left|w_i \cdot \left(\relu(a_i S x) - \relu\left(\frac{s}{n}a_i x\right) \right)\right|
        & \le
    \left|w_i \cdot \left(a_i S x - \frac{s}{n}a_i x\right)\right|
        \le
    \left|a_i S x - \frac{s}{n}a_i x\right|.
\end{align*}
Following our arguments above, for each of the rows we have the following bound for the contribution of the edge with weight $w_i$ that connects the $i$-th hidden layer node to the output node:
\begin{align*}
    \Pr\left[
        \left|w_i \cdot \left(\relu(a_i S x) - \relu\left(\frac{s}{n} \cdot a_ ix\right)\right)\right|
        > \frac{1}{32}\epsilon s \right]
    &\le 2e^{-2\epsilon^2/(s/32)^2/(4s)}
    \le \frac{\lambda}{2t}.
\end{align*}
Next, we observe that $TA$ equals $A$ in exactly $t$ rows (the sampled rows) and is $0$ otherwise. Taking a union bound over the contribution of these $t$ rows and observing that for the remaining rows the product with $Sx$ is $0$ for every $S$, we obtain
\begin{align*}
    \Pr\left[
        \left|w^T \cdot \relu(TASx) - \frac{s}{n} \cdot w^T \cdot \relu(TAx)\right| > \frac{1}{32} \epsilon st \; \Big| \; T
    \right]
    \le \frac{\lambda}{2}.
\end{align*}
Multiplying with $\frac{mn}{st}$ and observing that $\relu(TAx) = T \cdot \relu(Ax)$ (since $\relu$ is positively homogeneous), we obtain (\ref{second_ineq}):
\begin{align*}
    \Pr\left[
        \left|\frac{mn}{st} \cdot w^T \cdot \relu(TASx) -
            \frac{m}{t} \cdot w^T \cdot T \cdot \relu(Ax)
        \right| > \epsilon mn/32 \; \Big| \; T \right]
    &\le \frac{\lambda}{2}.
\end{align*}
Next observe that inequalities (\ref{first_ineq}) and (\ref{second_ineq}) are simultaneously satisfied with probability at least $1-\lambda$. Thus, the lemma follows from
\begin{align*}
    \lefteqn{
    \left|
        \frac{mn}{st} \cdot w^T \cdot \relu(TASx) - w^T \cdot \relu(Ax) \right| \le} & \\
    &\le
    \left|\frac{mn}{st}\cdot w^T\relu(TASx) - \frac{m}{t} \cdot w^T \cdot T \cdot \relu(Ax)\right| +
        \left|\frac{m}{t} \cdot w^T \cdot T \cdot \relu(Ax) - w^T \cdot \relu(Ax) \right|
    \\
    &\le \frac{1}{16} \epsilon mn
\end{align*}
with probability at least $1-\lambda$.
\end{proof}


\subsubsection{One hidden layer: If $(\epsilon, \delta)$-far then we reject}
\label{subsubsec:ReLU-testing-0-function-2-sided-far-1L}

Our next lemma incorporates \cref{lemma:output,lemma:sampling} to summarize the analysis of our testers for ReLU networks that are $(\epsilon,\delta)$-far from computing the constant $0$-function or the OR-function.

\begin{lemma}
\label{lem:ReLU-testing-0-OR-function-2-sided-far}
Let $\delta \ge e^{-n/16}$, $\frac{1}{m} < \epsilon < \frac{1}{2}$, and $0 < \lambda < \frac{1}{2}$. Let $(A,w)$ be a ReLU network with $n$ input nodes and $m$ hidden layer nodes.
\begin{enumerate}[(1)]
\item If $(A,w)$ is $(\epsilon,\delta)$-far from computing the constant $0$-function then algorithm \textsc{AllZeroTester} rejects with probability at least $1-\lambda$.
\item If $(A,w)$ is $(\epsilon,\delta)$-far from computing the OR-function then algorithm \textsc{ORTester} rejects with probability at least $1-\lambda$.
\end{enumerate}
\end{lemma}

\begin{proof}
We prove only part (1), since the analysis of part (2) is identical.

By part (1) of \cref{lemma:output}, since $(A,w)$ is $(\epsilon,\delta)$-far from computing the constant $0$-function, there exists an $x_0 \in\{0,1\}^n$ such that $w^T \cdot \relu(Ax_0) > \frac18\epsilon mn$. Next, by \cref{lemma:sampling}, with probability at least $1-\lambda$ (over the random choice of $S$ and $T$ in the algorithm) we have
\begin{align*}
    \left|\frac{mn}{st} \cdot w^T \cdot \relu(TASx_0) - w^T \cdot \relu(Ax_0) \right| \le \frac{1}{16}\epsilon nm\enspace.
\end{align*}
This implies that with probability $1-\lambda$ we obtain
\begin{align*}
    \frac{mn}{st} \cdot w^T \cdot \relu(TASx_0) > \frac{1}{16}\epsilon nm.
\end{align*}
Thus, with probability at least $1-\lambda$ we have
\begin{align*}
    \frac{mn}{st} \cdot w^T \cdot \relu(TASx_0) - \frac{1}{16}\epsilon mn > 0.
\end{align*}
and so \textsc{AllZeroTester} rejects with probability at least $1-\lambda$. This concludes part (1).
\junk{
The proof of part (2) mimics the proof of part (1).

By \cref{lemma:output-OR}, since $(A,w)$ is $(\epsilon,\delta)$-far from computing the OR-function, there exists an $x_0 \in\{0,1\}^n$ such that $w^T \cdot \relu(Ax_0) < - \frac14 \epsilon mn$. Next, by \cref{lemma:sampling}, with probability at least $1-\lambda$ (over the random choice of $S$ and $T$ as defined by the algorithm) we have
\begin{align*}
    \left|\frac{mn}{st} \cdot w^T \cdot \relu(TASx_0) - w^T \cdot \relu(Ax_0) \right| \le \frac18\epsilon nm\enspace.
\end{align*}
This implies that with probability $1-\lambda$ we obtain
\begin{align*}
    \frac{mn}{st} \cdot w^T \cdot \relu(TASx_0) < - \frac18\epsilon nm\enspace.
\end{align*}
Thus, with probability at least $1-\lambda$ we have
\begin{align*}
    \frac{mn}{st} \cdot w^T \cdot \relu(TASx_0) + \frac18\epsilon mn < 0\enspace.
\end{align*}
and so \textsc{ORTester} rejects with probability at least $1-\lambda$. This concludes the proof.
}
\end{proof}


\subsubsection{One hidden layer: If network computes 0 then we accept}
\label{subsubsec:ReLU-testing-0-function-2-sided-compute-1L}

Our final lemma of this section summarizes the second part of the analysis of \textsc{AllZeroTester} and of \textsc{ORTester} for ReLU networks and deals with the networks that correctly compute the constant $0$-function and the OR-function, respectively.

\begin{lemma}
\label{lem:ReLU-testing-0-OR-function-2-sided-compute}
Let $(A,w)$ be a ReLU network with $n$ input nodes and $m$ hidden layer nodes.
 Let $S,T$ be sampling matrices corresponding to samples of size $t \ge \frac{512 \ln(2^{s+1}/\lambda)}{\epsilon^2}$ and $s \ge 1$.
\begin{enumerate}[(1)]
\item If $(A,w)$ computes the constant $0$-function, that is, for all $x\in \{0,1\}^n$ we have
\begin{align*}
    \sgn(\relu(w^T \relu(Ax))) &=0,
\end{align*}
then it holds with probability at least $1-\lambda$ (over the random choice of $S$ and $T$ as defined by the algorithm) that for all $x\in\{0,1\}^n$
\begin{align*}
    \frac{nm}{st} \cdot w^T \cdot \relu(TASx) - \frac{1}{16}\epsilon nm < 0
\end{align*}
and so algorithm {\sc AllZeroTester} accepts.
\item If $(A,w)$ computes the OR-function, that is, such that for all $x \in \{0,1\}^n \setminus \bzero$ we have
\begin{align*}
    \sgn(\relu(w^T \relu(Ax))) &= 1,
\end{align*}
then it holds with probability at least $1-\lambda$ (over the random choices of $S, T$) that for all $x\in\{0,1\}^n \setminus \bzero$
\begin{align*}
    \frac{nm}{st} \cdot w^T \cdot \relu(TASx) + \frac{1}{16} \epsilon nm & > 0
\end{align*}
and so algorithm \textsc{ORTester} accepts.
\end{enumerate}
\end{lemma}

\begin{proof}
We prove only part (1), since the analysis of part (2) is identical.

Let $(A,w)$ be a ReLU network that computes the constant $0$-function. Let $I_S$ be the set of indices of the $s$ nodes sampled by our algorithm from the input layer and let $X_S = \{x = (x_1,\dots,x_n)^T \in\{0,1\}^n: x_i=0, \forall i\notin S\}$ be the set of different inputs to the sampled network. In what follows, we will condition on the choice of $S$ being fixed. Observe that
\begin{align*}
    \frac{nm}{st} \cdot w^T \cdot \relu(TASx) - \frac{1}{16} \epsilon nm < 0
\end{align*}
holds for all $x\in \{0,1\}^n$ if and only if this inequality holds for all $x\in X_S$. Therefore, in what follows we will show that the output value of a fixed $x\in X_S$ will be approximated with high probability and then we will apply a union bound over all $x\in X_S$ to conclude the analysis.

Let $x$ be an arbitrary fixed vector from $X_S$. Observe that $\frac{1}{s} \cdot w^T \cdot \relu(ASx)$ is an $m$-di\-men\-sional vector with entries from $[-1,1]$. We know that $w^T \relu(ASx)\le 0$ since the network computes the constant $0$-function. Since ReLU is positively homogeneous we can move $T$ out of the ReLU function. The fact that $T$ is a sampling matrix and that $w^T \relu(ASx)\le 0$ yields
\begin{align}
\label{ineq:expectation}
    \Ex\left[\frac{1}{s} \cdot w^T \cdot T \cdot \relu(ASx) \;\Big| \; S \right] \le 0.
\end{align}
Then, applying Hoeffding bound (\cref{lemma:Hoeffding}) gives us
\begin{align*}
    \Pr\left[
        \left|
            \frac{1}{s} \cdot w^T \cdot T \cdot \relu(ASx) -
            \Ex\left[\frac{1}{s} \cdot w^T \cdot T \cdot \relu(ASx)\right] \right|
        \ge
        \frac{1}{16} \epsilon t \; \Big| \; S\right] \le 2e^{-2 (\epsilon t/16)^2/(4t)},
\end{align*}
which is upper bounded by $\frac{\lambda}{2^s}$ for our choice of $t$, since $t \ge \frac{512 \ln(2^{s+1}/\lambda)}{\epsilon^2}$. By the union bound over the $|X_S|=2^s$ vectors in $X_S$, this implies the following inequality to hold for all $x \in \{0,1\}^n$:
\begin{align*}
    \Pr\left[
        \left| \frac{mn}{st} \cdot w^T \cdot T \cdot \relu(ASx) -
            \Ex\left[\frac{nm}{st} \cdot w^T \cdot T \cdot \relu(ASx)\right] \right|
        \ge \frac{1}{16} \epsilon nm
    \right] \le \lambda.
\end{align*}
Now the lemma follows from inequality (\ref{ineq:expectation}) and the fact that in the algorithm we are subtracting $\frac{1}{16} \epsilon nm$ from the result by taking a union bound over the $|X_S|=2^s$ vectors in $X_S$.

The proof of part (2) mimics the proof of part (1), completing the proof of \cref{lem:ReLU-testing-0-OR-function-2-sided-compute}.
\end{proof}


\subsubsection{One hidden layer: Completing the proof of \cref{thm:ReLU-testing-0-OR-function-2-sided}}
\label{subsubsec:proving-thm:ReLU-testing-0-function-2-sided-1L}

Now \cref{thm:ReLU-testing-0-OR-function-2-sided} follows immediately from \cref{lem:ReLU-testing-0-OR-function-2-sided-far,lem:ReLU-testing-0-OR-function-2-sided-compute} and from our analysis of the query complexity after the description of \textsc{AllZeroTester} on page \pageref{alg:AllZeroTester} and of \textsc{ORTester} on \pageref{alg:ORTester}.
\qed

\begin{remark}\rm
\label{remark:constraints:ReLU-testing-0-OR-function-2-sided}
In our analysis, we have three constraints for the values of $s$ and $t$: in \cref{lemma:sampling} we require $t \ge \frac{2048 \ln(4/\lambda)}{\epsilon^2}$ and $s \ge \frac{2048 \ln(4t/\lambda)}{\epsilon^2}$, and in \cref{lem:ReLU-testing-0-OR-function-2-sided-compute} we require $t \ge \frac{512 \ln(2^{s+1}/\lambda)}{\epsilon^2}$. One can easily verify that with our initial assumption that $0 < \lambda, \epsilon \le \frac12$, if $t = \lceil\frac{2^{30} \ln(1/\epsilon\lambda)}{\epsilon^4}\rceil$ and $s = \lceil\frac{2^{20} \ln(1/\epsilon\lambda)}{\epsilon^2}\rceil$, then these inequalities are satisfied.
\end{remark}

Our tester may be viewed as the structural result that any network that is $(\epsilon,\delta)$-far from computing the $0$-function (or from computing the OR-function) is far already on a small random subnetwork.


\section{A \emph{vanilla testing} procedure: Why random sampling is slow}
\label{subsec:comparisons-to-vanilla}

The routinely used and (arguably) most natural approach to test the functionality of a neural network is the following \emph{vanilla testing procedure}: one samples from an input distribution and then, after evaluating the network on the input samples, one determines the functionality of the neural network. This approach is equivalent to a different than our property testing model for such networks: instead of providing the oracle allowing to query the entries (weights) of the network, as we do in \cref{def:ReLU-testability}, the vanilla model allows to query the value of the network on any single input vector $x \in \{0,1\}^n$. Observe that this vanilla testing procedure is very standard for the property testing --- it corresponds to the main setting considered routinely in the property testing literature for testing functions, see, e.g., \cite{Goldreich17,BY22} and the references therein.

In this section, we compare our property tester from \cref{thm:ReLU-testing-0-OR-function-2-sided} to the vanilla testing approach and give an example that in some settings our approach will be able to detect a certain functionality (that is, not computing the 0-function), while the vanilla procedure will badly fail.

We will focus on sampling from the uniform distribution. We show that there are networks that are $(\epsilon,\delta)$-far from computing the constant $0$-function, for some constant $\epsilon>0$ and for $\delta = e^{-\Theta(n)}$, such that the vanilla tester will require $2^{\Omega(n)}$ samples (and so also exponential running time) to find an example on which the network does not evaluate to $0$, while our tester will reject the network in constant time.
In other words, the structural findings implied by our tester will lead in some cases to significantly improved testing compared to the vanilla procedure.


\subsection{The construction of a ReLU network $(A_L,w_L)$}
\label{subsubsec:comparisons-to-vanilla-construction}

We begin with an auxiliary construction of a ReLU network $(A_L,w_L)$ for which we will later show that (\cref{lemma:Vanilla1}) the probability that $(A_L,w_L)$ does not return $0$ is exponentially small (and hence it will not be detected by the vanilla tester) and at the same time that (\cref{lemma:Vanilla2}) $(A_L,w_L)$ is $(\epsilon,\delta)$-far from computing the constant $0$-function (and hence it will be detected by our tester \textsc{AllZeroTester}).


\begin{figure}[th]
\centerline{\includegraphics[width=.99\textwidth]{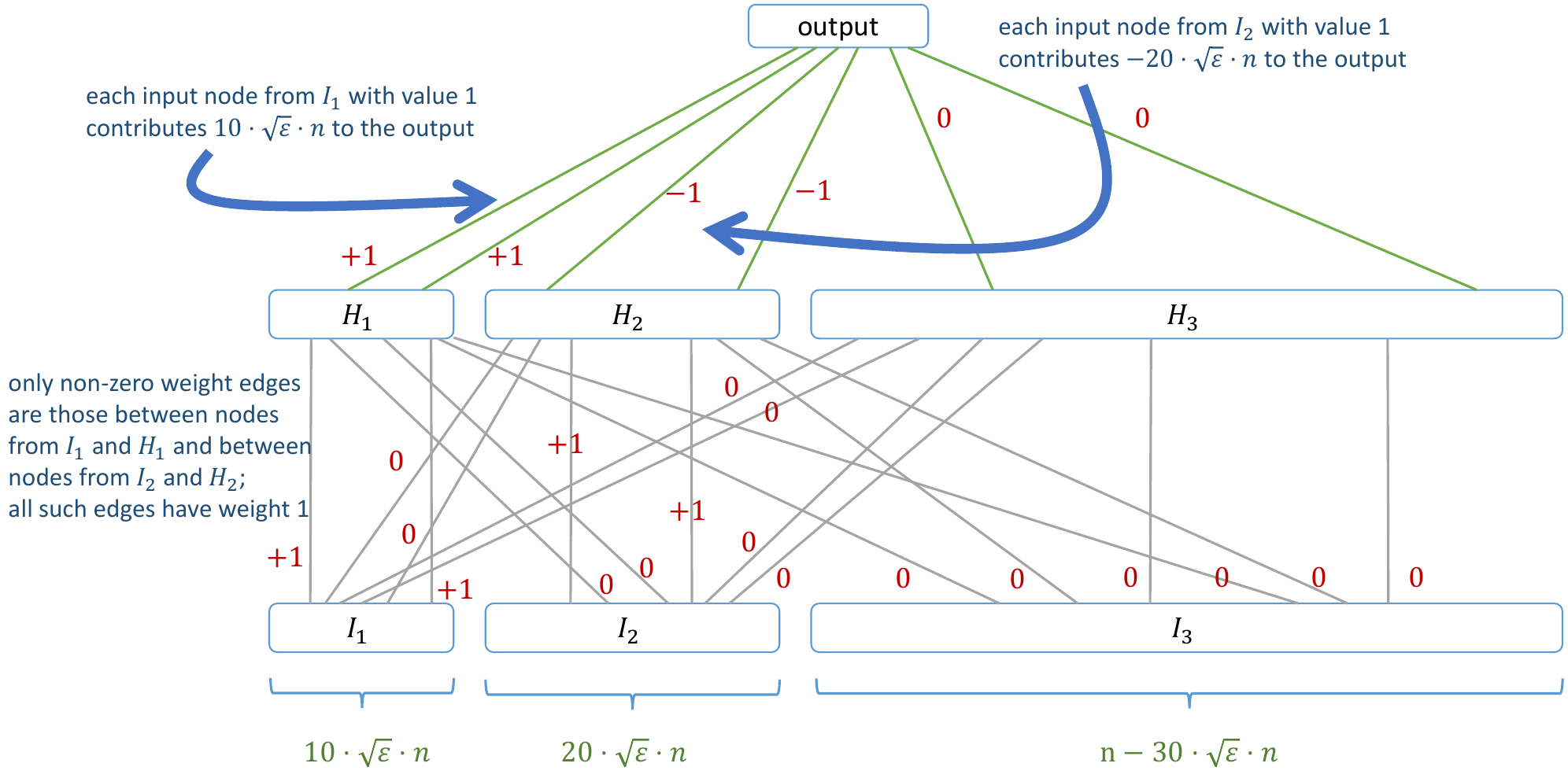}}
\caption{\small{}Construction of the ReLU network $(A_L,w_L)$ used in \cref{corollary:hardness-vanilla}.}
\label{fig:ReLU-network-vanilla}
\end{figure}


Our network is very simple, see Figure \ref{fig:ReLU-network-vanilla}. We have $n$ input nodes and $n$ hidden layer nodes. We partition the $n$ input nodes into three sets $I_1, I_2, I_3$ and the hidden layer nodes into three sets $H_1, H_2$ and $H_3$. $I_1$ contains the first $10 \sqrt{\epsilon} n$ input nodes
(we assume that $\sqrt{\epsilon} n$ is integral - if this is not the case, we can take a smaller $\epsilon$ that satisfies this condition), $I_2$ the next $20 \sqrt{\epsilon} n $ input nodes, and $I_3$ the remaining input nodes. Similarly, $H_1$ contains the first $10 \sqrt{\epsilon} n$ hidden layer nodes, $H_2$ the next $20 \sqrt{\epsilon} n $ hidden layer nodes, and $H_3$ the remaining hidden layer. In the first layer all edges connecting nodes from $I_1$ to $H_1$ and connecting nodes from $I_2$ to $H_2$ are set to $1$; all other edges are set to $0$. In the second layer, the nodes from $H_1$ are connected to the output node with an edge of weight $1$, the nodes from $H_2$ are connected to the output node with an edge of weight $-1$, and the nodes from $H_3$ are connected with an edge of weight $0$. Let $(A_L,w_L)$ be the constructed network.

Observe that each input node from $I_1$ contributes $10 \sqrt{\epsilon} n$ to the output, each input node from $I_2$ contributes $20 \sqrt{\epsilon} n$ to the output, and each input node from $I_3$ contributes $0$ to the output. Thus, the network evaluates to non-zero, if the number of ones among the input nodes from $I_1$ is more than twice the number of ones among the nodes from $I_2$. These observations allow us to prove that the constructed ReLU network $(A_L,w_L)$ returns $0$ for all but a tiny fraction of the inputs.

\begin{lemma}
\label{lemma:Vanilla1}
Let $0 < \epsilon < \frac{1}{1000}$ and $\sqrt{\epsilon} n$ be integral. Let $(A_L,w_L)$ be the network with $n$ hidden layer nodes and $n$ input nodes described above. The probability that a vector $x$ chosen uniformly at random from $\{0,1\}^n$ satisfies
\begin{align*}
    w_L^T \cdot \relu(A_L x) &> 0
\end{align*}
is at most $e^{-10 \sqrt{\epsilon} n/8}$.
\end{lemma}

\begin{proof}
Since every input node from $I_1$ that is $1$ contributes $10 \sqrt{\epsilon} n$ to the output node, the overall contribution of these nodes is at most $100 \epsilon n^2$. Every input node from $I_2$ that is $1$ contributes $-20 \sqrt{\epsilon} n$ to the output node. Thus, if  at least $5 \sqrt{\epsilon} n$ nodes from $I_2$ are $1$ then the output will be $0$. Let $X_i$ be the indicator random variable for the event that the $i$th input bit from $I_2$ is $1$ when the input is chosen according to the uniform distribution. For $X = \sum_{i=1}^{20 \sqrt{\epsilon} n} X_i$ we have $\Ex[X] = 10 \sqrt{\epsilon} n$. By Chernoff bounds we get
\begin{align*}
    \Pr[X \le 5 \sqrt{\epsilon} n] &= \Pr[X \le \tfrac12 \Ex[X]]
        \le
    e^{-E[X]/8} = e^{-10 \sqrt{\epsilon} n/8}
    \enspace.
    \qedhere
\end{align*}
\end{proof}

Our next lemma shows that the constructed network is $(\epsilon,\delta)$-far from computing $0$-function.

\begin{lemma}
\label{lemma:Vanilla2}
Let $n \in \NN$, $0 < \epsilon < \frac{1}{1000}$, and $\delta < 2^{-(16 \sqrt{\epsilon} n+1)}$. The ReLU network $(A_L,w_L)$ is $(\epsilon,\delta)$-far from computing the constant $0$-function.
\end{lemma}

\begin{proof}
Let us first observe that any modification of a first layer weight can change the output value by at most $2$ and any modification of a second layer weight can change the output value by at most $2n$. Therefore, any modifications of at most $\epsilon n^2$ first layer weights and of at most $\epsilon n$ second layer weights can change any output value by at most $4 \epsilon n^2$.

Next, notice that for any input that has at least $5 \sqrt{\epsilon} n$ ones among the input nodes in $I_1$ and at most $ \sqrt{\epsilon} n$ ones among the input nodes in $I_2$, the output value is at least $5 \sqrt{\epsilon} n \cdot |H_1| - \sqrt{\epsilon} n \cdot |H_2| = 30 \epsilon n^2$. Thus, the network will output $1$ on such inputs even after an $\epsilon$-fraction of modifications. The probability that a random input has at least $5 \sqrt{\epsilon} n$ ones among the input nodes in $I_1$ is at least $\frac12$. The probability that the input has at most $ \sqrt{\epsilon} n$ ones among the $20 \sqrt{\epsilon} n$ input nodes in $I_2$ is
\begin{align*}
    \binom{20 \sqrt{\epsilon} n}{\sqrt{\epsilon} n} \cdot \left(\frac12\right)^{20 \sqrt{\epsilon} n}
        &\ge
    \left(\frac{20 \sqrt{\epsilon} n}{\sqrt{\epsilon} n}\right)^{\sqrt{\epsilon} n} \cdot \left(\frac12\right)^{20 \sqrt{\epsilon} n}
        \ge
    2^{-16 \sqrt{\epsilon} n}
    \enspace.
\end{align*}
Therefore, the ReLU network $(A_L,w_L)$ evaluates to $1$ for at least a $2^{-(16 \sqrt{\epsilon} n+1)}$-fraction of the inputs, and so the network is $(\epsilon,\delta)$-far from computing the constant $0$-function for $\delta < 2^{-(16 \sqrt{\epsilon} n+1)}$.
\end{proof}

We can now combine \cref{lemma:Vanilla1,lemma:Vanilla2} to summarize our findings in the following corollary.

\begin{corollary}
\label{corollary:hardness-vanilla}
Let $n$ be the number of input layer (and hidden layer) nodes of the network $(A_L,w_L)$ as defined above. Let $\frac1n < \epsilon < 2^{-16}$ and $0 < \lambda < \frac12$. Then
\begin{itemize}
\item Algorithm {\sc AllZeroTester} rejects ReLU network $(A_L,w_L)$ with probability at least $1-\lambda$;
\item the probability that $\sgn(\relu(w_L^T \cdot \relu(A_L x)))$ evaluates to $1$ is at most $e^{-10 \sqrt{\epsilon} n/8}$.
\end{itemize}
\end{corollary}

\begin{proof}
  It follows from Lemma \ref{lemma:Vanilla2} that for
  $\delta = e^{-n/16} < 2^{-(16\sqrt{\epsilon}n+1)}$
  the network $(A_L,w_L)$ is $(\epsilon,\delta)$-far
  from computing the $0$-function.
  Thus, by Theorem \ref{thm:ReLU-testing-0-OR-function-2-sided}
our algorithm will reject with probability $1-\lambda$. The second item follows from Lemma \ref{lemma:Vanilla1}.
\end{proof}

\cref{corollary:hardness-vanilla} implies that with the right sampling approach
we can efficiently analyze the behavior of a neural network in some cases
where it would not be detected by the vanilla tester.



\section{\emph{Distribution-free} model and its limitations for testing 0-function}
\label{sec:distribution-free-lower-bound}

Our analysis so far has been assuming (see \cref{def:ReLU-farness-from-property}) that the notion of being $(\epsilon,\delta)$-far from a property defines the farness with respect to the second parameter $\delta$ in terms of the \emph{uniform distribution}. We say that a network is far from computing a certain function $f$, if one has to change more than an $\epsilon$-fraction of its description size in order to obtain a function that on an input chosen from the uniform distribution differs from $f$ with probability at most $\delta$. While the use of the uniform distribution is standard in the classical property testing with respect to the parameter $\epsilon$ (cf. \cite{BY22,Goldreich17}), this model may seem to be too restrictive for many applications, and in particular, in the study of neural networks. In such case one could prefer to consider the \emph{distribution-free model}, that is, the setting in which one assumes an arbitrary underlying distribution of the data and the distance with respect to the parameter $\delta$ is measured according to this distribution and where the analysis can be applied regardless of the true distribution of the data (see, e.g., \cite{GGR98}, Chapter~1.5.2 in \cite{BY22}, or Chapter~1.3.2 in \cite{Goldreich17}).

In this section, our aim is to extend our study to the \emph{distribution-free setting}, where an input can be generated from an arbitrary unknown distribution \DD. Since our goal is to develop testers whose complexity is sublinear in the number of input bits, we do not want to assume that one can sample and access an entire input $x \in \{0,1\}^n$ according to distribution \DD at once, but rather we want to be able to sample individual bits of $x$. In other words, we allow the tester to have sample access to the underlying distribution \DD as well as query access to every sample.

\paragraph{Distribution-free model.}
Our distribution-free model can be described as follows. We are given an unknown distribution \DD over $\{0,1\}^n$ to which we have query access through a function \textsf{query}$(i,j)$. When for a given $i$ the first \textsf{query}$(i,j)$ happens, an input $y_i \in \{0,1\}^n$ is sampled from \DD and \textsf{query}$(i,j)$ returns the $j$-th bit of $y_i$. For any further queries with the same index $i$, \textsf{query}$(i,j')$ returns the $j'$-th bit of $y_i$ (and no sampling of $y_i$ occurs since $y_i$ is already determined).

\medskip

We believe that the above model is the most natural extension of our framework to the distribution-free setting, if we aim for sublinear time and query testers. Unfortunately, as we will show below, testing in this distribution-free model is hard even for very simple functions. In the following, we will provide in \cref{thm:testing-distribution-free-0} an $\Omega(n^{1-1/k})$ lower bound on the query complexity of testing the constant $0$-function in the distribution-free setting, where $k$ is an arbitrary constant.

We start by defining the notion of $(\epsilon,\delta)$-far for the distribution-free model.

\begin{definition}[\textbf{Being far from computing a function under distribution \DD}]
Let $(A,w)$ be a ReLU network with $n$ input nodes and $m$ hidden layer nodes and let \DD be a distribution on $\{0,1\}^n$. $(A,w)$ is called \textbf{$(\epsilon,\delta)$-close to computing a function} $f:\{0,1\}^n \rightarrow \{0,1\}$ \textbf{under distribution \DD}, if one can change the matrix $A$ in at most $\epsilon nm$ places and the weight vector $w$ in at most $\epsilon m$ places  to obtain a ReLU network that computes a function $g$ such that $\Pr[g(x)\not= f(x))]\le \delta$, where $x$ is chosen at random according to distribution \DD.

If $(A,w)$ is not $(\epsilon,\delta)$-close to computing $f$ under distribution \DD we say that it is \textbf{$(\epsilon,\delta)$-far from computing $f$ under distribution \DD}.
\end{definition}

The definition extends in a straightforward way to properties of functions (cf. \cref{def:ReLU-farness-from-property}).


\subsection{The framework of proving hardness of distribution-free testing of 0-function}
\label{subsec:framework}

The main result of this section, \cref{thm:testing-distribution-free-0} proven in \cref{subsec:testing-distribution-free-0}, provides an $\Omega(n^{1-1/k})$ lower bound on the query complexity of testing the constant $0$-function in the distribution-free setting. Before we present the proof, let us first introduce the framework behind our lower bound.

Let $k\ge 2$ be an even constant that is not divisible by 4. We will assume that $n$ is a multiple of $k$. We will first define two distributions $\NE_1$ and $\NE_2$ over pairs $((A,w),\DD)$ of ReLU networks $(A,w)$ with $n$ input nodes an $m=n$ hidden layer nodes, and input distributions \DD over $\{0,1\}^n$.  Then, we will show that for a pair $((A,w),\DD)$ from $\NE_1$, high probability $(A,w)$ computes the $0$-function (not only with respect to $\DD$), while a pair $((A,w),\DD)$ from $\NE_2$ with high probability is $(\epsilon,\delta)$-far from computing the $0$-function under distribution $\DD$. In $\NE_1$, the distribution \DD will depend on a partition $M$ of the input nodes into subsets of size $k$ and which will be selected uniformly at random from the set of all such partitions. In $\NE_2$, both $(A,w)$ and \DD will depend on such a partition $M$.

We will then show that no algorithm with query complexity $o(n^{1-1/k})$ can decide with probability at least $\frac23$ whether a sample $((A,w),\DD)$ comes from $\NE_1$ or $\NE_2$, where the algorithm is given query access to the sampled network $(A,w)$ and to the sampled distribution $\DD$. However, since a pair $((A,w),\DD)$ from distribution $\NE_1$ computes with high probability the $0$-function and the one from distribution $\NE_2$ is $(\epsilon,\delta)$-far from computing the $0$-function with respect to $\DD$, any property testing algorithm with success probability, say, $\frac34$ has to distinguish between the two distributions (note that we can always boost the success probability of a property tester to any constant without blowing up the asymptotic query complexity).


\begin{figure}[th]
\centerline{\includegraphics[width=.99\textwidth]{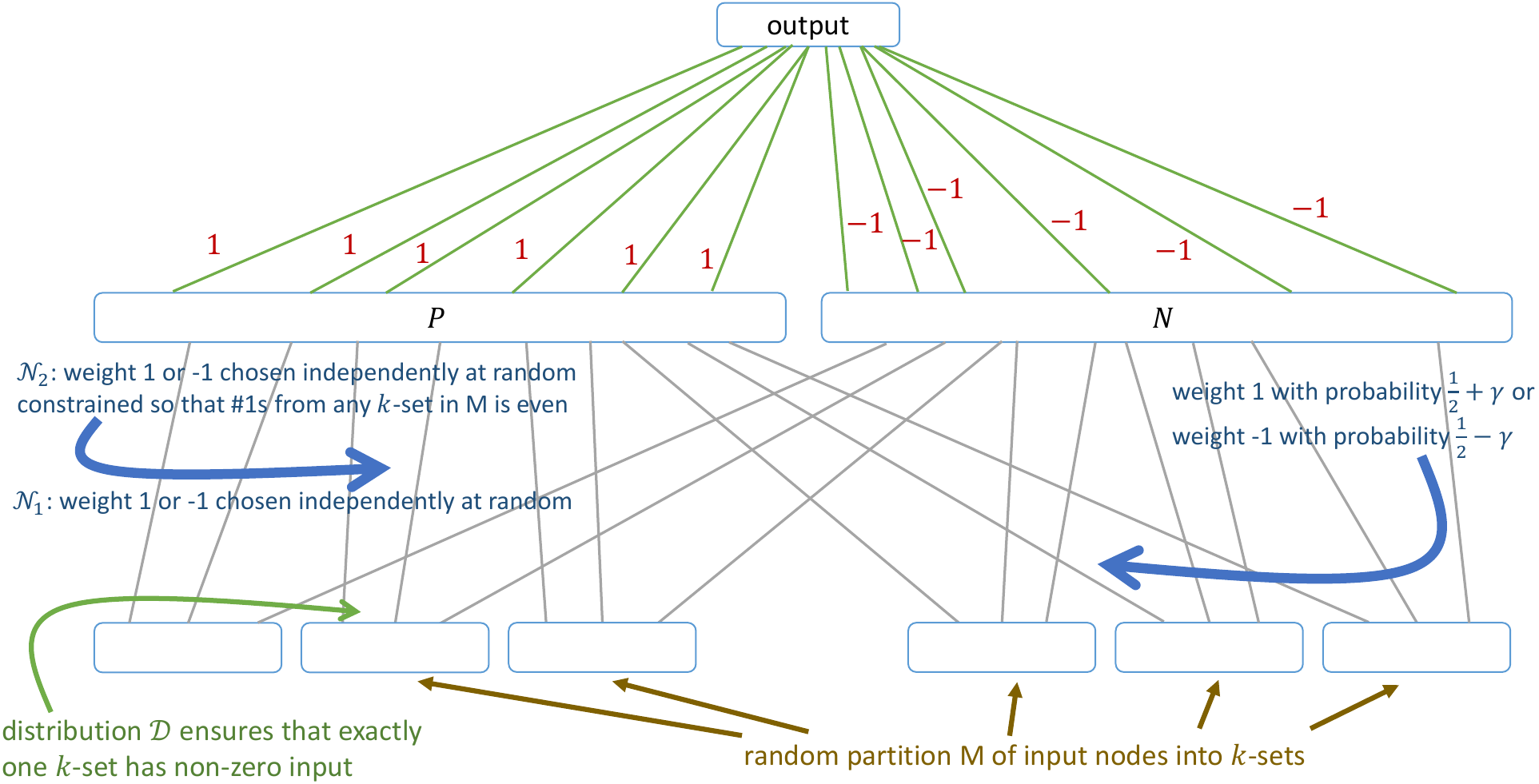}}
\caption{\small{}Construction of the networks used by the distributions $\NE_1$ and $\NE_2$. Notice that the $\frac{n}{2}$ nodes in $P$ have positive contribution to the output while the $\frac{n}{2}$ nodes in $N$ have negative contribution.}
\label{fig:ReLU-network-distribution-free}
\end{figure}


In the following, we describe how both distribution $\NE_1$ and $\NE_2$ generate pairs $((A,w),\DD)$ of a network $(A,w)$ and an input distribution $\DD$. For both $\NE_1$ and $\NE_2$, we first compute a random partition $M$ of the input nodes into subsets of size $k$. Each $k$-set in the partition will be denoted as $\{i_1,i_2,\dots,i_k\} \in M$, where $i_1, \dots, i_k$ are the indices of the input nodes in a given partition set in $M$. For both $\NE_1$ and $\NE_2$ the distribution \DD will be the uniform distribution on the set $\{x \in \{0,1\}^n: \exists \{i_1, \dots, i_k\} \in M \text{ such that } x = \sum_{j=1}^k e_{i_j} \}$. That is, distribution \DD assigns identical probability to $\frac{n}{k}$ inputs $x \in \{0,1\}^n$, each input containing $k$ 1s, as in one set of the partition~$M$.

Next, we describe how the ReLU networks $(A,w)$ are constructed (see also Figure~\ref{fig:ReLU-network-distribution-free}). Let $P$ denote the first $\frac{n}{2}$ hidden layer nodes and $N$ be the remaining $\frac{n}{2}$ hidden layer nodes. For both distributions $\NE_1$ and $\NE_1$ and $\NE_2$, we assign weight $1$ to the edges connecting the nodes from $P$ to the output node and weight $-1$ to connecting the nodes from $N$ to the output node.

In both distributions $\NE_1$ and $\NE_2$, each first layer edge incident to $N$ is assigned independently weight $1$ with probability $\frac12 + \gamma$ and weight $-1$ with probability $\frac12 - \gamma$, where $\gamma$ is a positive constant (that depends on an even $k \ge 2$ that is not divisible by $4$) defined as $\gamma := \frac{\Xi_k}{8k}$ with $\Xi_k := \frac{(-1)^{k/2-1}}{2^{k-1}} \cdot \binom{k-2}{k/2-1}$. Distributions $\NE_1$ and $\NE_2$ differ only in the way they assign weights to the first layer edges incident to $P$. In distribution $\NE_1$, for edges incident to $P$, we assign a weight from $\{-1,1\}$ chosen independently and uniformly at random. We will show that with high probability such a network will compute the $0$-function. The intuitive reason is that  each input node has in expectation more edges with weight $1$ to the nodes in $N$ than to the nodes in $P$, which means that the overall expected contribution per input node to the output will be negative. We will then show that this does not only hold in expectation but also with high probability (the reason is that we sample independently and for sufficiently large $n$ we have good concentration properties). The main challenge of the analysis following this approach is to control the impact of the ReLU function.

For $\NE_2$, the values assigned to the edges from $P$ depend on the sets in $M$. For every node $v \in P$ and every set $\{i_1, \dots, i_k\} \in M$, we select weights of edges $(i_1,v), \dots, (i_{k-1},v)$ independently and uniformly at random from $\{-1,1\}$. For the final edge $(i_k,v)$, we assign weight $1$ if the number of edges from $(i_1,v), \dots, (i_{k-1},v)$ that have been assigned weight $1$ is odd; otherwise, the weight is $-1$. Thus for example, in the case $k=2$, after assigning a random weight from $\{-1,1\}$ to $(i_1,v)$, we assign the same weight to $(i_2,v)$ as to $(i_1,v)$. Observe a useful property of this distribution 
that for every subset of at most $k-1$ from any set in $M$ the distribution is identical to sampling every weight independently and uniformly at random from $\{-1,1\}$. We will show that such a network  is with high probability $(\epsilon,\delta)$-far from computing the constant $0$-function under distribution \DD.

\paragraph{Building intuitions, case $k=2$.}
Before we go into the details of the proof let us first consider the setting $k=2$.
In this setting $\gamma = \frac{1}{32}$; for larger $k$ we will have much smaller values.
Thus, every edge between any input node and a node from $N$ has a probability $\frac23$ to be $1$ and $\frac13$ to be $-1$.
For samples $((A,w),\DD)$ from $\NE_1$, we observe that the expected number of edges with weight $1$ from a fixed input to the set $N$ is $\frac{17}{32} |N|$, while the number of edges with weight $-1$ is only $\frac{15}{32} |N|$. Thus, in expectation an input node contributes $\frac{1}{32} |N|$ to the values at the hidden layer. For the nodes in $P$, the expected number of edges with weight $1$ is the same as the number of edges with weight $-1$, and it is $\frac12 |P|$. Thus, the expected contribution to the hidden layer nodes is $0$. However, due to the ReLU, the expected contribution of the hidden layer nodes from $P$ to the output node is still positive. Nevertheless, the overall contribution will be significantly lower than to the nodes in $N$. Since the nodes in $N$ contribute negatively to the output node and the nodes in $P$ contribute positively (and both sets have the same size), the value at the output node will be negative in expectation (and also with high probability) and so the network  - with sufficiently high probability - outputs $0$ for all inputs, i.e. it computes the $0$-function. This can be proven by a properly parameterized union bound. For the samples from $\NE_1$ the distribution $\DD$ is irrelevant; it is only important to notice that the distribution of $\DD$ is the same as under $\NE_2$.

Next, consider a sample $((A,w), \DD)$ from distribution $\NE_2$. Recall that we first compute a random matching of the input nodes and the only change to the sampling from $\NE_1$ is that we assign identical (random) value to the edges $(u,w)$ and $(v,w)$
that connect matched input nodes $u$ and $v$ to the same hidden layer node $w$. This has two main reasons. Firstly, this will increase the expected contribution to the edges of $P$. This may seem counter-intuitive at first glance, but here the ReLU operation is important. Consider the contribution of two matched input nodes. If both nodes are $1$ (which happens with probability $\frac12)$, then the contribution to the output node is $2$, and if they are both $-1$, then the contribution is $0$. Thus the expected contribution is $1$. Now, when both edge weights are sampled independently, then the only value combination that contributes to the output node is when both edges are $1$, which happens with probability $\frac14$ (and contributes $2$). Thus, in this case the expected contribution is $\frac12$. In other words, when an input vector that is $1$ for two input nodes that belong to the same pair from $M$, then the value at the output node is very large and it requires a lot of changes to the network to fix this. We therefore pick distribution $\DD$ to be the uniform distribution on such inputs.

Next, we will show that it is hard to distinguish between a sample from $\NE_1$ and from $\NE_2$. We consider an arbitrary algorithm that determines whether a sample $((A,w),\DD)$ comes from $\NE_1$ or $\NE_2$. We assume that the algorithm has query access to $(A,w)$ and $\DD$. In order to prove that an algorithm with query complexity $o(\sqrt{n})$ (for $k=2$; larger $k$ we get $o(n^{1-1/k})$ cannot succeed we define two random processes that answer queries to assumed $(A,w)$ and $\DD$ on the fly rather than sampling them at the beginning. The two processes are giving identically distributed answers until for the first time all $k$ elements from one of the sets in $M$ have been queried. In the latter case, the answer is deterministic when the sample is supposed to come from $\NE_2$. Now, the key observation is that finding all elements of such a set requires $\Omega(n^{1-1/k})$ queries. Indeed, we may assume that the input nodes have been randomly perturbed, such that a new query to an input node is uniformly distributed (under this random permutation) among the unqueried nodes. Thus, we may assume that we are taking a uniform random sample from the input nodes. By a birthday-paradox type argument one can prove that we need $\Omega(n^{1-1/k})$ samples to have one set of a partition into sets of size $k$ completely contained in the sample. Thus, we require that many samples to distinguish between $\NE_1$ and $\NE_2$. Since the networks in $\NE_1$ compute the $0$-function (with high probability) and the networks in $\NE_2$ are $(\epsilon,\delta)$-far from computing the $0$-function under $\DD$, any property testing algorithm for the $0$-function in the distribution-free model must have a query complexity of $\Omega(n^{1-1/k})$.


\subsection{Analyzing distribution $\NE_1$}
\label{subsec:distribution-N1}

We begin with showing that if we sample $((A,w),\DD)$ from $\NE_1$ then with high probability $(A,w)$ computes the constant $0$-function. (The analysis in \cref{subsec:distribution-N1} works for any constant $0 < \gamma < \frac12$.)

We will study \emph{the expected value of a hidden layer node} (both from $P$ and $N$) \emph{for a fixed input vector $x$ with $\|x\|_1 = \ell$}. For this purpose, we introduce some random variables. Let $X_1, \dots, X_\ell$ be sampled independently and uniformly at random from $\{-1,1\}$ and let $X^{(\ell)} = \relu(\sum_{i=1}^\ell X_i)$. Let $Y_i$ be independent random variables with $\Pr[Y_i=1] = \frac12 + \gamma$ and $Y_i=-1$, otherwise; let $Y^{(\ell)} = \relu(\sum_{i=1}^\ell Y_i)$. Observe that our definition of the distribution $\NE_1$ implies that $X^{(\ell)}$ is the random variable for the value computed by a single node from $P$ and $Y^{(\ell)}$ is the random variable for the value computed by a single node from $N$. We first bound the difference between the expected value of a node from $P$ and from $N$ (\emph{notice that the use of $\relu$ makes the claim non-standard}.)

\begin{lemma}
\label{lemma:Expectation}
Let $\ell \in \NN$. Let $X_1, \dots, X_{\ell}$ be random variables taking a value from $\{-1,1\}$ independently and uniformly at random. Let $Y_1, \dots, Y_{\ell}$ be independent random variables such that $\Pr[Y_i=1] = \frac12 + \gamma$ and $\Pr[Y_i=-1] = \frac12 - \gamma$
for some constant $\gamma \in(0,1/2)$. Let $X^{(\ell)} = \relu(\sum_{i=1}^\ell X_i)$ and $Y^{(\ell)} = \relu(\sum_{i=1}^\ell Y_i)$. Then
\begin{align*}
    \tfrac14 \cdot \gamma \cdot \ell
        & \le
    \Ex[Y^{(\ell)}] - \Ex[X^{(\ell)}]
        \le
    4 \cdot \gamma \cdot \ell
    \enspace.
\end{align*}
\end{lemma}

\begin{proof}
We verify for $\ell = 1$ that
$\Ex[\relu(Y_1)] = \frac12 + \gamma$ and $\Ex[\relu(X_1)] = \frac12$, and therefore $\frac14 \gamma \le \Ex[Y^{(1)}] - \Ex[X^{(1)}] = \gamma \le 4 \gamma$.
In the following we will focus on the case $\ell \ge 2$.
In order to analyze the difference between the expected values of $Y^{(\ell)}$ and $X^{(\ell)}$, we sample the $Y_i$ according to a different process, which allows us to relate the expected values of $Y^{(\ell)}$ and $X^{(\ell)}$. Let us define an independent random variable $Z_i$ to be $2$ with probability $2\gamma$ and $0$ otherwise. We then define $Y_i= 1$, if $X_i=1$, and $Y_i = X_i + Z_i$, if $X_i=-1$. Observe that $Y_i$ only takes values $-1$ or $1$ and that we have (as required)
\begin{align}
\label{identity-useful-lemma:Expectation}
    \Pr[Y_i=1]
        &=
    \Pr[X_i=1] + \Pr[Z_i=2 | X_i=-1] \cdot \Pr[X_i=-1]
        =
    \tfrac12 + 2\gamma \cdot \tfrac12
        =
    \tfrac12 + \gamma
    \enspace,
\end{align}
where we use the fact that $X_i$ and $Z_i$ are independent. 
Therefore random variables $Y_1, \dots, Y_{\ell}$ have the right distribution and we defined a correlation between them and random variables $X_1, \dots, X_{\ell}$.

Next, since $Y_i \le X_i + Z_i$ and $Z \ge 0$, we can write
\begin{align*}
    \Ex[Y^{(\ell)}] & =
        \Ex[\relu(\sum_{i=1}^\ell Y_i)] \le
        \Ex[\relu(\sum_{i=1}^\ell (X_i + Z_i))] \le
        \Ex[\relu(\sum_{i=1}^\ell X_i) + \sum_{i=1}^\ell Z_i]
    \\& =
        \Ex[\relu(\sum_{i=1}^\ell X_i)] + \Ex[\sum_{i=1}^\ell Z_i] =
        \Ex[X^{(\ell)}] + 4 \gamma \ell
        \enspace.
\end{align*}
This proves the lower bound in \cref{lemma:Expectation}.

For the other direction, in order to process the ReLU function, let us define the event $\EPS$ that $\sum_{i=1}^{\ell} X_i \ge 0$ and notice that $\relu(\sum_{i=1}^\ell X_i \; | \;\EPS) = \sum_{i=1}^\ell X_i$. Next, since our process ensures that $Y_i \ge X_i$ for every $1 \le i \le \ell$, we obtain $\relu(\sum_{i=1}^\ell Y_i \; | \;\EPS) = \sum_{i=1}^\ell Y_i$. Next, notice the following:
\begin{align*}
    \Ex[Y^{(\ell)}]
        &=
    \Ex[\relu(\sum_{i=1}^\ell Y_i)]
        =
    \Ex[\relu(\sum_{i=1}^\ell Y_i) \; | \;\EPS] \cdot \Pr[\EPS] + \Ex[\relu(\sum_{i=1}^\ell Y_i) \; | \; \neg \EPS] \cdot \Pr[\neg \EPS]
        \\
        & \ge
    \Ex[\relu(\sum_{i=1}^\ell Y_i) \; | \;\EPS] \cdot \Pr[\EPS]
        =
    \Ex[\sum_{i=1}^\ell Y_i \; | \;\EPS] \cdot \Pr[\EPS]
        =
    \left(\sum_{i=1}^\ell \Ex[Y_i \; | \;\EPS]\right) \cdot \Pr[\EPS]
    \enspace.
\end{align*}
Similarly, since $\Ex[\relu(\sum_{i=1}^\ell X_i) \; | \;\EPS] = \Ex[\sum_{i=1}^\ell X_i \; | \;\EPS]$ and $\Ex[\relu(\sum_{i=1}^\ell X_i) \; | \; \neg \EPS] = 0$, we have
\begin{align*}
    \Ex[X^{(\ell)}]
        &=
    \Ex[\relu(\sum_{i=1}^\ell X_i)]
        =
    \Ex[\relu(\sum_{i=1}^\ell X_i) \; | \;\EPS] \cdot \Pr[\EPS] + \Ex[\relu(\sum_{i=1}^\ell X_i) \; | \; \neg \EPS] \cdot \Pr[\neg \EPS]
        \\
        & =
    \Ex[\relu(\sum_{i=1}^\ell X_i) \; | \;\EPS] \cdot \Pr[\EPS]
        =
    \Ex[\sum_{i=1}^\ell X_i \; | \;\EPS] \cdot \Pr[\EPS]
        =
    \sum_{i=1}^\ell \left(\Ex[X_i \; | \;\EPS]\right) \cdot \Pr[\EPS]
    \enspace.
\end{align*}
Therefore, we can combine the two bounds above to obtain the following:
\begin{align}
\label{inequality-useful-lemma:Expectation}
    \Ex[Y^{(\ell)}] - \Ex[X^{(\ell)}]
        &\ge
    \left(\sum_{i=1}^\ell \left(\Ex[Y_i \; | \;\EPS] - \Ex[X_i \; | \;\EPS]\right)\right) \cdot \Pr[\EPS]
    \enspace.
\end{align}
Conditioning on the event $\EPS$ can be modeled by identifying the possible outcomes for $X_1,\dots, X_\ell$ with the set of vectors $\alpha = (\alpha_1, \dots, \alpha_{\ell}) \in \{-1,1\}^{\ell}$ such that $\sum_{i=1}^\ell \alpha_i \ge 0$. Let $S^{\ge 0}$ denote this set. Thus, conditioning on $\EPS$ may be viewed as sampling uniformly at random a vector $\alpha = (\alpha_1, \dots, \alpha_{\ell})$ from $S^{\ge 0}$ and assigning $X_i = \alpha_i$ for every $1 \le i \le \ell$. Observe that $2^{\ell-1} \le |S^{\ge 0}| \le 2^{\ell}$. Furthermore,
notice that the number of vectors in $S^{\ge 0}$ with $\alpha_1 = -1$ is at least $2^{\ell-3}$.
(This follows from the fact that the number of vectors with $\alpha_1 = -1$, $\alpha_2 = 1$, and $\sum_{i=3}^\ell \alpha_i \ge 0$ is at least $2^{\ell-3}$.)
By symmetry, for any fixed $1 \le i \le \ell$, the number of vectors $\alpha \in S^{\ge 0}$ with $\alpha_i = -1$ is also at least $2^{\ell-3}$. This implies that $\Pr[X_i = -1 \; | \; \EPS] \ge \frac18$. This in turn, together with the fact that $\Pr[Z_i=2 \; | \; \EPS] = \Pr[Z_i=2] = 2\gamma$ and with the arguments used in (\ref{identity-useful-lemma:Expectation}), implies that
\begin{align*}
    \Pr[Y_i=1 \; | \; \EPS]
        &=
    \Pr[X_i=1 \; | \; \EPS] + \Pr[Z_i=2 | X_i=-1 \wedge \EPS] \cdot \Pr[X_i=-1 \; | \; \EPS]
        \\&=
    \Pr[X_i=1 \; | \; \EPS] + \Pr[Z_i=2] \cdot \Pr[X_i=-1 \; | \; \EPS]
        \\&\ge
    \Pr[X_i=1 \; | \; \EPS] + 2\gamma \cdot \tfrac18
    \enspace.
\end{align*}
This immediately yields
\begin{align*}
    \Pr[Y_i=-1 \; | \; \EPS]
        &=
    1 - \Pr[Y_i=1 \; | \; \EPS]
        \le
    1 - \left(\Pr[X_i=1 \; | \; \EPS] + \tfrac{\gamma}{4}\right)
        =
    \Pr[X_i=-1 \; | \; \EPS] - \tfrac{\gamma}{4}
    \enspace.
\end{align*}
Therefore,
\begin{align*}
    \Ex[Y_i \; | \; \EPS]
        &=
    \Pr[Y_i=1 \; | \; \EPS] - \Pr[Y_i=-1 \; | \; \EPS]
        \ge
    \left(\Pr[X_i=1 \; | \; \EPS] + \tfrac{\gamma}{4}\right) - \left(\Pr[X_i=-1 \; | \; \EPS] - \tfrac{\gamma}{4}\right)
        \\&=
    \Ex[X_i \; | \; \EPS] + \tfrac{\gamma}{2}
    \enspace.
\end{align*}
Now, we can plug the bound above in (\ref{inequality-useful-lemma:Expectation}) and combine it with $\Pr[\EPS] \ge \frac12$ to conclude
\begin{align*}
    \Ex[Y^{(\ell)}] - \Ex[X^{(\ell)}]
        &\ge
    \left(\sum_{i=1}^\ell \left(\Ex[Y_i \; | \;\EPS] - \Ex[X_i \; | \;\EPS]\right)\right) \cdot \Pr[\EPS]
        \ge
    \ell \cdot \tfrac{\gamma}{2} \cdot \tfrac12
    \enspace.
\end{align*}
which yields the first inequality and completes the proof of \cref{lemma:Expectation}.
\end{proof}

The following is a direct corollary of \cref{lemma:Expectation}.

\begin{lemma}
\label{lemma:Expectation-output}
For any input vector $x$ with $\|x\|_1 = \ell$, the expected value of $w^T \cdot \relu(Ax)$ satisfies
\begin{align*}
    - 2 \ell \gamma n
        & \le
    \Ex[w^T \cdot \relu(Ax)]
        \le
    - \tfrac18 \ell \gamma n
    \enspace.
\end{align*}
\end{lemma}

\begin{proof}
Since $\Ex[X^{(\ell)}]$ is the expected value computed by a single node from $P$, $\Ex[Y^{(\ell)}]$ is the expected value computed by a single node from $N$, all nodes from $P$ are connected to the output node by edges with weight $1$, all nodes from $N$ are connected to the output node by edges with weight $-1$, and $|P| = |N| = \frac{n}{2}$, the expected value of $w^T \cdot \relu(Ax)$ is equal to $\frac{n}{2} \cdot \left(\Ex[X^{(\ell)}] - \Ex[Y^{(\ell)}]\right)$. Then the claim follows directly from \cref{lemma:Expectation}.
\end{proof}

\cref{lemma:Expectation-output} implies that for any non-zero input vector $x$, the function computed according to the distribution $\NE_1$ satisfies $\Ex[w^T \cdot \relu(Ax)] < 0$, and now we extend these arguments to obtain a respective high-probability bound. We begin with a simple concentration bound.

\begin{lemma}
\label{lemma:distribution-free-N1-McDiarmind}
Let $x$ be an input with $\|x\|_1 = \ell$. Let $(A,w)$ be chosen from distribution $\NE_1$ for $n$ sufficiently large (as a function of $\gamma$). Let $X = w^T \cdot \relu(Ax)$ be the random variable for the output value. Then
\begin{align*}
    \Pr\left[X - \Ex[X] \ge \tfrac{\ell \gamma n}{8}\right] & \le n^{-10 \ell}
    \enspace.
\end{align*}
\end{lemma}

\begin{proof}
In distribution $\NE_1$, if $\|x\|_1 = \ell$ then random variable $X = w^T \cdot \relu(Ax)$ depends only on the choices of $\ell n$ values in matrix $A$ (entries $A[i,j]$ with $x_i = 1$). Observe that $\NE_1$ selects the values in $A$ independently for each entry, and that changing an entry in $A$ changes the output value by at most $2$. Thus, by McDiarmid's inequality (one-sided version of \cref{lemma:McDiarmid} with $c_k = 2$ for the $\ell n$ relevant entries in $A$), we get
\begin{align*}
    \Pr\left[X - \Ex[X] \ge \tfrac{\ell \gamma n}{8}\right]
        &\le
    e^{-2(\ell \gamma n/8)^2 / (2^2 \cdot \ell \cdot n)} = e^{- \ell \gamma^2 n / 128}
    \enspace.
\end{align*}
For sufficiently large $n$ such that $\frac{n}{\ln n} \ge \frac{1280}{\gamma^2}$, we get the lemma.
\end{proof}

We can now combine \cref{lemma:Expectation-output} and \cref{lemma:distribution-free-N1-McDiarmind} to obtain the following high-probability bound.

\begin{lemma}
\label{lemma:distribution-free-N1-computes-0-whp}
Let $((A,w),\DD)$ be chosen from distribution $\NE_1$ for $n$ sufficiently large (as a function of $\gamma$). Then $(A,w)$ computes the constant $0$-function with probability at least $1-n^{-8}$.
\end{lemma}

\begin{proof}
We have at most $n^\ell$ choices of $x$ with $\|x\|_1 = \ell$. Since $\Ex[w^T \relu(Ax)] \le - \frac18 \ell \gamma n$ by \cref{lemma:Expectation-output}, a union bound applied to the bound from \cref{lemma:distribution-free-N1-McDiarmind} yields for any fixed value of $\ell>0$
\begin{align*}
    \Pr \left[ \exists x\in\{0,1\}^n : \|x\|_1 = \ell \text{ and } w^T \relu(Ax)> 0 \right] &\le
    n^\ell \cdot n^{-10\ell} =
    n^{-9\ell}
    \enspace.
\end{align*}
Next, we take a union bound over the $n$ choices for $\ell \ge 1$ (for $\ell=0$ the output value is always $0$) and obtain
\begin{align*}
    \Pr\left[\exists x\in\{0,1\}^n :  w^T \relu(Ax)> 0\right]
        & \le
    \sum_{\ell=1}^n \Pr\left[\exists x\in\{0,1\}^n : \|x\|_1 = \ell \text{ and } w^T \relu(Ax)> 0\right]  \\
        & \le
    n \cdot n^{-9}
        =
    n^{-8}
    \enspace.
    \qedhere
\end{align*}
\end{proof}


\subsection{Analyzing distribution $\NE_2$}
\label{subsec:distribution-N2}

We next consider distribution $\NE_2$ and observe that it depends strongly on $k$.
We will first determine the difference of the expected value of the ReLU of sum of $k$ i.i.d. binomial random variables and the sum of $k-1$ such variables plus the ``parity,'' i.e., a value that is $1$ if there is an odd number of ones among the first $k-1$ random variables and is $-1$ otherwise.

\begin{lemma}
\label{lemma:Parity}
Let $k \ge 2$ be even. Let $X_1, \dots, X_k$ be random variables such that
\begin{inparaenum}[(i)]
\item $X_1, \dots, X_{k-1}$ take values from $\{-1,1\}$ independently and uniformly at random and
\item let $X_k$ be $1$ if $|\{i \in \{1, \dots, k-1\} : X_i=1 \}|$ is odd and $X_k=-1$ otherwise.
\end{inparaenum}
Let $U_1, \dots, U_k$ be random variables that take values from $\{-1,1\}$ independently and uniformly at random. Then we have
\begin{align*}
    \Ex\left[\relu\left(\sum_{i=1}^k X_i\right)\right] - \Ex\left[\relu\left(\sum_{i=1}^k U_i\right)\right]
        &=
    \frac{(-1)^{k/2-1}}{2^{k-1}} \cdot \binom{k-2}{k/2-1}
    \enspace.
\end{align*}
\end{lemma}

\begin{proof}
Let $k\ge 2$ be even.
Let $\psi:\{-1,1\}^{k-1} \rightarrow \{-1,1\}$ with
\begin{align*}
    \psi(x_1, \dots, x_{k-1}) &=
    \begin{cases}
        0 & \text{ if } \sum_{i=1}^{k-1} x_i < 0;\\
        1 & \text{ if } \sum_{i=1}^{k-1} x_i \ge 0 \text{ and } |\{1 \le i \le k-1 : x_i=1\}| \text{ is odd;}\\
        -1 & \text{ if } \sum_{i=1}^{k-1} x_i \ge 0 \text{ and } |\{1 \le i \le k-1 : x_i=1\}| \text{ is even.}
    \end{cases}
\end{align*}
The definition of $X_k$ ensures that conditioned on $X_1, \dots, X_{k-1}$, if $\sum_{i=1}^{k-1} X_i < 0$ then $\relu(\sum_{i=1}^k X_i) = 0$, and if $\sum_{i=1}^{k-1} X_i \ge 0$ then $X_k = \psi(X_1, \dots, X_{k-1})$.
We observe that since $k-1$ is odd, we get that $\sum_{i=1}^{k-1} X_i$ cannot be $0$. Thus, $\sum_{i=1}^{k-1} X_i\ge 0$ implies
$\sum_{i=1}^{k-1} X_i\ge 1$.
This gives us the following identity:
\begin{align*}
    \relu\left(\sum_{i=1}^k X_i\right) &=
    \relu\left(\psi(X_1, \dots, X_{k-1}) + \sum_{i=1}^{k-1} X_i\right) =
    \psi(X_1, \dots, X_{k-1}) + \relu\left(\sum_{i=1}^{k-1} X_i\right)
    \enspace.
\end{align*}
Hence, since $X_k$ is determined by $X_1, \dots, X_{k-1}$, by definition of expectation we have for even $k \ge 2$:
\begin{align}
        \notag
    \Ex\left[\relu\left(\sum_{i=1}^k X_i\right)\right]
        & = \!\!\!\!\!
    \sum_{x = (x_1, \dots, x_k) \in \{-1,1\}^{k}} \relu\left(\sum_{i=1}^k x_i\right) \cdot \Pr[X^{(k)} = x]
            \\\notag
        & = \!\!\!\!\!
    \sum_{x \in \{-1,1\}^{k-1}} \left(\psi(x_1, \dots, x_{k-1})+ \relu\left(\sum_{i=1}^{k-1} x_i\right) \right) \cdot \prod_{i=1}^{k-1} \Pr[X_i = x_i]
            \\\notag
        & = \!\!\!\!\!
    \sum_{x \in \{-1,1\}^{k-1}} \!\!\!\!\! \psi(x_1, \dots, x_{k-1}) \cdot \prod_{i=1}^{k-1} \Pr[X_i = x_i] +
            \relu\left(\sum_{i=1}^{k-1} x_i \right) \cdot \prod_{i=1}^{k-1} \Pr[X_i = x_i]
            \\
        \label{identity1:lemma:Parity}
        & = \!\!\!\!\!
    \sum_{x \in \{-1,1\}^{k-1}} \!\!\!\!\! \psi(x_1, \dots, x_{k-1}) \cdot \prod_{i=1}^{k-1} \Pr[X_i = x_i] +
            \Ex\left[\relu\left(\sum_{i=1}^{k-1} X_i\right)\right]
    \enspace.
\end{align}
Next, we will study the two terms in (\ref{identity1:lemma:Parity}).

Since $\psi(x_1, \dots, x_{k-1})$ depends only on the number of ones in $x_1, \dots, x_{k-1}$ and since $\psi(x_1, \dots, x_{k-1})$ is non-zero only if 
set $x_1, \dots, x_{k-1}$ has at least $\frac{k}{2}$ ones, we observe that
\begin{align}
    \notag
    \sum_{x \in \{-1,1\}^{k-1}} \psi(x_1, \dots, x_{k-1}) \cdot \prod_{i=1}^{k-1} \Pr[X_i = x_i]
        & =
    \sum_{i=k/2}^{k-1} (-1)^{i+1} \cdot \binom{k-1}{i} \cdot \left(\frac{1}{2}\right)^i \cdot \left(\frac{1}{2}\right)^{k-1-i} \\
        \notag
        &=
    \frac{1}{2^{k-1}} \cdot \sum_{i=0}^{k/2-1}(-1)^i \cdot \binom{k-1}{i} \\
        \label{identity2:lemma:Parity}
        &=
    \frac{1}{2^{k-1}} \cdot (-1)^{k/2-1} \cdot \binom{k-2}{k/2-1}
    \enspace.
\end{align}
The last identity follows, for example, from identity (5.16) in \cite{GKP94}.

For the study of the second term in (\ref{identity1:lemma:Parity}), let us first observe that our construction ensures that random variables $X_1, \dots, X_{k-1}$ are independent, and therefore
\begin{align}
\label{identity3:lemma:Parity}
    \Ex\left[\relu\left(\sum_{i=1}^{k-1} X_i\right)\right] &=
    \Ex\left[\relu\left(\sum_{i=1}^{k-1} U_i\right)\right]
    \enspace.
\end{align}
Next, since $k$ is even, for any $u = (u_1, \dots, u_k) \in \{-1,1\}^k$ the sum $\sum_{i=1}^{k-1} u_i$ is odd and therefore
\begin{inparaenum}[(i)]
\item if $\sum_{i=1}^{k-1} u_i \le 0$ then $\sum_{i=1}^{k} u_i \le 0$, and
\item if $\sum_{i=1}^{k-1} u_i \ge 0$ then $\sum_{i=1}^{k} u_i \ge 0$.
\end{inparaenum}
(i) implies directly that if $\sum_{i=1}^{k-1} u_i \le 0$ then $\relu(\sum_{i=1}^{k-1} u_i) = \relu(\sum_{i=1}^{k} u_i) = 0$ and (ii) yields
\begin{align*}
    \Ex\left[\relu\left(\sum_{i=1}^k U_i\right) \; \Big| \; \sum_{i=1}^{k-1} U_i \ge 0 \right] &=
    \Ex\left[\relu\left(\sum_{i=1}^k U_i\right) \; \Big| \; \sum_{i=1}^{k-1} U_i \ge 0 \wedge U_k = 1 \right] \cdot \Pr[U_k = 1]
        \\& \quad +
    \Ex\left[\relu\left(\sum_{i=1}^k U_i\right) \; \Big| \; \sum_{i=1}^{k-1} U_i \ge 0 \wedge U_k = -1 \right] \cdot \Pr[U_k = -1]
        \\&=
    \left(\Ex\left[\relu\left(\sum_{i=1}^{k-1} U_i\right) \; \Big| \; \sum_{i=1}^{k-1} U_i \ge 0 \wedge U_k = 1 \right] + 1 \right) \cdot \frac12
        \\& \quad +
    \left(\Ex\left[\relu\left(\sum_{i=1}^{k-1} U_i\right) \; \Big| \; \sum_{i=1}^{k-1} U_i \ge 0 \wedge U_k = -1 \right] - 1 \right) \cdot \frac12
        \\&=
    \Ex\left[\relu\left(\sum_{i=1}^{k-1} U_i\right) \; \Big| \; \sum_{i=1}^{k-1} U_i \ge 0 \right]
    \enspace.
\end{align*}
Therefore, we can summarize the discussion above to conclude that
\begin{itemize}
\item $\Ex\left[\relu\left(\sum_{i=1}^k U_i\right) \; \Big| \; \sum_{i=1}^{k-1} U_i < 0 \right] = \Ex\left[\relu\left(\sum_{i=1}^{k-1} U_i\right) \; \Big| \; \sum_{i=1}^{k-1} U_i < 0 \right]$ and
\item $\Ex\left[\relu\left(\sum_{i=1}^k U_i\right) \; \Big| \; \sum_{i=1}^{k-1} U_i \ge 0 \right] = \Ex\left[\relu\left(\sum_{i=1}^{k-1} U_i\right) \; \Big| \; \sum_{i=1}^{k-1} U_i \ge 0 \right]$.
\end{itemize}
This, when combined with (\ref{identity3:lemma:Parity}), directly implies the following identity
\begin{align*}
    \Ex\left[\relu\left(\sum_{i=1}^{k-1} X_i\right)\right] &=
    \Ex\left[\relu\left(\sum_{i=1}^{k-1} U_i\right)\right] =
    \Ex\left[\relu\left(\sum_{i=1}^k U_i\right)\right]
    \enspace.
\end{align*}

We can now plug this identity together with (\ref{identity2:lemma:Parity}) in (\ref{identity1:lemma:Parity}) to conclude the proof of \cref{lemma:Parity}.
\end{proof}

Let us remind that we defined earlier
$\Xi_k := \frac{(-1)^{k/2-1}}{2^{k-1}} \cdot \binom{k-2}{k/2-1}$ and thus \cref{lemma:Parity} says that $\Ex[\relu(\sum_{i=1}^k X_i)] - \Ex[\relu(\sum_{i=1}^k U_i)] = \Xi_k$ for even $k \ge 2$. Observe that $\Xi_2 = \frac12$, but as $k$ increases, the absolute value of $\Xi_k$ decreases. Further $\Xi_k$ is positive when $k$ is even and not divisible by $4$.
Since we set $\gamma := \frac{\Xi_k}{8k}$, $\gamma$ is also positive when $k$ is even and not divisible by $4$.

Our next lemma shows that in distribution $\NE_2$, the expected value of $w^T \cdot \relu(Ax)$ is large.

\begin{lemma}
\label{lemma:Expectation2}
Let $k\ge 2$ be even and not divisible by $4$. Let $((A,w),\DD)$ be chosen at random from distribution $\NE_2$. Let $x = \sum_k e_{i_j}$ for some set $\{i_1, \dots, i_k\} \in M$. Let $X = w^T \cdot \relu(Ax)$ be the random variable for the value at the output node. Then
\begin{align*}
    \Ex[X]
        &\ge
    2 \gamma k n
    \enspace.
\end{align*}
\end{lemma}

\begin{proof}
Our definition ensures that the input $x$ has exactly $k$ ones. Next, let $\Ex[X^{(k)}]$ denote the expected contribution of a single node from $P$ and $\Ex[Y^{(k)}]$ denote the expected contribution of a single node from $N$. Since all nodes from $P$ are connected to the output node by edges with weight $1$, all nodes from $N$ are connected to the output node by edges with weight $-1$, and $|P| = |N| = \frac{n}{2}$, the expected value of $w^T \cdot \relu(Ax)$ is equal to $\frac{n}{2} \cdot \left(\Ex[X^{(k)}] - \Ex[Y^{(k)}]\right)$.

As in distribution $\NE_1$, $Y^{(k)} = \relu(\sum_{i=1}^k Y_i)$, where $Y_1, \dots, Y_k$ are independent random variables such that $\Pr[Y_i=1] = \frac12 + \gamma$ and $\Pr[Y_i=-1] = \frac12 - \gamma$. Therefore, \cref{lemma:Expectation} yields
\begin{align}
\label{ineq1-lemma:Expectation2}
    \Ex\left[Y^{(k)}\right] - \Ex\left[\relu\left(\sum_{i=1}^k U_i\right)\right] &\le 4 \gamma k \enspace,
\end{align}
where $U_1, \dots, U_k$ are random variables from $\{-1,1\}$ chosen independently and uniformly at random.

Observe that $X^{(k)} = \relu(\sum_{i=1}^k X_i)$, where $X_1, \dots, X_k$ are random variables that take values from $\{-1,1\}$ independently and uniformly at random and $X_k = 1$ if $|\{i \in \{1, \dots, k-1\} : X_i=1 \}|$ is odd and $X_k=-1$ otherwise. Therefore, we can use \cref{lemma:Parity} to obtain
\begin{align}
\label{eq1-lemma:Expectation2}
    \Ex\left[X^{(k)}\right] - \Ex\left[\relu\left(\sum_{i=1}^k U_i\right)\right] &= \Xi_k \enspace.
\end{align}
We can use (\ref{ineq1-lemma:Expectation2}) and (\ref{eq1-lemma:Expectation2}) to obtain
\begin{align*}
    \Ex\left[X^{(k)}\right] - \Ex\left[Y^{(k)}\right]
        &=
    \left(\Ex\left[X^{(k)}\right] - \Ex\left[\relu\left(\sum_{i=1}^k U_i\right)\right]\right) -
    \left(\Ex\left[Y^{(k)}\right] - \Ex\left[\relu\left(\sum_{i=1}^k U_i\right)\right]\right)
        \\&\ge
    \Xi_k - 4 \gamma k
        \enspace,
\end{align*}
which yields the lemma since the expected value of $w^T \cdot \relu(Ax)$ is $\frac{n}{2} \cdot \left(\Ex[X^{(k)}] - \Ex[Y^{(k)}]\right)$ and since we set $\gamma := \frac{\Xi_k}{8k}$.
\end{proof}

\cref{lemma:Expectation2} implies that for any non-zero input vector $x$, the function computed according to the distribution $\NE_2$ satisfies $\Ex[w^T \cdot \relu(Ax)] \ge 2 \gamma k n$, and now we extend these arguments to obtain a respective high-probability bound. We begin with a simple concentration bound.

\begin{lemma}
\label{lemma:N2-1}
Let $k\ge 2$ be even and not divisible by $4$. Let $((A,w),\DD)$ be chosen from distribution $\NE_2$. Let $x = \sum_k e_{i_j}$ for some set $\{i_1,\dots, i_k\} \in M$. Let $X = w^T \cdot \relu(Ax)$ be the random variable for the output value. Then
\begin{align*}
    \Pr\left[X \le \gamma k n \right]
        &\le
    e^{-\gamma^2 n/2}
    \enspace.
\end{align*}
\end{lemma}

\begin{proof}
Let $R_i$ be the random variable for $w_i \cdot \relu(a_i x)$, so that $X=\sum_{i=1}^n R_i$. We observe that the $R_i$ are mutually independent  and that changing one $R_i$ changes the output by at most $2k$.

We then apply McDiarmid's inequality (one-sided version of \cref{lemma:McDiarmid}) to obtain
\begin{align*}
    \Pr\left[X \le \Ex[X] - \gamma k n\right] &\le
    e^{-2 (\gamma k n)^2/((2k)^2n)} =
    e^{-\gamma^2 n/2}
    \enspace,
\end{align*}
which immediately yields the claim when combined with the bound $\Ex[X] \ge 2 \gamma k n$ (\cref{lemma:Expectation2}).
\end{proof}

We can now apply \cref{lemma:N2-1} to obtain the following.

\begin{lemma}
\label{lemma:N2-2}
Let $k$ be an even integer that is not divisible by $4$. Let $0 < \epsilon < \frac18 \gamma$ and $0 < \delta < \frac12$. Let $((A,w),\DD)$ be chosen from distribution $\NE_2$. Then with probability at least $1-n \cdot e^{- \gamma^2 n/2}$ the network $(A,w)$ is $(\epsilon,\delta)$-far under distribution \DD from computing the constant $0$-function.
\end{lemma}

\begin{proof}
Let $M$ be the partition of the input nodes associated with $((A,w),\DD)$. Observe that for every input node $v$ there is exactly one input $x \in \{0,1\}^n$ that with non-zero probability assigns a $1$ to $v$. This is the $x$ that is $1$ on $k$ input nodes: $v = i_1$ and some other input nodes $i_2, \dots, i_k$, where $\{i_1, \dots, i_k\} \in M$ is the unique subset that contains $v$.

Let $\mathbb{I}$ be the set of $x \in \{0,1\}^n$ that have non-zero probability under \DD; $|\mathbb{I}| = \frac{n}{k}$. Observe that by \cref{lemma:N2-1}, with probability at most $n \cdot e^{-\gamma^2 n/2}$, there is an $x \in \mathbb{I}$ on which the resulting network computes a value at most $\gamma k n$. In the following, we will assume that this is not the case, that is, that for every $x \in \mathbb{I}$ the output value of the network with input $x$ is greater than $\gamma k n$. We will show that then the network is $(\epsilon,\delta)$-far from computing the constant $0$-function.

Observe that any change of an edge weight in the second layer may reduce the output value by at most $2k$. Thus, if we modify weights of any $\epsilon n$ edges in the second layer, then can reduce the output value by at most $2 k \epsilon n$. Observe that if $\epsilon < \frac14 \gamma$, then $2 k \epsilon n < \frac12 \gamma k n$.

Now, consider an arbitrary $x \in \mathbb{I}$. Since on input $x$ the value of the output node in the network is greater than $\gamma k n$, and after changing at most $\epsilon n$ edge weights in the second layer (which reduced the value by $\frac12 \gamma k n$, and hence the value is more than $\frac12 \gamma k n$), in order to make the network compute $0$ on this input $x$, we have to change more than $\frac14 \gamma k n$ edges in the first layer.

Since these changes are disjoint for different $x, x' \in \mathbb{I}$ and since $|\mathbb{I}| = \frac{n}{k}$ and we have the uniform distribution on $\mathbb{I}$, we have to change more than $\frac{\gamma n^2}{8}$ edge weights in order to make $(A,w)$ \ $(\epsilon,\frac12)$-close to computing the constant $0$-function. Therefore, for $\epsilon <\frac18 \gamma$ and $\delta < \frac12$ the network is not $(\epsilon,\delta)$-close to computing the $0$-function.
\end{proof}


\subsection{Distinguishing between $\NE_1$ and $\NE_2$ requires $\Omega(n^{1-\frac{1}{k}})$ queries}
\label{subsec:distinguishing-distribution-free}

In this section, we study the complexity of distinguishing between distributions $\NE_1$ and $\NE_2$.
We will need the following claim shows that random variables $X_1, \dots, X_k$ are $(k-1)$-wise independent.

\begin{claim}
\label{claim:k-1-wise-independence}
Let $k \ge 2$ be even. Let $X_1, \dots, X_k$ be random variables such that $X_1, \dots, X_{k-1}$ take values from $\{-1,1\}$ independently and uniformly at random and let $X_k$ be $1$ if $|\{i \in \{1, \dots, k-1\} : X_i=1 \}|$ is odd and $X_k=-1$ otherwise. Then the random variables $X_1, \dots, X_k$ are $(k-1)$-wise independent.
\end{claim}

\begin{proof}
Let $\mathcal{I}$, $|\mathcal{I}| = k-1$, be a subset of $\{1, \dots, k\}$. We have to prove that for any $y_1, \dots, y_k \in \{-1,1\}$ and for any $\mathcal{I}$, we have
\begin{align}
\label{claim:k-1-wise-independence-def}
    \Pr\left[\bigwedge_{i \in \mathcal{I}} X_i = y_i\right] &= \frac{1}{2^{k-1}}
    \enspace.
\end{align}
Observe that our definition ensures that (\ref{claim:k-1-wise-independence-def}) holds for $\mathcal{I} = \{1, \dots, k-1\}$

Now assume that $k \in \mathcal{I}$ and let $\ell$ be the number from $\{1, \dots, k\}$ that is not in $\mathcal{I}$. Furthermore, let $\chi = 1$ if the number of $i \in \mathcal{I} \setminus \{k\}$ with $y_i = 1$ is even and $\chi = -1$ otherwise. Then we have
\begin{align*}
    \Pr\left[\bigwedge_{i \in \mathcal{I}} X_i = y_i \right]
        &=
    \Pr\left[X_k = y_k \; \Big| \; \bigwedge_{i \in \mathcal{I} \setminus \{k\}} X_i = y_i \right] \cdot
            \Pr\left[\bigwedge_{i \in \mathcal{I} \setminus \{k\}} X_i = y_i \right] \\
        & =
    \Pr\left[X_{\ell} = \chi \cdot y_i \; \Big| \; \bigwedge_{i \in \mathcal{I} \setminus \{k\}} X_i = y_i\right] \cdot
            \frac{1}{2^{k-2}} \\
        &=
    \frac{1}{2^{k-1}}
    \enspace.
    \qedhere
\end{align*}
\end{proof}

The above claim implies that if we access an arbitrary subset of size $k-1$ of th random variables $X_1,\dots X_k$ then their distribution is identical to a uniform distribution. This will be used in the proof of the following main lemma.

\begin{lemma}
\label{lemma:hardness}
Let $k$ be a constant that is an even integer and not divisible by $4$. Any algorithm
\begin{itemize}
\item that accepts with probability at least $\frac23$ a pair $((A,w),\DD)$ of a ReLU network $(A,w)$ and a distribution $\DD$ drawn from $\NE_1$ and
\item that rejects a pair $((A,w),\DD)$ from $\NE_2$ with probability at least $\frac23$
\end{itemize}
must perform at least $\frac{1}{10} \cdot n^{1-1/k}$ queries to $A$, $w$, and (samples from) $\DD$.
\end{lemma}

\begin{proof}
We prove the claim for any deterministic algorithm. Then the result follows directly from Yao's principle.

Let $A$ be an arbitrary deterministic algorithm that makes $q \le \frac{1}{10} \cdot n^{1-1/k}$ queries to the entries of $A$ and $w$. Recall that to construct sample from $\NE_1$ or $\NE_2$ we first need to determine a random partition of the input nodes into subsets of size $k$. We will construct this partition $M$ as follows: At the beginning $M$ consists of the sets $\{1, \dots, k\}$, $\{k+1, \dots, 2k\}$, etc. We then apply a random permutation $\pi$ of the inputs. This way, $M$ has the desired property. We assume that the at most $q$ samples from distribution $\DD$ are fixed at the beginning and we also apply $\pi$ when we access them.

We will assume that the algorithm knows the second layer edges.
We will further simplify the task for the algorithm if an edge $(u,v)$ from the first layer is queried. We will then return not only the weight of $(u,v)$ but all weights of edges incident to the input node $u$. Similarly, we will return all entries at position $u$ in the $q$ samples from $\DD$. Thus, the algorithm can be thought of to only query input nodes and we will show that it requires at least $\frac{1}{10} \cdot n^{1-1/k}$ such queries. In order to prove this, we need to combine the following steps:
\begin{itemize}
\item We observe that every input node queried (we will assume, without loss of generality, that no node is queried more than once) is uniformly distributed among the unqueried nodes (since we took a random permutation $\pi$). This is regardless of the queries made and answers received. In particular, if we access distribution $\DD$ then the next query is a random entry of the corresponding sample and follows the same distribution for $\NE_1$ and $\NE_2$.
\item We will generate the answers to the algorithm's queries on the fly. If we sample from distribution $\NE_1$ then, if the edge is incident to $P$ then we give a weight $1$ with probability $\frac12$ and $-1$ with probability $\frac12$, and if the edge is incident to $N$ then we give a weight $1$ with probability $\frac12 + \gamma$ and $-1$ with probability $\frac12-\gamma$.
\item If we sample from distribution $\NE_2$ then we proceed similarly for the first $k-1$ queries to a fixed set in $M$.  Claim \ref{claim:k-1-wise-independence} asserts that the first $k-1$ answers to queries to a fixed set in $m$ are distributed uniformly.
Only the last query is answered deterministically depending on the answers to the $k-1$ earlier queries.
\item Thus, the two processes are identical unless we eventually find $k$ nodes from one of the sets in $M$. Since our algorithm queries the nodes uniformly at random (with respect to our random permutation $\pi$), it suffices to analyze the threshold at which we will typically find/not find such a set. This is the case for our sample size as we will see below.
\end{itemize}
The probability of having a set of size $k$ out of $n$ elements fully contained in a sample of size $q \le \frac{1}{10} \cdot n^{1-1/k}$ is equal to
\begin{align*}
    \frac{\binom{n-k}{q-k}}{\binom{n}{q}} &=
    \frac{q \cdot \dots \cdot (q-k+1)}{n \cdot \dots \cdot (n-k+1)} \le
    \frac{q^k}{(n/2)^k} \le
    \frac{1}{10n}
\end{align*}
for $n \ge 2k$. Thus, the expected number of the $\frac{n}{k}$ sets that are fully contained in a sample of size $q \le \frac{1}{10} \cdot n^{1-1/k}$ is at most $\frac{1}{10}$. This implies that the statistical distance of the outputs generated by the two processes above is at most $\frac{1}{10}$. Hence, the algorithm cannot accept $\NE_1$ with probability $\frac23$ and at the same time reject $\NE_2$ with probability $\frac23$.
\end{proof}


\subsection{A lower bound on testing the 0-function in the distribution-free model}
\label{subsec:testing-distribution-free-0}

We are now ready to conclude our study and to combine the analysis of distributions $\NE_1$ and $\NE_2$ in \cref{subsec:distribution-N1,subsec:distribution-N2} with \cref{lemma:hardness} to prove the following theorem.

\LowerBound*

\begin{proof}
Let $M$ be a partition of the $n$ inputs of a ReLU network into subsets of size $k$ (where we assume that $k$ is even and not divisible by $4$; if this is not the case, we can choose the next bigger feasible value of $k$ instead). By Yao's principle, we know that if any deterministic algorithm that distinguishes between distributions $\NE_1$ and $\NE_2$ as in Lemma \ref{lemma:hardness} has query complexity $\Omega(n^{1-1/k})$ then this is also true for any randomized algorithm.

Now, for the purpose of contradiction, assume that there is a distribution-free property tester for the constant $0$-function with a query complexity $o(n^{1-1/k})$. Using standard amplification techniques we can increase the success probability of the algorithm to, say, $\frac{9}{10}$ while only increasing the query complexity by a constant factor. Next, we observe that with high probability a network from $\NE_1$ computes the constant $0$-function (\cref{lemma:distribution-free-N1-computes-0-whp}) and a network from $\NE_2$ is $(\epsilon,\delta)$-far from computing the constant $0$-function under distribution \DD (\cref{lemma:N2-2}). Let us use $\eta(n)$ to be the maximum of the probabilities that a network from $\NE_1$ does not compute the constant $0$-function and $\NE_2$ does not compute a network that is $\NE_2$ is $(\epsilon,\delta)$-far from computing the constant $0$-function under distribution \DD. Note that $\eta(n)$ goes to $0$ as $n$ goes to infinity.

Thus, the property tester accepts a random network from $\NE_1$ with probability at least $\frac{9}{10} - \eta(n)$ and rejects with probability at least $\frac{9}{10}-\eta(n)$ the ones from $\NE_2$. In particular, for sufficiently large $n$ it accepts networks from $\NE_1$ with probability at least $\frac23$ and rejects those from $\NE_2$ with probability at least $\frac23$. Since the query complexity of the property tester is $o(n^{1-1/k})$ this contradicts Lemma \ref{lemma:hardness} for sufficiently large $n$.
\end{proof}



\section{Testing 0-function and OR-function with \emph{one-sided error}}
\label{sec:ReLU-testing-0-OR-function-1-sided}

\cref{sec:ReLU-testing-0-OR-function-2-sided} and in particular, \cref{thm:ReLU-testing-0-OR-function-2-sided}, demonstrates that testing if a given ReLU network computes the constant $0$-function (or the OR-function) can be done very efficiently, with $\poly(1/\epsilon)$ query complexity (in the model of ReLU networks with one hidden layer and a single output). However, the tester in \cref{thm:ReLU-testing-0-OR-function-2-sided} is two-sided and it can err both when the network computes the constant $0$-function (or the OR-function) and when the network is $(\epsilon,\delta)$-far from computing the constant $0$-function (or the OR-function). It is then only natural to try to design a tester with \emph{one-sided error}, that is, to design a tester that can err when the network is $(\epsilon,\delta)$-far from computing the constant $0$-function (or the OR-function), but which always accepts any network that computes the constant $0$-function (or the OR-function). This section studies the complexity of this task. Our main result is that this task is significantly harder: we show that there is a tester with query complexity $\widetilde{O}(m/\epsilon^2)$ and give a near matching lower bound of $\Omega(m)$.


\subsection{Upper bound for testing 0-function and OR-function with one-sided error}
\label{sec:ReLU-testing-0-OR-function-1-sided-upper}

We first present a one-sided error tester for the constant $0$-function and the OR-function whose query complexity is $\widetilde{O}(m/\epsilon^2)$, that is, it almost linear in the number of hidden layer nodes. Observe that since the ReLU network is of size $\Theta(nm)$, the query complexity of our tester is \emph{sublinear}, but it is significantly worse than the query complexity of our testers with two-sided error, where the query complexity was $O(\ln^2(1/\epsilon) / \epsilon^6)$, which is \emph{independent} of the network size, see \cref{thm:ReLU-testing-0-OR-function-2-sided}.

The idea of the property tester is to randomly sample a subset of $O(\log m/\epsilon^2)$ input nodes and query all edges between the sampled input nodes and all the hidden layer nodes, as well as all edges that connect the hidden layer to the output node. We reject, if the resulting subsampled network evaluates an input to one and accept otherwise. We present the one-sided error algorithms for testing the constant $0$-function and for testing the OR-function below.


\begin{algorithm}[htbp]
\SetAlgoLined\DontPrintSemicolon
\caption{\textsc{OneSidedZeroTester}$(\epsilon,\delta, \lambda, n, m)$}

Sample $s = \lceil \frac{128 \ln(2m/\lambda)}{\epsilon^2} \rceil$ nodes from the input layer uniformly at random without replacement.

Let $S \in \NN_0^{n\times n}$ be the corresponding sampling matrix.

\lIf{\emph{there is $x \in\{0,1\}^n$ with
$w^T \cdot \relu(ASx) > 0$}}{\textbf{reject};}

\lElse{\textbf{accept}.}
\end{algorithm}



\begin{algorithm}[htbp]
\SetAlgoLined\DontPrintSemicolon
\caption{\textsc{OneSidedORTester}$(\epsilon, \lambda, n, m)$}

Sample $s = \lceil \frac{128 \ln (2m/\lambda)}{\epsilon^2} \rceil$ nodes from the input layer uniformly at random without replacement.

Let $S \in \NN_0^{n\times n}$ be the corresponding sampling matrix.

\lIf{\emph{there is $x \in\{0,1\}^n \setminus \bzero$ with
$w^T \cdot \relu(ASx) \le 0$}}{\textbf{reject};}

\lElse{\textbf{accept}.}
\end{algorithm}


\begin{restatable}{theorem}{OneSidedOOR}
\label{thm:ReLU-testing-0-OR-function-1-sided}
Let $(A,w)$ be a ReLU network with $n$ input nodes and $m$ hidden layer nodes. Let $e^{-n/16} \le \delta <1$, $\frac{1}{m} < \epsilon < \frac{1}{2}$ and $0 < \lambda < \frac{1}{2}$.
\begin{enumerate}[(1)]
\item There is a tester that queries $O(\frac{m \cdot\ln(m/\lambda) }{\epsilon^2})$ entries from $A$ and $w$ and always accepts, if the ReLU network $(A,w)$ computes the constant $0$-function, and rejects with probability at least $1-\lambda$, if the ReLU network $(A,w)$ is $(\epsilon,\delta)$-far from computing the constant $0$-function.
\item There is a tester that queries $O( \frac{m \cdot\ln(m/\lambda)}{\epsilon^2})$ entries from $A$ and $w$ and always accepts, if the ReLU network $(A,w)$ computes the OR-function and rejects with probability at least $1-\lambda$, if the ReLU network $(A,w)$ is $(\epsilon,\delta)$-far from computing the OR-function.
\end{enumerate}
\end{restatable}

\begin{proof}
We prove only part (1), since the analysis of part (2) is identical.

The query complexity follows immediately from the description given above: Since $S$ has only $s$ non-zero entries, computing $w^T \cdot \relu(ASx) > 0$ for all $x \in\{0,1\}^n$ can be done by checking only $O(s m)$ entries of $A$ and $w$. Thus, we concentrate our analysis on the inputs on which the network accepts and on which it rejects. We first show that the tester has one-sided error and then prove that it rejects with probability at least $\frac23$ any input that is $(\epsilon,\delta)$-far from the constant $0$-function.

Let us assume that the tester rejects. Then there is an input $x$ for the sampled network on which the network outputs a 1. Now let $\bar{x}$ be the $n$-dimensional vector that agrees with $x$ on the sampled input nodes and is 0 otherwise. Note that the full network on input $\bar{x}$ computes the same output value as the sampled network on input $x$. Therefore, there exists an input on which the network computes a 1, implying that the network \emph{does not} compute the constant $0$-function. Hence, if the ReLU network computes the constant $0$-function then \textsc{OneSidedZeroTester} accepts.

Now let us assume that our network is $(\epsilon,\delta)$-far from computing the constant $0$-function. In this case, by \cref{lemma:output} for $\delta \ge e^{-n/16}$ there is an input $x$ on which the network has a value of more than $\frac18 \epsilon n m$ at the output node. Let $a_i$ denote a row of $A$. Using Hoeffding bounds (see \cref{lemma:Hoeffding}) we can show
\begin{align*}
    \Pr\left[\left|\frac{n}{s} \cdot a_iSx - a_ix\right| \ge \frac18 \epsilon n \right]
        &\le
    2 \cdot e^{-\frac{2s (\epsilon/8)^2}{4}}
        \le
    \frac{\lambda}{m}
\end{align*}
for the choice of $s$ in the algorithm. Taking a union bound over all $m$ hidden layer nodes implies this estimate holds  with probability at least $1-\lambda$ simultaneously for all hidden layer nodes. Thus, the overall error at the output node is less than $\pm \frac18\epsilon nm$ and so the value computed is bigger than $0$ and the tester rejects.
\end{proof}


\subsection{Lower bound for testing constant 0-function with one-sided error}
\label{sec:ReLU-testing-0-OR-function-1-sided-hardness}

Now we show that any tester for the constant $0$-function or for the OR-function that has one-sided error has a query complexity of $\Omega(m)$. We begin with an auxiliary claim.

\begin{lemma}
\label{lem:fixing-m/4-wont-suffice-to-make-it-all-0}
Let $A \in [-1,1]^{m \times n}$  and $w \in [-1,1]^m$ be arbitrary. Then for any choice of at most $\frac14m$ entries from $A$ and $w$, we can modify the remaining entries to obtain a matrix $A' \in [-1,1]^{m \times n}$ and vector $w' \in [-1,1]^m$ such that the ReLU network $(A',w')$ computes the constant $0$-function.
\end{lemma}

\begin{proof}
Fix at most $\frac14 m$ entries from $A$ and $w$ and let $E^*$ denote the corresponding set of edges. Let $E^*_1$ be the set of edges from $E^*$ in the first layer (between the input nodes and the hidden layer) and $E^*_2$ be the set of edges from $E^*$ in the second layer.
We first set the weights of all edges in the second layer that are not in $E^*$ to be $-1$.
Then we consider every input node separately. For every edge from the first layer that is not in $E^*$: if it is connected to a positively weighted edge in the second layer, we change its weight to $-1$; if it is connected to a non-positively weighted edge, then we change its weight to 1.
Observe that every (non-zero) input node contributes at least $\frac12 m$ to hidden layer nodes that have a negatively weighted second layer edge and no incident edge from~$E^*$.

Now consider an arbitrary input $x$. If $x$ is the zero vector then the output value is also $0$. Otherwise, at least one bit is 1 and this node contributes no more than $-\frac12 m$ to the output value. At the same point using the originally fixed entries in $A$ and $w$, we can contribute at most $\frac14 m$ to the output value. Any edge not in $E^*$ does not contribute positively to the output and so the output value is at most $-\frac14 m$. Thus, the modified network computes the constant $0$-function.
\end{proof}

Almost identical arguments can be used to the prove the following.

\begin{lemma}
\label{lem:fixing-m/4-wont-suffice-to-make-it-OR}
Let $A \in [-1,1]^{m \times n}$  and $w \in [-1,1]^m$ be arbitrary. Then for any choice of less than $\frac14 m$ entries from $A$ and $w$, we can modify the remaining entries to obtain a matrix $A' \in [-1,1]^{m \times n}$ and vector $w' \in [-1,1]^m$ such that the ReLU network $(A',w')$ computes the OR-function.
\end{lemma}

\begin{proof}
The proof of \cref{lem:fixing-m/4-wont-suffice-to-make-it-OR} is identical to the proof of \cref{lem:fixing-m/4-wont-suffice-to-make-it-all-0}, except that in the construction, all edges in the second layer that are not in $E^*$ will be set to $1$.
\end{proof}

Using \cref{lem:fixing-m/4-wont-suffice-to-make-it-all-0,lem:fixing-m/4-wont-suffice-to-make-it-OR} we can now prove the following lower bound.

\begin{restatable}{theorem}{OneSidedTwo}
\label{thm:ReLU-testing-0-OR-function-1-sided-hardness}
Any tester for the constant $0$-function (or for the OR-function) that has one-sided error has query complexity of $\Omega(m)$.
\end{restatable}

\begin{proof}
For the constant $0$-function, we observe that after a series of $\frac14 m$ queries and answers the algorithm can only reject if it can guarantee that no input, which agrees with the series of queries and answers, computes a constant $0$-function. However, by \cref{lem:fixing-m/4-wont-suffice-to-make-it-all-0}, every network that has $\frac14 m$ entries fixed can be completed in a way that it computes the constant $0$-function. Thus, a tester that makes at most $\frac14 m$ queries has to always accept. However, there are inputs that are $(\epsilon,\delta)$-far from computing the $0$-function. Thus, the lower bound holds.

Similarly, for the OR-function, after a series of $\frac14 m$ queries and answers the algorithm can only reject if it can guarantee that no input that agrees with the series of queries and answers computes the OR-function. However, by \cref{lem:fixing-m/4-wont-suffice-to-make-it-OR}, every network that has $\frac14 m$ entries fixed can be completed in a way that it computes the OR-function. Thus, a tester that makes at most $\frac14 m$ queries has to always accept. However, there are inputs that are $(\epsilon,\delta)$-far from computing the OR-function. Thus, the lower bound holds.
\end{proof}



\section{A \emph{structural theorem} and its consequences for property testing}

\label{sec:monotonicity-symmetry}


In this section we prove that for $\epsilon \ge 2 \sqrt{\frac{\log(1/\delta)}{n}}$, every network $(A,w)$ is $(\epsilon,\delta)$-close to the $0$-function or to the OR-function (which can be thought of the counterpart of the $0$-function as for input $0$ any network outputs $0$). We remark that \textbf{if two networks are close to computing the $0$-function this does not means that they are $(2\epsilon,2\delta)$-close} as there are many networks computing the $0$-function (see also the comments in footnote~\ref{footnote:not-triangle-inequality}). The main result of this section is the following theorem.

\Closeness*

The above result has some immediate consequences on property testing, namely, every property that contains the constant $0$-function and the OR-function is trivially testable for $\epsilon \ge 2 \sqrt{\frac{\log(1/\delta)}{n}}$,
which implies a property testing for $\epsilon \in (0,1)$ with
query complexity $\poly(1/\epsilon) \cdot \polylog(1/\delta)$, assuming that $n$ and $m$ are polynomially related.
Before we present a proof of \cref{theorem:0and1}, let us discuss two interesting applications showing that monotonicity and symmetry for ReLU networks with $m = \poly(n)$ are trivially testable.

\paragraph{Testing monotone functions.}
Recall that a Boolean function $f$ is \emph{monotone}, if for every $x,y \in \{0,1\}^n$, $x\le y$ implies $f(x) \le f(y)$. Let $\PP = \bigcup_{n \in \NN} \PP_n$ be the set of monotone Boolean functions. We observe that the constant $0$-function as well as the OR-function are monotone. Thus, the $0$-function and the OR function are in $\PP_n$. We get the following corollary.

\begin{corollary}
\label{corollary:monotonicity}
Let $(A,w)$ be a ReLU network with $n=m^{O(1)}$ input nodes and $m=n^{O(1)}$ hidden layer nodes. Then monotonicity can be tested
with a trivial tester with $\emph{\text{poly}}(1/\epsilon) \cdot \emph{\polylog}(1/\delta)$
query complexity.
\end{corollary}

\begin{proof}
When $\epsilon \ge  \max\left\{\frac1m, 2 \sqrt{\frac{\log(2/\delta)}{n}}\right\}$ then we can always accept since \cref{theorem:0and1} implies that every function is close to a monotone function. Otherwise, if $\epsilon < \max\left\{\frac1m, 2 \sqrt{\frac{\log(2/\delta)}{n}}\right\}$, we can query the whole network and check whether the function is monotone. Since $m= n^{O(1)}$ and $n=m^{O(1)}$, the query complexity is $O(nm) \le \left(\frac{\log(1/\delta)}{\epsilon}\right)^{O(1)}$, as promised.
\end{proof}

\paragraph{Testing symmetric functions.}
A Boolean function is \emph{symmetric}, if its value does not depend on the order of its arguments, i.e., it only depends on the number of ones. We observe that the $0$-function and OR-function are symmetric. This yields the following corollary with a similar proof as above.

\begin{corollary}
\label{corollary:symmetry}
Let $(A,w)$ be a ReLU network with $n = m^{O(1)}$ input nodes and $m=n^{O(1)}$ hidden layer nodes. Then symmetry can be tested with a trivial tester with $\emph{\poly}(1/\epsilon) \cdot \emph{\polylog}(1/\delta)$
query complexity.
\end{corollary}

\begin{remark}\rm
We believe that the result in \cref{theorem:0and1} is interesting and highly non-trivial. As mentioned earlier, since our $(\eps,\delta)$-distance does not satisfy the triangle inequality (as it is typical in property testing, see also the comments in footnote~\ref{footnote:not-triangle-inequality}) and as there are many networks computing the $0$-function (or the OR-function), \cref{theorem:0and1} does not imply that any two ReLU networks are $(2\epsilon,2\delta)$-close to each other.

Furthermore, on the basis of the analysis in the proof of \cref{theorem:0and1}, one could, for example, hope that even though that are trivial examples of ReLU networks that are far from the constant $0$-function, or far from the OR-function, most (or many) networks are close to both the constant $0$-function and the OR-functions. For example, one could conjecture that such a claim holds for any ``balanced'' ReLU networks, networks that given a uniformly random input in $\{0,1\}^n$ have at least a 0.01 probability to output $+1$ and at least a 0.01 probability to output $-1$. However, this intuition is incorrect, at least for small $\delta$ (in the range where our tester works). This can be seen on the case of a network where a set $H_L$ of half of the hidden layer nodes has 1-edges to the output nodes and $H_R$ (the other half) has $-1$-edges to the output nodes, and where half of the input nodes have 1-edges to $H_L$ and the other half of the input nodes have 1-edges to $H_R$ (all other edges have 0-weights).
\end{remark}


\subsection{Every ReLU network is close to the 0-function or to the OR-function}
\label{subsec:everything-close}

In this section we prove our structural result, \cref{theorem:0and1}. We first start with a simple concentration bound, which is a direct application of McDiarmid's inequality (see \cref{lemma:McDiarmid}).

\begin{lemma}
\label{lemma:ReLUconcentration}
Let $A\in [-1,1]^{m\times n}$ and $w\in[-1,1]^m$. Let $x\in\{0,1\}^n$ be chosen uniformly at random. Then we have
\begin{align*}
    \Pr\left[\big|w^T\relu(Ax) - \Ex[w^T \cdot \relu(Ax)] \big| \ge \epsilon n m \right]
        &\le
    2 e^{-2\epsilon^2 n}
    \enspace.
\end{align*}
\end{lemma}

\begin{proof}
We observe that changing a single bit of $x$ changes $g(x) := w^T \cdot \relu(Ax)$ by a value in $[-m,m]$. Thus, McDiarmid's inequality (see \cref{lemma:McDiarmid} with $t = \epsilon n m$ and $c_k = m$) yields
\begin{align*}
    \Pr\left[\left| g(x) - \Ex[g(x)]\right| \ge \epsilon nm\right]
        &\le
    2 e^{-\frac{2(\epsilon n m)^2}{n m^2}}
        =
    2e^{-2\epsilon^2 n}
    \enspace,
\end{align*}
which proves the lemma.
\end{proof}

Now we complete the proof \cref{theorem:0and1}, which is an immediate consequence of the following.

\begin{lemma}
\label{lemma:ReLUClose}
Let $(A,w)$ be a ReLU network with $n$ input nodes and $m$ hidden layer nodes. Let $0 < \delta \le \frac12$ and let $\epsilon \ge \max\left\{\frac1m, 2 \sqrt{\frac{\log(2/\delta)}{n}}\right\}$. Let $x\in\{0,1\}^n$ be chosen uniformly at random. Then
\begin{itemize}
\item If $\Ex[w^T \cdot \relu(Ax)]>0$ then $(A,w)$ is $(\epsilon, \delta)$-close to computing the OR-function.
\item If $\Ex[w^T \cdot \relu(Ax)]\le0$ then $(A,w)$ is $(\epsilon, \delta)$-close to computing the constant $0$-function.
\end{itemize}
\end{lemma}

\begin{proof}
Let us first consider the case when $\Ex[w^T \cdot \relu(Ax)]>0$ and modify the ReLU network into one that computes the OR-function for all but at most a $\delta$ fraction of the inputs.\footnote{Let us add, as an intuition, that for all but one input $\bzero$ the OR-function returns value 1 (cf. \cref{remark:constant-functions-and-ReLU}), and hence our task reduces to enforcing that the output of the ReLU network is 1 for all but an $\delta$ fraction of the inputs.}

If there are at most $\epsilon m$ non-positive entries in $w$ then we change all these entries into ones; then all entries in $w$ are positive. Consider an arbitrary hidden layer node $j$ and modify all its edges to the input nodes to have weight 1. Observe that the value of this node will become at least 1 unless $x^T = (0,0,\dots,0)$, and since all entries $w_r$ are positive, the output value will be $\sgn(\relu(w^T \relu(Ax)))=1$ for all input but one (when $x^T = (0,0,\dots,0)$). Therefore, we have edited the ReLU network in at most $n$ edges between the input nodes and the hidden layer nodes, and in at most $\epsilon m$ edges between the hidden layer nodes and the output node, and obtained a network that returns 1 for all inputs except $x^T = (0,0,\dots0)$. Therefore the obtained network is $(\epsilon, 0)$-close to computing the OR-function for any $\epsilon > \frac1m$.

Otherwise, if there are more than $\epsilon m$ non-positive entries in $w$ then we change $\epsilon m$ entries into ones and for the corresponding hidden layer nodes, change the weights of all of their edges to the input nodes to ones. These edges were previously contributing at most 0 to $\Ex[w^T \cdot \relu(Ax)]$ and after the change, every hidden layer node modified in this way contributes in expectation $\frac12 n$ (over the choice of $x$) to the output value. Thus $\Ex[w^T \relu(Ax)]$ increases by at least $\frac12 \epsilon m n$ with respect to the original network, and hence $\Ex[w^T \relu(Ax)] > \frac12 \epsilon m n$ after the modification.

Next, we observe that by \cref{lemma:ReLUconcentration} we obtain the following
\begin{align*}
    \Pr\left[ w^T\cdot\relu(Ax) \le 0 \right] & \le
    \Pr\left[\left| w^T \cdot \relu(Ax) -  \Ex[w^T \cdot \relu(Ax)] \right| > \tfrac12 \epsilon n m \right]
    \le
    2e^{-(\epsilon/2)^2 n}.
\end{align*}
Observe that if $\epsilon \ge 2 \sqrt{\frac{ \log(2/\delta)}{n}}$ then $\Pr\left[ w^T\cdot\relu(Ax) \le 0 \right] \le \delta$, and therefore we have modified the original ReLU network to obtain a network that computes the OR-function on all but at most a $\delta$ fraction of inputs. Thus the original network is $(\epsilon, \delta)$-close to computing the OR-function.

If $\Ex[w^T \cdot \relu(Ax)]\le 0$ we can do a similar approach and change an $\epsilon$-fraction of the hidden layer nodes to contribute $-\frac12 n$ in expectation. This results in $\Ex[w^T \relu(Ax)] \le - \frac12 \epsilon n m$ after the modification.
The concentration bounds follows in the same way as above; the details follow.

If there are at most $\epsilon m$ non-negative entries in $w$ then we change all these entries into $-1$s; then all entries in $w$ are negative. This ensures that with non-negatives values in the hidden layer nodes, we will have $w^T \relu(Ax) \le 0$. Hence, $\sgn(\relu(w^T \relu(Ax))) = 0$, and thus the obtained network is $(\epsilon, 0)$-close to computing the constant $0$-function for any $\epsilon > \frac1m$.

Otherwise, if there are more than $\epsilon m$ non-negative entries in $w$ then we change $\epsilon m$ entries into $-1$s and for the corresponding hidden layer nodes, change the weights of all of their edges to the input nodes to $-1$. These edges were previously contributing at least 0 to $\Ex[w^T \cdot \relu(Ax)]$ and after the change, every hidden layer node modified in this way contributes in expectation $-\frac12 n$ (over the choice of $x$) to the output value. Thus $\Ex[w^T \relu(Ax)]$ decreases by at least $\frac12 \epsilon m n$ with respect to the original network, and hence $\Ex[w^T \relu(Ax)] \le - \frac12 \epsilon m n$ after the modification.

Next, we observe that by \cref{lemma:ReLUconcentration} we obtain the following
\begin{align*}
    \Pr\left[ w^T\cdot\relu(Ax) > 0 \right] & \le
    \Pr\left[\left| w^T \cdot \relu(Ax) -  \Ex[w^T \cdot \relu(Ax)] \right| > \tfrac12 \epsilon n m \right]
    \le
    2e^{-(\epsilon/2)^2 n}.
\end{align*}
Observe that if $\epsilon \ge 2 \sqrt{\frac{ \log(2/\delta)}{n}}$ then $\Pr\left[ w^T\cdot\relu(Ax) > 0 \right] \le \delta$, and therefore we have modified the original ReLU network to obtain a network that computes the constant $0$-function on all but at most a $\delta$ fraction of inputs. Thus the original network is $(\epsilon, \delta)$-close to the constant $0$-function.

\end{proof}



\section{Testing \emph{monotone properties}}
\label{sec:monotone-properties}


In this section we study testability of \emph{monotone properties}. We start with some definitions and a useful characterization of monotone properties. For any pair of Boolean functions $f:\{0,1\}^n \rightarrow \{0,1\}$ and $g:\{0,1\}^n \rightarrow \{0,1\}$ define $(f \vee g) : \{0,1\}^n \rightarrow \{0,1\}$, $(f \vee g)(x) := f(x) \vee g(x)$.

\begin{definition}[\textbf{Monotone properties of functions}]
A property $\PP = \bigcup_{n \in \NN} \PP_n$ of Boolean functions is \textbf{monotone} if for every $n \in \NN$, $f \in \PP_n$ implies that $f \vee g \in \PP_n$ for every Boolean function $g: \{0,1\}^n \rightarrow \{0,1\}$.
\end{definition}

If we think of a Boolean function $f$ given as a truth table, that is, a $2^n$-dimensional vector that specifies $f$, then a property is monotone if $f \in \PP_n$ implies that every vector that can be obtained from the truth table of $f$ by turning any number of zeros into ones is also in the property. We next define two simple structural properties of monotone properties.

\begin{definition}
The \emph{monotone closure of a property \PP} is defined as the intersection $\bigcap_{M:{\PP} \subseteq M} M$ of all monotone properties $M$ with $\PP \subseteq M$.
\end{definition}

\begin{definition}
A property $G$ is called \textbf{generator} of a monotone property \PP if \PP is the monotone closure of $G$ and there is no property $G' \subsetneq G$ such that the monotone closure of $G'$ is \PP.
\end{definition}

We prove the following simple lemma that characterizes generators of monotone properties.

\begin{definition}
A monotone Boolean property $\PP = \bigcup \PP_n$ is \textbf{$f(n)$-generatible} if there exists a generator $G = \bigcup_{n \in \NN} G_n$ with $|G_n| \le f(n)$.
\end{definition}

\begin{lemma}
For every monotone property $\PP = \bigcup_n \PP_n$ there exists a generator $G = \bigcup_n G_n$ of \PP.
\end{lemma}

\begin{proof}
Let $X_n \subset \PP_n$ be the set of functions whose truth table has a minimal number of ones, i.e., the set of functions $f \in \PP_n$ such that every function $f'$ that is obtained from $f$ by turning a $1$ in the truth table into a $0$ is not in $\PP_n$. We define $G_n=X_n$ for every $n \in \NN$.

We first show that the monotone closure of $G$ is \PP. Consider an arbitrary $f \in \PP_n$. Then either $f$ has a minimal number of ones, in which case $f\in G_n$, and so $f$ is also in the monotone closure of $G$, or $f$ does not have a minimal number of ones. In the latter case we can construct a minimal $f' \in \PP_n$ from $f$ by turning a finite number of ones into zeros. By construction $f' \in G_n$ and since the monotone closure of $G$ is the intersection of monotone properties that contain $G$, we have that $f$ is in the monotone closure (since $f$ is in all these properties).

Next we show that there is no property $G' \subsetneq G$ that has \PP as monotone closure. Assume that such a $G'$ exists. Then there must be $n \in \NN$ and $f \in G_n \subseteq \PP_n$ with $f \notin G_n'$. By construction $f$ is minimal, so there is no function $f' \in G$ such that one can obtain $f$ from $f'$ by turning zeros into ones. Since $G'\subseteq G$ it follows that there is not such function in $G'$. Hence, $f$ is not in the monotone closure of $G'$, but $f \in \PP$. A contradiction!
\end{proof}


\subsection{Testing 1-generatible properties}
\label{sec:monotone-properties-1-generatible}

Let us consider a property \PP with a generator $G$ such that $|G_n| = 1$ for all $n \in \NN$. In this case, we also write that \PP is \emph{1-generatible}. We first develop a tester for this case and later in \cref{sec:monotone-properties-general} generalize it to larger generators. We use $g_n$ to denote the unique function in $G_n$.

Let $(A,w)$ be the input network with $n$ input nodes and $m$ hidden layer nodes. The idea behind the property testing algorithm is related to the tester for the constant $0$-function (see \cref{sec:ReLU-testing-0-OR-function-2-sided}). In that tester, if the network is $(\epsilon,\delta)$-far from computing the constant 0 function then there must be an input on which the output value is $\Omega(\epsilon nm)$. We will show that if the network is $(\epsilon,\delta)$-far from computing a function from $\PP_n$ then there must be an $\frac12 \delta$-fraction of the inputs such that $w^T \relu(Ax)< - \frac18 \epsilon nm$ \emph{and} $g_n(x) = 1$. The algorithm samples $O(\frac1{\delta})$ input vectors $x \in \{0.1\}^n$ uniformly at random and evaluates them on the ReLU network by sampling a subset of input nodes and hidden layer nodes. The description of the algorithm is given below.


\begin{algorithm}[htbp]
\SetAlgoLined\DontPrintSemicolon
\caption{\textsc{MonotonePropertyTester}$(\epsilon,\delta,\lambda, n,m, g_n)$}
\label{alg:MonotonePropertyTester}

\lIf{$g_n(\bzero)=1$}{\textbf{reject}.}

Let $r = \frac{2 \ln (2/\lambda)}{\delta}$, $t = \frac{512 \ln(4r/\lambda)}{\epsilon^2}$, and $s= \frac{512 \ln(4tr/\lambda)}{\epsilon^2}$.

Sample $s$ nodes from the input layer and $t$ nodes from the hidden layer uniformly at random without replacement.

Let $S \in \NN_0^{n\times n}$ and $T \in \NN_0^{m \times m}$ be the corresponding sampling matrices.

Sample a set $X\subseteq \{0,1\}^n$ of size $r$ uniformly at random.

\lIf{\emph{there is $x \in X$ with $\frac{nm}{st} \cdot w^T \cdot \relu(TASx) + \frac18 \epsilon nm < 0$ and $g_n(x) = 1$}}{\textbf{reject};}

\lElse{\textbf{accept}.}
\end{algorithm}


In order to analyze the above algorithm, we first observe that every ReLU network (without bias) outputs $0$ if the input is the $\bzero$-vector. Thus there is a trivial tester that always rejects, if $g_n(\bzero)=1$.
We further need the following statement.
\begin{lemma}
\label{lemma:monotoneproperties1}
Let $\PP = \bigcup_n \PP_n$ be a 1-generatible monotone property. Let $G = \bigcup_n G_n$ be a generator of \PP with $|G_n| = 1$ for all $n \in \NN$ and let $G_n = \{g_n\}$ and let $g_n(\bzero) = 0$. Let $(A,w)$ be a ReLU network with $n$ input nodes and $m$ hidden layer nodes that is $(\epsilon,\delta)$-far from computing a function from $\PP_n$ for some $\frac1m < \epsilon < 1$ and $e^{-n/16} < \delta < 1$. Then for an $x \in \{0,1\}^n$ chosen uniformly at random we have
\begin{align*}
    \Pr\left[w^T \relu(Ax) < - \tfrac14\epsilon nm  \text{ and } g(x) = 1 \right] &> \tfrac12\delta
    \enspace.
\end{align*}
\end{lemma}

\begin{proof}
Assume that the network is $(\epsilon,\delta)$-far from computing a function from $\PP_n$ and we have $\Pr\left[w^T \relu(Ax) < - \frac14\epsilon nm  \text{ and } g(x) = 1 \right]\le \frac12 \delta$. We show that in this case we can modify at most $\epsilon m$ entries in $w$ and at most $\epsilon nm$ entries in $A$ such that for all but a $\frac12 \delta$-fraction of the inputs the value $w^T \relu(Ax)$ is increased by at least $\frac14 \epsilon nm$. In order to do so, we proceed in similar way as in \cref{lemma:output}. If there are at most $\epsilon m$ non-positive entries in $w$ then we replace all of them with a 1. Furthermore, we replace the value of all edges connecting a single node to the input nodes by $1$. These are $n \le \epsilon n m$ changes for our range of $\epsilon$. We observe that the resulting network computes the OR-function, which is contained in any monotone property with $g_n(\bzero) = 0$. Thus, there must by more than $\epsilon m$ non-positive entries in $w$. We take $\epsilon m$ of them and change them into 1s. Then we replace all edges connecting them to the input nodes with 1s.

Similarly as in the proof of \cref{lemma:output}, we can apply Chernoff bounds and obtain that with probability at least $1-e^{-n/16}$ the number of ones in the input vector $x$ is at least $\frac14n$.  In this case, the sum of inputs to the output node is increased by at least $\frac14\epsilon nm$. This implies that for the modified network we have
\begin{align*}
    \Pr\left[w^T \relu(Ax) < 0 \text{ and } g(x) = 1 \right] &\le \delta \enspace.
\end{align*}
Thus, the network was not $(\epsilon,\delta)$-far from computing a function from \PP, which is a contradiction.
\end{proof}

With \cref{lemma:monotoneproperties1} at hand, we can prove that our tester works.

\begin{theorem}
\label{thm:monotone-properties-1-generatible}
Let $\frac1m < \epsilon < 1$ and $e^{-n/16} \le \delta \le 1$.
Let $\PP = \bigcup_n \PP_n$ be a 1-generatible monotone property with generator $G = \bigcup_n \{g_n\}$.
Then algorithm \textsc{MonotonePropertyTester} with query access to a ReLU network $(A,w)$ with $n$ input nodes and $m$ hidden layer nodes
\begin{itemize}
\item accepts with probability at least $1-\lambda$, if $(A,w)$ computes a function in \PP, and
\item rejects with probability at least $1-\lambda$, if $(A,w)$ is $(\epsilon,\delta)$-far from computing a function in \PP.
\end{itemize}
The algorithm has a query complexity of $O\left(\frac{\ln(\frac{1}{\delta\epsilon\lambda})}{\delta \epsilon^4}\right)$.
\end{theorem}

\begin{proof}
We may assume that $g_n(\bzero) = 0$ as otherwise there is no ReLU network that computes a property from \PP and the tester trivially works.

Thus, from now on we assume $g_n(\bzero) = 0$. We first prove the first statement. Let $(A,w)$ be a ReLU network that computes a function $f \in \PP_n$. Let us fix an arbitrary $x\in\{0,1\}^n$. By \cref{lemma:sampling}, we get for our choice of $s$ and $t$ that with probability at least $1-\frac{\lambda}{r}$ we have that
\begin{align*}
    \left|
        \frac{mn}{st} \cdot w^T \cdot \relu(TASx) - w^T \cdot \relu(Ax)
    \right|
    \le \tfrac18 \epsilon nm \enspace.
\end{align*}
We observe that if $g_n(x)=0$ then the tester does not reject. If $g_n(x)=1$ then we know that with probability at least $1-\frac{\lambda}{r}$ we have
\begin{align*}
    \frac{mn}{st} \cdot w^T \cdot \relu(TASx) \ge -\tfrac18 \epsilon nm
\end{align*}
and so the tester accepts. Taking a union bound over all $r$ different choices for $x$ the first part follows.

Now consider the case that $(A,w)$ is $(\epsilon,\delta)$-far from computing functions in $\PP_n$.
We know from \cref{lemma:monotoneproperties1} that a random $x\in\{0,1\}^n$ satisfies with probability at least $\frac12 \delta$ that $w^T\relu(Ax) < -\frac14 \epsilon nm$ and $g(x)=1$.
If we take $r$ samples then the probability that we do not sample any $x$ with the above property is at most $(1-\frac{\delta}{2})^{\frac{2}{\delta} \cdot \ln(2/\lambda)} \le \frac12 \lambda$.
By \cref{lemma:sampling} we get for our choice of
$s$ and $t$ that with probability at least $1-\frac{\lambda}{r} \ge 1-\frac12 \lambda$ we have that
\begin{align*}
    \left|
        \frac{mn}{st} \cdot w^T \cdot \relu(TASx) - w^T \cdot \relu(Ax)
    \right|
        \le \tfrac18\epsilon nm\enspace.
\end{align*}
This implies that
\begin{align*}
    \frac{mn}{st} \cdot w^T \cdot \relu(TASx) < - \tfrac18 \epsilon nm
\end{align*}
and so algorithm \textsc{MonotonePropertyTester} rejects.
The latter happens with probability at least $1-\lambda$.

Finally, observe that the query complexity is trivially $O(rst)$.
\end{proof}


\subsection{The general case: testing arbitrary $f(n)$-generatible properties}
\label{sec:monotone-properties-general}

Next we generalize our tester to arbitrary $f(n)$-generatible properties.


\begin{algorithm}[h]
\SetAlgoLined\DontPrintSemicolon
\caption{\textsc{FullMonotonePropertyTester}$(\epsilon,\delta,\lambda, n,m, G_n)$}
\label{alg:FullMonotonePropertyTester}

\lIf{$g_n(\bzero)=1$}{\textbf{reject}.}

Let $r = \frac{2 \ln (2 f(n)/\lambda)}{\delta}$, $t = \frac{512 \ln(4r/\lambda)}{\epsilon^2}$, and $s= \frac{512 \ln(4tr/\lambda)}{\epsilon^2}$.

Sample $s$ nodes from the input layer and $t$ nodes from the hidden layer uniformly at random without replacement.

Let $S \in \NN_0^{n\times n}$ and $T \in \NN_0^{m \times m}$ be the corresponding sampling matrices.

Sample a set $X\subseteq \{0,1\}^n$ of size $r$ uniformly at random.

\lIf{\emph{$\forall g_n \in G_n$} $\exists x \in X$ with $\frac{nm}{st} \cdot w^T \cdot \relu(TASx) + \frac18 \epsilon nm < 0$ and $g_n(x) = 1$}{\textbf{reject};}

\lElse{\textbf{accept}.}
\end{algorithm}


\MonotoneFull*

\begin{proof}
Our property testing algorithm is \textsc{FullMonotonePropertyTester} presented above.

We view $\PP_n$ as the union of 1-generatible properties generated by the elements of $G_n$. We test for all $g \in G_n$ whether the input ReLU network computes a function in the property generated by $g$. If one of the testers accepts, we accept. Otherwise, we reject. We observe that all testers work in the same way and so we can use the same sample for all testers and use a union bound to argue that all of them work properly. Thus, the probability that all of them work properly is at least $1-f(n) \cdot \lambda$.

If the network computes a function $f$ from \PP then there must be at least one 1-generatible property that contains that function and the corresponding tester accepts. If the computed function is $(\epsilon,\delta)$-far from $\PP_n$ then it is $(\epsilon,\delta)$-far from every sub-property of \PP. Thus, if all testers work properly, none of them accepts.
\end{proof}


\subsection{Relation to monotone DNF formulas}

Let us finally briefly remark that monotone properties are related to the standard notion of monotone DNF formulas.

A \emph{disjunctive normal form} formula, or DNF, is a disjunction of Boolean literals. A DNF formula is \emph{monotone} if it contains no negations.
We can relate monotone properties to monotone DNF formulas in the following way. We view a function $f: \{0,1\}^n \rightarrow \{0,1\}$ as an assignment of truth values to a set of $2^n$ variables. That is, for input $x$ the function $f$ assigns the value $f(x)$ to the $i$th variable where $x$ is the binary encoding of $i$.

With this assignment a 1-generatible monotone property $\PP = \bigcup_{n \in \NN} \PP_n$ is simply a property, where each $\PP_n$ can be described by a monomial $M_n$ that does not contain negations (a 1-term monotone DNF). Namely, $g \in \PP_n$ if $M_n(y) = 1$ where $y$ is the vector of assignments to $2^n$ variables that is described by $f$. More generally, a monotone property $\PP = \bigcup_{n \in \NN} \PP_n$ is a property where each $\PP_n$ can be described in the same way by a disjunctive normal form (DNF) without negations. And a $f(n)$-generatible monotone property $\PP = \bigcup_{n \in \NN} \PP_n$ is a property where each $\PP_n$ can be described in the same way by a $f(n)$-term monotone DNF.


\section{ReLU networks with \emph{multiple outputs} (and \emph{one hidden layer})}
\label{sec:multiple-outputs}

Our main analysis so far considered the model of ReLU networks with one hidden layer that returns a single output bit. However, it is natural to consider also an extension of this setting to allow multiple outputs. In this section we consider such extension and demonstrate that in some natural cases, this model can be reduced to the model with a single output. We begin with a formal definition of the model considered and then, in Sections~\ref{sec:multiple-outputs-reduction}--\ref{sec:multiple-outputs-applications}, present our main results.

\subsection{Model of ReLU networks with multiple outputs}
\label{sec:multiple-outputs-model}

We begin by generalizing our definitions from \cref{sec:property-testing-model}.
We first extend the definition of $\sgn$ to vectors $v = (v_1,\dots,v_r)$ by applying $\sgn$ element-wise and thus defining $\sgn(v) := (\sgn(v_1),\dots,\sgn(v_r))$.
Next, we extend \cref{def:ReLU} in a natural way to allow multiple outputs (output bits).

\begin{definition}[\textbf{ReLU network with multiple outputs}]
\label{def:ReLU-mo}
A ReLU network with $n$ input nodes, $m$ hidden layer nodes, and $r$ output nodes is a pair $(A,W)$, where $A\in[-1,1]^{m\times n}$ and $W \in [-1,1]^{m \times r}$. The function $f:\{0,1\}^n \rightarrow \{0,1\}^r$ computed by a ReLU network $(A,W)$ is defined as
\begin{align*}
    f(x) &:= \sgn(\relu(W^T \cdot \relu(Ax))) \enspace.
\end{align*}
\end{definition}


\paragraph{Query access to the network.}
Similarly to the case of a single output node, we will allow a property testing algorithm to query weights, i.e., entries in $A$ or $W$ in constant time per query.

\medskip

Next, we extend the definition of properties to functions with multiple output~bits.

\begin{definition}[\textbf{Property of functions}]
\label{def:props-of-functions-mo}
A \textbf{property of functions} from $\{0,1\}^n$ to $\{0,1\}^r$ is a family $\PP = \bigcup_{n \in \NN} \PP_n$, where $\PP_n$ is a set of functions $f:\{0,1\}^n \rightarrow \{0,1\}^r$.
\end{definition}

Next, we extend \cref{def:ReLU-farness-from-property} to define farness and distances to functions and properties for the case of multiple output bits.

\begin{definition}[\textbf{ReLU network being far from computing a function}]
\label{def:ReLU-farness-from-function-mo}
Let $(A,W)$ be a ReLU network with $n$ input nodes, $m$ hidden layer nodes, and $r$ output nodes. $(A,W)$ is called \textbf{$(\epsilon,\delta)$-close to computing a function} $f:\{0,1\}^n \rightarrow \{0,1\}^r$, if one can change the matrix $A$ in at most $\epsilon nm$ places and the matrix $W$ in at most $\epsilon mr$ places to obtain a ReLU network that computes a function $g$ such that $\Pr[g(x) \ne f(x))]\le \delta$, where $x$ is chosen uniformly at random from $\{0,1\}^n$. If $(A,W)$ is not $(\epsilon,\delta)$-close to computing $f$ we say that it is \textbf{$(\epsilon,\delta)$-far from computing $f$}.
\end{definition}

\begin{definition}[\textbf{ReLU network being far from a property of functions}]
\label{def:ReLU-farness-from-property-mo}
Let $(A,W)$ be a ReLU network with $n$ input nodes, $m$ hidden layer nodes, and $r$ output nodes. $(A,W)$ is called \textbf{$(\epsilon,\delta)$-close to computing a function} $f:\{0,1\}^n \rightarrow \{0,1\}^r$ with property $\PP = \bigcup_{n \in \NN} \PP_n$, if one can change the matrix $A$ in at most $\epsilon nm$ places and $W$ in at most $\epsilon mr$ places to obtain a ReLU network that computes a function $g$ such that $\Pr[g(x) \ne f(x)] \le \delta$ for some function $f \in \PP_n$, where $x$ is chosen uniformly at random from $\{0,1\}^n$. If $(A,W)$ is not $(\epsilon,\delta)$-close to computing $f$ with property \PP then we say that it is \textbf{$(\epsilon,\delta)$-far from computing $f$ with property \PP}.
\end{definition}


\subsection{Reducing multiple outputs networks to single output networks}
\label{sec:multiple-outputs-reduction}

In this section we show that our testers for computing the constant $0$-function and the OR-function (\cref{thm:ReLU-testing-0-OR-function-2-sided})
can be extended to the setting with multiple outputs to obtain a property tester for computing any near constant function. Before we proceed, let us introduce the notion of near constant functions with multiple outputs.

\begin{definition}[\textbf{Near constant functions}]
\label{def:constant-functions-mo}
We call a function $f: \{0,1\}^n \rightarrow \{0,1\}^r$ \textbf{near constant} if there is $\bb \in \{0,1\}^r$
such that $f(x) = \bb$ for every $x \in \{0,1\}^n \setminus\{\bzero\}$ and $f(\bzero)=0.$
\end{definition}

For $\ob_i \in \{0,1\}$ we will sometimes refer to the near constant function $\ob_i: \{0,1\}^n \rightarrow \{0,1\}$ as the function
$f(x)= \ob_i$ for all $x \in \{0,1\}^n \setminus \{\bzero\}$ and $f(\bzero) = 0$.

\begin{remark}\rm
\label{remark:constant-functions-and-ReLU}
Notice that our definition of ReLU networks (\cref{def:ReLU-farness-from-property,def:ReLU-mo}) implies that on input $\bzero$ the output is \emph{always} $0$. Hence, the only \emph{truly constant function} that can be computed by a ReLU network is the constant $\bzero$ function. For all other vectors \bb we will consider a function that is constant on all inputs but the $\bzero$-vector. We observe that if a network computes a near constant function $f(x) = \bb$ for $x \ne \bzero$ and $\bb=(\ob_1,\dots, \ob_r)$ then the restriction of the network to any output node $i$ with $\ob_i=0$ is the constant $0$-function and to any output node $i$ with $\ob_i=1$ is the OR-function (see a formal definition of restricted network below).
\end{remark}

For any ReLU network $(A,W)$ with $r$ output nodes, define a \emph{ReLU network $(A,W_i)$ restricted to the output node $i$} to be the sub-network of $(A,W)$ restricted to the output node $i$ as well as the input and hidden layer (here $W_i$ is the $m$-vector corresponding to the $i$-th column in $W$). That is, $(A,W_i)$ is the ReLU network which is a sub-network of $(A,W)$ computing $\sgn(\relu(W_i^T \cdot \relu(Ax)))$.

The following theorem, central for our analysis, demonstrates that if a ReLU network with $r$ outputs is $(\epsilon, \delta)$-far from computing near constant function \bb, then for at least an $\frac{\epsilon}{2}$-fraction of the output bits $i$, the ReLU network restricted to the output node $i$ is $(\epsilon',\delta')$-far from computing an appropriate near constant function, where $\delta' \sim \delta/r$ and $\epsilon' = \epsilon^2/1025$.

\TestingConstantFunctionReduction*

\junk{
\begin{theorem}
\label{thm:reduction-constant-ReLU-mo}
Let \NE be a ReLU network $(A,W)$ with $n$ input nodes, $m$ hidden layer nodes, and $r$ output nodes.
Let $e^{-n/16} \le \delta < 1$ and $16 \cdot \sqrt{\frac{\ln(2m)}{m}} \le \epsilon < \frac12$. Let $\bb = (\ob_1, \dots, \ob_r) \in \{0,1\}^r$. If \NE is $(\epsilon, \delta)$-far from computing a near constant function \bb then there are more than $\frac12 \epsilon r$ output nodes such that for any such output node $i$, the ReLU network \NE restricted to the output node $i$ is $(\frac{\epsilon^2}{1025}, \frac{\delta - e^{-n/16}}{r})$-far from computing near constant function $\ob_i$.
\end{theorem}
}

\begin{proof}
Fix $\bb \in \{0,1\}^r$ and let $\delta' := \frac{\delta-e^{-n/16}}{r} \ge 0$; hence $\delta = r \delta' + e^{-n/16}$.
Let $B_0$ be the set of output nodes $i$ with $\ob_i = 0$ and $B_1$ be the set of remaining output nodes. Let $Q$ be the set of output nodes $i$ such that the restricted network $(A,W_i)$ is $(\frac{\epsilon^2}{1025},\delta')$-far from computing the near constant function $\ob_i$. Let $Q_0 := Q \cap B_0$, $Q_1 := Q \cap B_1$, $L_0 := B_0 \setminus Q_0$, and $L_1 := B_1 \setminus Q_1$.

The proof is by contradiction. Assume that there are at most $\frac12 \epsilon r$ output nodes $i$ such that the restricted network $(A,W_i)$ is $(\frac{\epsilon^2}{1025},\delta')$-far from computing the near constant function $\ob_i$, that is, that $|Q| \le \frac12 \epsilon r$. We will show that then we can modify weights of at most $\epsilon n m$ first layer edges and of at most $\epsilon r m$ second layer edges to obtain a ReLU network that computes \bb for at least a $(1-\delta)$-fraction of the inputs. This would contradict the assumption that $(A,W)$ is $(\epsilon,\delta)$-far from computing the near constant function \bb, and thus, would complete the proof by contradiction.

In the following, we describe our modifications. We assume that $|Q| \le \frac12 \epsilon r$.
\begin{enumerate}[(I)]
\item For each $i \in Q$, we set to $\ob_i$ the weights of all second layer edges incident to the output node $i$. This modifies $|Q| \cdot m \le \frac12 \epsilon r m$ weights of second layer edges and ensures that each output node in $Q_0$ computes the constant $0$-function. For nodes in $Q_1$ we also need to make sure that on input $x \ne \bzero$ at least one hidden layer node outputs a $1$. At the very end of our construction, we will select a certain hidden layer node $y^*$ and set to $1$ the  weights of all its edges incident to the input nodes to take care of this.

Overall, this phase requires modification of weights of at most $\frac12 \epsilon mr$ second layer edges.

\item Consider the remaining output nodes - these are the output nodes in $L_0 \cup L_1$. For all such nodes $i$, the restricted network $(A,W_i)$ is $(\frac{\epsilon^2}{1025},\delta')$-close to computing near constant function $\ob_i$ (and before the changes applied to $y^*$).
\begin{itemize}
\item For any output node $i \in L_0$ incident to at most $\frac14 \epsilon m$ edges with positive edge weights in the second layer we set all these edge weights to $0$. Since no other edges incident to output node $i$ have positive weights, the output node computes the constant $0$-function on all inputs.

\item Similarly, for any output node $i \in L_1$ incident to at most $\frac14 \epsilon m$ edges with non-positive edge weights in the second layer we set all these edge weights to $1$. Since all other edges incident to output node $i$ have positive weights, this implies that after the modification all edges incident to $i$ have positive weights. Therefore, the modification of  $y^*$ at the end of the construction will also make sure that the output value will be positive for non-zero inputs.
\end{itemize}
Overall, this phase requires modification of weights of at most $\frac14 \epsilon mr$ second layer edges.

\item It remains to deal with the output nodes $i$ in $L_0 \cup L_1$ that are either incident to more than $\frac14 \epsilon m$ edges with positive weight (if $i \in L_0$) or incident to more than $\frac14 \epsilon m$ edges with non-positive weights (if $i \in L_1$). Let $T$ be the set of such output nodes and let $T_0 := T \cap L_0$ and $T_1 := T \cap L_1$. We will show that there exist two disjoint sets $R_0$ and $R_1$ such that $R_0$ is a set of $\frac18 \epsilon m$ second layer nodes such that every node from $T_0$ has at least $\frac1{64} \epsilon^2 m$ positive edges incident to nodes from $R_0$, and $R_1$ is a set of $\frac18 \epsilon m$ second layer nodes such that every node from $T_1$ has at least $\frac1{64}\epsilon^2 m$ non-positive edges incident to nodes from $R_1$.

We use the probabilistic method to show that $R_0$ and $R_1$ exist. Select uniformly at random two disjoint sets $R_0'$ and $R_1'$ of $\frac18 \epsilon m$ hidden layer nodes each. Note that one may first sample $R_0'$ uniformly at random and then $R_1'$ uniformly at random from the remaining nodes or the other way around. Thus, we may think of each of the sets $R_0'$ and $R_1'$ to be sampled uniformly at random without replacement from the hidden layer nodes.

Fix an output node $i \in T_0$ and let $X_j$ be the indicator random variable for the event that the end-point of the $j$th edge with positive weight incident to $i$ is selected in $R_0'$. Notice that $\Ex[X_j] = \frac18 \epsilon$. Since any output node in $T_0$ is incident to more than $\frac14 \epsilon m$ positive edges, we have $\Ex\left[\sum_j X_j \right] \ge \frac1{32}\epsilon^2 m$. Thus, by Chernoff bounds (we remark that Chernoff bounds also apply to the setting of sampling without replacement; see, for example, \cite{BM15}) we get
\begin{align*}
    \Pr\left[\sum_j X_j \le \frac1{64} \epsilon^2 m \right] & \le
    \Pr\left[\sum_j X_j \le \frac{1}{2} \Ex\left[\sum_j X_j\right] \right]
    \le e^{-\Ex[\sum_j{X_j}]/8}
    \le \frac{1}{2m},
\end{align*}
where the last inequality holds for $1/2\ge \epsilon \ge 16 \cdot \sqrt{\frac{\ln(2m)}{m}}$. Identical arguments can be used to show that (assuming $1/2\ge \epsilon \ge 16 \cdot \sqrt{\frac{\ln(2m)}{m}}$) for any fixed output node $i \in T_1$, the probability that $i$ is incident to at most $\frac1{64}\epsilon^2 m$ edges with non-positive weight connecting to nodes from $R$ is at most $\frac{1}{2m}$. By taking a union bound over all nodes from $T$, this implies that with probability at least $\frac12$ the random sets $R_0'$ and $R_1'$ satisfy the requirements of $R_0$ and $R_1$. In particular, there exist two disjoint sets $R_0$ and $R_1$ such that every output node $i \in T_0$ has at least $\frac1{64}\epsilon^2 m$ edges with positive weight connecting $i$ to nodes from $R_0$ and every output node $i \in T_1$ has at least $\frac1{64}\epsilon^2 m$ edges with non-positive weight connecting $i$ to nodes from~$R_1$.

Next, modify the edge weights of the ReLU network as follows:
\begin{itemize}\itemsep0em
\item set to $+1$ the weights of all $\frac14 \epsilon mn$ first layer edges incident to $R_0 \cup R_1$,
\item set to $-1$ the weights of all second layer edges between $R_0$ and $B_0$, and
\item set to $+1$ the weights of all second layer edges between $R_1$ and $B_1$.
\end{itemize}
Overall, this phase requires modification of weights of at most $|R_0| \cdot |B_0| + |R_1| \cdot |B_1| \le \frac14 \epsilon m r$ second layer edges and of $\frac14 \epsilon mn$ first layer edges. We also observe that our modification does not increase the value at output nodes from $Q_0$ and does not decrease the value at output nodes from $Q_1$.

Notice that the value of every node in $R_0 \cup R_1$ is equal to $\|x\|_1$, the number of 1s in the input. Therefore, since there were originally at least $\frac1{64}\epsilon^2m$ edges with positive weight between node $i$ and $R_0$, the modifications of edge weights ensure that every output bit corresponding to output node $i \in T_0$ is reduced by at least $\|x\|_1 \cdot \frac1{64}\epsilon^2m$. Similarly, since there were originally at least $\frac1{64}\epsilon^2m$ edges with non-positive weights between node $i$ and $R_1$, the modifications of edge weights ensure that every output bit corresponding to output node $i \in T_1$ is increased by at least $\|x\|_1 \cdot \frac1{64} \epsilon^2m$. Notice that (see, e.g., Claim \ref{claim:NumberOfOnes}) the probability that a random $x \in \{0,1\}^n$ satisfies $\|x\|_1 \ge \frac14 n$ is at most $e^{-n/16}$. Therefore our construction above ensures that for any output node $i \in T_0$, for at least a $(1-e^{-n/16})$-fraction of the inputs we \emph{reduce the value at $i$} by at least $\frac1{256}\epsilon^2mn$, and for any output node $i \in T_1$, for at least a $(1-e^{-n/16})$-fraction of the inputs we \emph{increase the value at $i$} by at least $ \frac1{256}\epsilon^2mn$.

Observe that a modification of an edge weight in the first layer changes the value of a single output node by at most $2$; one edge weight change in the second layer may change it by at most $2n$. Therefore, since for every output node $i \in T_0$, the restricted network $(A,W_i)$ is $(\frac{\epsilon^2}{1025},\delta')$-close to computing constant $0$-function, we obtain that before our modifications, for any output node $i \in T_0$, for at least a $(1-\delta')$-fraction of the inputs the output value of $i$ was at most $2 \cdot (\frac{1}{1025} \epsilon^2 m n) + 2n \cdot (\frac{1}{1025} \epsilon^2 m) < \frac1{256} \epsilon^2 n m$ as otherwise, the output node could not be $(\frac{\epsilon^2}{1025},\delta')$-close to computing the constant $0$-function. Similarly, identical arguments imply that for any output node $i \in T_1$, for at least a $(1-\delta')$-fraction of the inputs the output value of $i$ was greater than $- \frac1{256} \epsilon^2 n m$.

By the arguments above, for every output node $i \in T_0$, on one hand, before our modifications, for at least a $(1-\delta')$-fraction of the inputs the output value of $i$ was strictly less than $\frac1{256} \epsilon^2 n m$, and on the other hand, for at least a $(1-e^{-n/16})$-fraction of the inputs we \emph{reduce} the value at $i$ by at least $\frac1{256}\epsilon^2mn$. This implies that for every $i \in T_0$, after our modifications, for at least a $(1-\delta'-e^{-n/16})$-fraction of the inputs the output value of $i$ is negative, and hence the ReLU network $(A,W_i)$ restricted to the output node $i$ returns $0$. Identical arguments imply that for every output node $i \in T_1$, after our modifications, for at least a $(1-\delta'-e^{-n/16})$-fraction of the inputs the output value of $i$ is positive, and hence the ReLU network $(A,W_i)$ restricted to the output node $i$ returns~$1$.

\item We still need to argue how to choose the node $y^*$. If at the beginning of our construction we had $L_0 \cup L_1 \ne \emptyset$ any node from $R_0$ or $R_1$ will do and no further modification is necessary. Otherwise, We can take an arbitrary input node and set all of its $n$ first layer edges to $1$.
\end{enumerate}

Now let us summarize our discussion above. We have shown that if there are at most $\frac12 \epsilon r$ output nodes $i$ such that the restricted network $(A,W_i)$ is $(\frac{\epsilon^2}{1025},\delta')$-far from computing constant function $\ob_i$, then we can modify network $(A,W)$ by changing weights of at most $\epsilon n m$ first layer edges and at most $\epsilon r m$ second layer edges into a network that computes \bb on all but a $\delta$-fraction of the inputs.

This implies that $(A,W)$ is $(\epsilon, \delta)$-close to computing the near constant function \bb, which is a contradiction, completing the proof of \cref{thm:reduction-constant-ReLU-mo}.
\end{proof}


\subsection{Applications: \emph{multiple output nodes}}
\label{sec:multiple-outputs-applications}

It is now easy to combine \cref{thm:reduction-constant-ReLU-mo} with our analysis of ReLU networks with a single bit output for the constant $0$-function and the OR-function (\cref{thm:ReLU-testing-0-OR-function-2-sided}) to obtain the following theorem.

\TestingConstantFunctionShl*

\junk{
\begin{theorem}
\label{ReLU-testing-constant-function-2-sided-mo}
Let \NE be a ReLU network $(A,W)$ with $n$ input nodes, $m$ hidden layer nodes, and $r$ output nodes. Let $\bb \in \{0,1\}^r$, $(r+1) e^{-n/16} \le \delta < 1$, and $c \cdot \sqrt{\ln(2m)/m} < \epsilon < \frac12$ for a sufficiently large constant $c$. There is a property tester that queries $O(\frac{\ln^2(1/\epsilon \lambda)}{\epsilon^{13}})$ entries from $A$ and $W$ and
\begin{itemize}
\item accepts with probability at least $\frac23$, if \NE computes near constant function \bb, and
\item rejects with probability at least $\frac23$, if \NE is $(\epsilon,\delta)$-far from computing near constant function \bb.
\end{itemize}
\end{theorem}
}

\begin{proof}
Let $\bb = (\ob_1, \dots, \ob_r) \in \{0,1\}^r$. By \cref{thm:reduction-constant-ReLU-mo}, if $(A,W)$ is $(\epsilon,\delta)$-far from computing near constant function \bb then there are more than $\frac12 \epsilon r$ output nodes such that for any such output node $i$, the ReLU network $(A,W_i)$ restricted to the output node $i$ is $(\frac{\epsilon^2}{1025}, \frac{\delta - e^{-n/16}}{r})$-far from computing near constant function $\ob_i$.

Therefore \cref{ReLU-testing-constant-function-2-sided-mo} follows by sampling $O(1/\epsilon)$ output nodes uniformly at random and then, for each sampled output node $i$, testing whether the ReLU network $(A,W_i)$ restricted to the output node $i$ computes near constant function $\ob_i$ using the algorithm from \cref{thm:ReLU-testing-0-OR-function-2-sided} (either \textsc{AllZeroTester} or \textsc{ORTester}, depending on $\ob_i$). By outputting the majority outcome of the calls to \cref{thm:ReLU-testing-0-OR-function-2-sided}, we obtain the promised property testing algorithm.
\end{proof}

\junk{

Similarly, we can incorporate \cref{thm:reduction-constant-ReLU-mo} to extend \cref{theorem:0and1} to ReLU networks with multiple output bits.

\begin{theorem}
\label{theorem:close-to-constant-functions}
Let $(A,W)$ be a ReLU network with $n$ input nodes, $m$ hidden layer nodes, and $r$ output nodes. Let $0 < \delta \le \frac12$ and let $\epsilon \ge c \cdot \max\{\sqrt{\frac{\ln(2m)}{m}}, \sqrt{\frac{\log(2/\delta)}{n}}\}$ for a sufficiently large constant $c$. 
Then $(A,w)$ is $(\epsilon,\delta)$-close to computing some constant function. That is, for every such ReLU network $(A,W)$ there is $\bb \in \{0,1\}^r$ such that $(A,w)$ is $(\epsilon,\delta)$-close to computing function $f:\{0,1\}^n \rightarrow \{0,1\}^r$ with $f(x) = \bb$ for every $x \in \{0,1\}^n$.
\end{theorem}

\begin{proof}
\Artur{Why is this missing?}
\end{proof}
}


\section{ReLU networks with \emph{multiple hidden layers}}
\label{sec:multiple-layers}

Our analysis so far has studied the model of ReLU networks with \emph{one hidden layer}. In this section we consider a natural generalization of this setting to allow multiple hidden layers.


\subsection{ReLU networks with multiple hidden layers and multiple outputs}
\label{sec:multiple-layers-and-outputs-model}

We begin with a formal definition of the model of ReLU network with multiple layers, generalizing our definitions from \cref{sec:property-testing-model} and \cref{sec:multiple-outputs-model}.

We start with the extension of \cref{def:ReLU} and \cref{def:ReLU-mo} in a natural way.

\begin{definition}[\textbf{ReLU network with multiple hidden layers and multiple outputs}]
\label{def:ReLU-mol}
A ReLU network with $n$ input nodes, $\ell \ge 1$ hidden layers with $m_1, \dots, m_{\ell}$ hidden layer nodes each, and $r$ output nodes is a tuple $(W_0, W_1, \dots, W_{\ell})$, where $W_i \in [-1,1]^{m_{i+1} \times m_i}$ for every $0 \le i \le \ell$, with $m_0 = n$ and $m_{\ell+1} = r$. The function $f:\{0,1\}^n \rightarrow \{0,1\}^r$ computed by a ReLU network $(W_0, \dots, W_{\ell})$ is defined recursively as $f(x) := \sgn(\relu(f_{\ell+1}(x)))$,
where\footnote{Observe that unfolding the recurrence in \cref{def:ReLU-mo} gives (and notice that $\relu(x) = x$ for $x \in \{0,1\}^n$),
\begin{align*}
    f_i(x) &:=
        \begin{cases}
            x & \text{ if } i = 0, \\
            W_0 \cdot \relu(f_0(x)) =
                W_0 \cdot \relu(x) &
                \text{ if } i = 1, \\
            W_1 \cdot \relu(f_1(x)) =
                W_1 \cdot \relu(W_0 \cdot \relu(x)) &
                \text{ if } i = 2, \\
            W_2 \cdot \relu(f_2(x)) =
                W_2 \cdot \relu(W_1 \cdot \relu(W_0 \cdot \relu(x))) &
                \text{ if } i = 3, \\
            W_3 \cdot \relu(f_3(x)) =
                W_3 \cdot \relu(W_2 \cdot \relu(W_1 \cdot \relu(W_0 \cdot \relu(x)))) &
                \text{ if } i = 4
        \end{cases}
\end{align*}
Therefore, a non-recursive version of the value $f(x)$ computed by the network (in a form resembling \cref{def:ReLU}) is
\begin{align*}
    f(x) &:= \sgn(\relu(W_{\ell} \cdot \ \dots \ \cdot
        \relu(W_2 \cdot \relu(W_1 \cdot \relu(W_0 \cdot \relu(x))))
    \dots )).
\end{align*}
}
\begin{align*}
    f_i(x) &:=
        \begin{cases}
            W_{i-1} \cdot \relu(f_{i-1}(x)) & \text{ if } 1 \le i \le \ell+1, \\
            x & \text{ if } i=0 \enspace.
        \end{cases}
\end{align*}
\end{definition}


\paragraph{An intuitive description.}
\cref{def:ReLU-mol} can be also seen as a recursive description of the following computational process: first compute $W_0 \cdot x$, then apply the ReLU activation function, then multiply $W_1$ by it, then again apply the ReLU activation function, then multiply $W_2$ by it, and so, so forth, until after multiplying $W_{\ell}$ by the result, we apply the ReLU activation function, and then, finally, the $\sgn$ function.


\paragraph{Another intuitive description.}
\Artur{I like the description below, but most likely it's not needed.}%
While \cref{def:ReLU-mol} precisely and formally describes the model of ReLU network with multiple hidden layers and multiple outputs, it is good to keep in mind also its more informal tree-based definition (see also Figure~\ref{fig:ReLU-network-mL}).

The ReLU network $(W_0, \dots, W_{\ell})$ corresponds to a layered weighted network with $\ell+2$ layers: layer zero containing $m_0$ \emph{input nodes}, $\ell$ \emph{hidden layers} $1, \dots, \ell$, with layer $1 \le i \le \ell$ containing $m_i$ nodes, and layer $(\ell+1)$ containing $m_{\ell+1}$ \emph{output nodes}. The network has weighted edges connecting pairs of consecutive layers. For every $0 \le k \le \ell$ and for every $1 \le i \le m_k$ and $1 \le j \le m_{k+1}$, there is an edge with weight $W_k[j,i]$ connecting the $i$-th node at layer $k$ with the $j$-th node at layer $k+1$. We will be assuming that all weights are real numbers in $[-1,1]$.

We use such ReLU layered network as a computational device computing a function $f:\{0,1\}^{m_0} \rightarrow \{0,1\}^{m_{\ell+1}}$ as follows. The input is any vector $x \in \{0,1\}^{m_0}$, so that the $i$-th input node obtains as its input value $x_i$. Then, any node in the network at a layer $0 \le k \le \ell+1$ has a value determined by the input $x$ and the weights of the network, which is computed bottom-up in the network as follows:
\begin{itemize}
\item for any node $i$ at layer $0$, the \emph{value} of that node $\val^0_i(x)$ is $x_i$;
\item for any node $j$ at layer $1 \le k \le \ell+1$, the \emph{value} of that node $\val^k_j(x)$ is equal to
\begin{align*}
    \val^k_j(x) &= \sum_{i=1}^{m_{k-1}} W_{k-1}[j,i] \cdot \relu(\val^{k-1}_i(x));
\end{align*}
\item for any node $j$ at layer $\ell+1$, the binary \emph{output} computed by that node is equal to
\begin{align*}
    \out_j(x) &:= \sgn(\relu(\val^{\ell+1}_j(x))).
\end{align*}
\end{itemize}


\begin{figure}[t]
\centerline{\includegraphics[width=.9\textwidth]{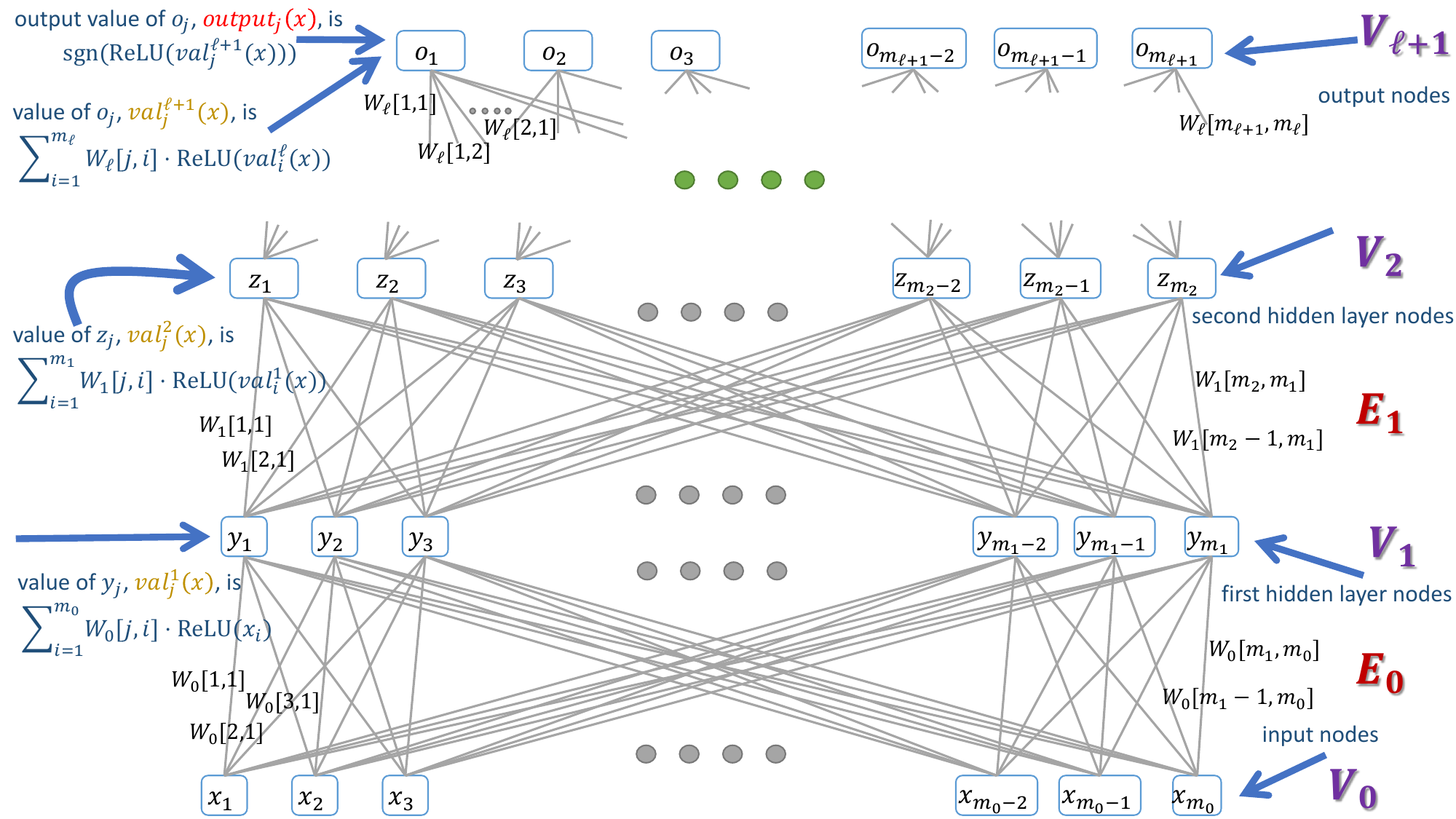}}
\caption{\small{}
A ReLU network $(W_0, \dots, W_{\ell})$ with $\ell+2$ layers: layer zero $V_0$ containing $m_0$ \emph{input nodes}, $\ell$ \emph{hidden layers} $1, \dots, \ell$, with layer $1 \le k \le \ell$ containing $m_k$ nodes $V_k$, and layer $(\ell+1)$ containing $m_{\ell+1}$ \emph{output nodes} $V_{\ell+1}$. The network has weighted edges connecting pairs of consecutive layers. For every $0 \le k \le \ell$ and for every $1 \le i \le m_k$ and $1 \le j \le m_{k+1}$, there is an edge with weight $W_k[j,i] \in [-1,1]$ connecting the $i$-th node at layer $k$ with the $j$-th node at layer $k+1$.
For a given input $x = (x_1, \dots, x_{m_0})^T \in \{0,1\}^{m_0}$, for any node $j$ at layer $0$, the \emph{value} of that node $\val^0_j(x) = x_j$ (which is also the $j$-th row of $f_0(x)$).
For any node $j$ at layer $1 \le k \le \ell+1$, the \emph{value} of that node is $\val^k_j(x) = \sum_{i=1}^{m_{k-1}} W_{k-1}[j,i] \cdot \relu(\val^{k-1}_i(x))$; notice that this is the $j$-th row of $f_k(x)$.
For any node $j$ at layer $\ell+1$ (i.e., for the $j$-th output node), the binary \emph{output} computed by that node equals $\out_j(x) = \sgn(\relu(\val^{\ell+1}_j(x)))$; notice that this is the $j$-th row of~$f(x)$.
}
\label{fig:ReLU-network-mL}
\end{figure}


It is easy to see that this definition matches \cref{def:ReLU-mol}. That is, let us consider a ReLU network with $m_0$ input nodes, $\ell \ge 1$ hidden layers with $m_1, \dots, m_{\ell}$ hidden layer nodes each, and $m_{\ell+1}$ output nodes. Then for the function $f(x) := \sgn(\relu(f_{\ell+1}(x)))$ as in \cref{def:ReLU-mol} computed by $(W_0, W_1, \dots, W_{\ell})$, where $W_i \in [-1,1]^{m_{i+1} \times m_i}$ for every $0 \le i \le \ell$, the $j$-th term of $f(x)$ (the $j$-th row of vector $f(x)$) is equal to $\out_j(x)$, the output of the $j$-th node at the final layer $\ell+1$. Furthermore, $\val_k(j)$, the value of node $j$ at layer $0 \le k \le \ell+1$, is equal to the $j$-th term of vector $f_k(x)$ (the $j$-th row of vector $f_k(x)$) defined by $(W_0, W_1, \dots, W_{\ell})$.


\paragraph{Query access to the network.}
Similarly to the case of a single hidden layer, we will allow a property testing algorithm to query weights, i.e., entries in $W_0, \dots, W_{\ell}$ in constant time per query.


\subsubsection{The notion of a ReLU network being far}
\label{sec:ReLU-farness-mL}

Next, we extend the definitions (\cref{def:ReLU-farness-from-property} and \cref{def:ReLU-farness-from-function-mo,def:ReLU-farness-from-property-mo}) of farness and distances to functions and properties for ReLU networks with multiple hidden layers and multiple outputs as follows.

\begin{definition}[\textbf{ReLU network being far from computing a function}]
\label{def:ReLU-farness-from-function-mol}
Let $(W_0, \dots, W_{\ell})$ be a ReLU network with $m_0 = n$ input nodes, $\ell \ge 1$ hidden layers with $m_1, \dots, m_{\ell}$ hidden layer nodes each, and $m_{\ell+1}$ output nodes. $(W_0, \dots, W_{\ell})$ is called \textbf{$(\epsilon,\delta)$-close to computing a function} $f:\{0,1\}^n \rightarrow \{0,1\}^{m_{\ell+1}}$, if one can change each matrix $W_i$, $0 \le i \le \ell$, in at most $\epsilon m_i m_{i+1}$ places to obtain a ReLU network that computes a function $g$ such that $\Pr[g(x) \ne f(x))] \le \delta$, where $x$ is chosen uniformly at random from $\{0,1\}^n$. If $(W_0, \dots, W_{\ell})$ is not $(\epsilon,\delta)$-close to computing $f$ we say that it is \textbf{$(\epsilon,\delta)$-far from computing $f$}.
\end{definition}

\begin{definition}[\textbf{ReLU network being far from a property of functions}]
\label{def:ReLU-farness-from-property-mol}
Let $(W_0, \dots, W_{\ell})$ be a ReLU network with $m_0 = n$ input nodes, $\ell \ge 1$ hidden layers with $m_1, \dots, m_{\ell}$ hidden layer nodes each, and $m_{\ell+1}$ output nodes. $(W_0, \dots, W_{\ell})$ is called \textbf{$(\epsilon,\delta)$-close to computing a function} $f:\{0,1\}^n \rightarrow \{0,1\}^{m_{\ell+1}}$ with property $\PP = \bigcup_{n \in \NN} \PP_n$, if one can change each matrix $W_i$, $0 \le i \le \ell$, in at most $\epsilon m_i m_{i+1}$ places to obtain a ReLU network that computes a function $g$ such that $\Pr[g(x) \ne f(x)] \le \delta$ for some function $f \in \PP_n$, where $x$ is chosen uniformly at random from $\{0,1\}^n$. If $(W_0, \dots, W_{\ell})$ is not $(\epsilon,\delta)$-close to computing $f$ with property \PP then we say that it is \textbf{$(\epsilon,\delta)$-far from computing $f$ with property \PP}.
\end{definition}


\subsection{Testing the constant 0-function of a ReLU network with multiple layers}
\label{sec:ReLU-testing-0-function-2-sided-mL}

We now present Algorithm \textsc{AllZeroTesterMHL} which extends Algorithm 
\textsc{AllZeroTester} for ReLU networks with a single hidden layer from \cref{sec:ReLU-testing-0-OR-function-2-sided} and which tests whether a given ReLU network with multiple hidden layers computes the constant $0$-function.

The tester follows the framework from 
\cref{sec:ReLU-testing-0-OR-function-2-sided}
and it samples nodes from each layer and computes the function generated by the corresponding sub-network of the original ReLU network.

As before, we use matrix notation and write $S_i$ to denote a sampling matrix that selects $s_i$ input/hidden layer nodes $\sset_i$ at layer $i$ of a ReLU network at random. That is, $S_i$ is an $m_i \times m_i$ diagonal matrix whose diagonal entries $(i,i)$ are equal to 1 iff node $j$ at the hidden layer $i$ (or the $j$-th input node if $i=0$) is in the sample set (and is 0 otherwise). With this notation we can describe our Algorithm \ref{alg:AllZeroTester-mL}  \textsc{AllZeroTesterMHL} below.


\begin{algorithm}[h]
\SetAlgoLined\DontPrintSemicolon
\caption{\textsc{AllZeroTesterMHL}$(\epsilon, \lambda, m_0, m_1, \dots, m_{\ell})$}
\label{alg:AllZeroTester-mL}

\KwIn{ReLU network $(W_0, \dots, W_{\ell})$ with $m_0$ input nodes, $\ell \ge 1$ hidden layers with $m_1, \dots, m_{\ell}$ hidden layer nodes each, and a single output node;
parameters $\epsilon$ and $\lambda$}

set $s_0 := c_0 \cdot \epsilon^{-2\ell} \cdot \ln (1/\lambda\epsilon)$ for an appropriate constant $c_0 = c_0(\ell) > 0$

for every $1 \le k \le \ell$, set $s_k := c_k \cdot \epsilon^{-4\ell} \cdot \ln^2(1/\lambda\epsilon)$ for an appropriate constant $c_k = c_k(\ell) > 0$

\For{$i=0$ \KwTo $\ell$}{
    sample $s_i$ nodes $\sset_i$ from $V_i$ uniformly at random without replacement

    let $S_i \in \NN_0^{m_i \times m_i}$ be the corresponding sampling matrix
}

let $h :\{0,1\}^{m_0} \rightarrow \{0,1\}$ be defined recursively as $h(x) := \sgn(\relu(h_{\ell+1}(x)))$, where
\begin{align*}
    h_i(x) &:=
        \begin{cases}
            W_{i-1} \cdot S_{i-1} \cdot \relu(h_{i-1}(x)) & \text{ if } 1 \le i \le \ell+1, \\
            x & \text{ if } i=0.
        \end{cases}
\end{align*}

\eIf{\emph{there is $x \in \{0,1\}^{m_0}$ with $\prod_{i=0}^{\ell} \frac{m_i}{s_i} \cdot h_{\ell+1}(x) - \frac{1}{16} \cdot (\epsilon/2)^{\ell} \cdot \prod_{i=0}^{\ell} m_i > 0$}}{\textbf{reject}}
{\textbf{accept}}
\end{algorithm}


Let us observe that this algebraic description of \textsc{AllZeroTesterMHL} can be also visualized using the following description: 
We take the ReLU network $(W_0, \dots, W_{\ell})$, and for every $0 \le i \le \ell$, sample $s_i$ nodes $\sset_i$ from $V_i$ uniformly at random without replacement. Then, $h_{\ell+1}$ is equal to (the modified value of) $\val^{\ell+1}_1(x)$. And so, \textsc{AllZeroTesterMHL} checks whether after randomly sampling sets $\sset_0, \dots, \sset_{\ell}$ of sufficiently large size from $V_0, \dots, V_{\ell}$, respectively, there is an input $x \in \{0,1\}^n$ for which the scaled value $\prod_{i=0}^{\ell} \frac{m_i}{s_i} \cdot \val^{\ell+1}_1(x)$ is greater than the bias $\frac{1}{16} \cdot (\epsilon/2)^{\ell}\cdot \prod_{i=0}^{\ell} m_i$, in which case we reject, or there is no such input, in which case we accept.


We present two theorems that extend \cref{thm:ReLU-testing-0-OR-function-2-sided}
to multiple hidden layers.

Since the main theorem is heavily parameterized, we begin with a simpler statement about tester \textsc{AllZeroTesterMHL} for ReLU networks with the same number of nodes at every layer.

\MultipleLayers*

\junk{
\begin{theorem}
\label{thm:ReLU-testing-0-function-2-sided-mL-n}
Let $0 < \lambda < \frac{1}{\ell+1}$, $\frac{1}{n} < \epsilon < \frac12$, and $\delta \ge e^{-n/16} + e^{- (\epsilon/2)^{\ell} n /(342 (\ell+1)^2)}$.
Let $(W_0, \dots, W_{\ell})$ be a ReLU network with $n$ input nodes, $\ell \ge 1$ hidden layers with $n$ hidden layer nodes each, and a single output node. Further, suppose that
$n \ge 171 \cdot (\ell+1)^2 \cdot (2/\epsilon)^{2\ell} \cdot (\ln(2/\delta) + \ell \ln n)$.
Assuming that~$\ell$ is constant, tester \textsc{AllZeroTesterMHL} queries $\Theta(\epsilon^{-8\ell} \cdot \ln^4(1/\lambda\epsilon))$ entries from $W_0, \dots, W_{\ell}$, and
\begin{enumerate}[(i)]
\item rejects with probability at least $1 - \lambda$, if the ReLU network $(W_0, \dots, W_{\ell})$ is $(\epsilon,\delta)$-far from computing the constant $0$-function, and
\item accepts with probability at least $1 - \lambda$, if the ReLU network $(W_0, \dots, W_{\ell})$ computes the constant $0$-function.
\end{enumerate}
\end{theorem}
}


\cref{thm:ReLU-testing-0-function-2-sided-mL-n} follows from the following more general theorem.

\begin{theorem}
\label{thm:ReLU-testing-0-function-2-sided-mL}
Let $0 < \lambda < \frac{1}{\ell+1}$ and $\frac{1}{\min_{0 \le k \le \ell}\{m_k\}} < \epsilon < \frac12$.
Let \NE be a ReLU network $(W_0, \dots, W_{\ell})$ with $m_0$ input nodes, $\ell \ge 1$ hidden layers with $m_1, \dots, m_{\ell}$ hidden layer nodes each, and a single output node. Further, suppose that for an arbitrary parameter $0 < \probp < \frac{1}{\ell+1}$, it holds that $m_0 \ge 128 \cdot (\ell+1)^2 \cdot (2/\epsilon)^{2\ell} \cdot \ln(\frac{1}{\probp} \prod_{i=1}^{\ell} m_i)$ and that $m_k \ge 171 \cdot (\ell+1)^2 \cdot (2/\epsilon)^{2\ell} \cdot \ln(\frac{1}{\probp} \prod_{i=k+1}^{\ell} m_i)$ for every $1 \le k \le \ell$.
Let $\delta \ge \probp + e^{-m_0/16}$.
Assuming that $\ell$ is constant, tester \textsc{AllZeroTesterMHL} queries $\Theta(\epsilon^{-8\ell} \cdot \ln^4(1/\lambda\epsilon))$ entries from $W_0, \dots, W_{\ell}$, and
\begin{enumerate}[(i)]
\item rejects with probability at least $1 - \lambda$, if \NE is $(\epsilon,\delta)$-far from computing the constant $0$-function, and
\item accepts with probability at least $1 - \lambda$, if \NE computes the constant $0$-function.
\end{enumerate}
\end{theorem}


\begin{remark}\rm
\label{remark:thm:ReLU-testing-0-function-2-sided-mL}
\cref{thm:ReLU-testing-0-function-2-sided-mL} is heavily parameterized and we present it in a general form to on one hand demonstrate that one can study ReLU networks with arbitrary many hidden layers and to obtain for them efficient testers for some basic functions, and on the other hand, to make our result concrete. In many natural scenarios, and this is how one should read \cref{thm:ReLU-testing-0-function-2-sided-mL}, when the number of hidden layers is constant, then the tester works with each $s_k = \poly(1/\epsilon)$, having total number of entries of $W_0, \dots, W_{\ell}$ sampled to be polynomial in $\poly(1/\epsilon)$.

The dependency on $m_0, \dots, m_{\ell}$ is rather mild, and bar their dependency on $\probp$, it essentially requires that all $m_i$ are larger than some constant (depending on $\epsilon, \ell$, and $\lambda$) and the largest $m_i$ is not more than exponentially larger than any other $m_j$.
The dependency on $\probp$ and $\delta$ is discussed in \cref{remark:auxiliary-to-lemma:output-mL}. One would like to keep both $\probp$ and $\delta$ as small as possible and \cref{claim:auxiliary-to-lemma:output-mL} demonstrates that one can normally get $\delta$ as low as $e^{-m_0/16} + e^{-\Theta((\epsilon/2)^{2\ell}/\ell^2 \cdot \min m_k)}$; when all $m_k$ are the same, then $\probp$ and $\delta$ can be as small as $e^{-\Theta((\epsilon/2)^{2\ell}m_0/\ell^2)}$ and $e^{-m_0/16} + e^{-\Theta((\epsilon/2)^{2\ell}m_0/\ell^2)}$, respectively, as stated in \cref{thm:ReLU-testing-0-function-2-sided-mL-n}.
Observe that the lower bounds on $m_0, \dots, m_{\ell}$ put some mild constraints on~$\epsilon$.

We present \cref{thm:ReLU-testing-0-function-2-sided-mL} with the assumption that $\ell$ is a constant. However, our analysis below works for arbitrarily large $\ell$ and the use of constant $\ell$ is done (to simplify the calculations) only at the very final step. In the general case, we could define $s_0, \dots, s_{\ell}$ to satisfy the constraints that $s_k \ge 512 \cdot (\ell+1)^2 \cdot (2/\epsilon)^{2\ell} \cdot \ln\left(\frac{2 \prod_{i=k+1}^{\ell} s_i}{\lambda(\ell+1)}\right)$ for every $0 \le k \le \ell$, and that $s_k \ge 512 \cdot \ell^2 \cdot (2/\epsilon)^{2\ell} \cdot \ln\left(\frac{2^{s_0+1} \cdot \ell \cdot \prod_{i=k+1}^{\ell} s_i}{\lambda}\right)$ for every $1 \le k \le \ell$. Then \textsc{AllZeroTesterMHL} (with modified values of $s_0, \dots, s_{\ell}$) would be a tester with properties as in \cref{thm:ReLU-testing-0-function-2-sided-mL} and query complexity $\sum_{k=0}^{\ell-1} s_k \cdot s_{k+1}$.
\end{remark}


\subsubsection{Overview of the arguments used in the proof of \cref{thm:ReLU-testing-0-function-2-sided-mL}}
\Artur{If we had space, that text could fit well description in \cref{subsection:Overview-multiple-outputs+layers} in \cref{section:Overview}. But we have no space.}%
The proof of \cref{thm:ReLU-testing-0-function-2-sided-mL} follows the approach presented earlier in the proofs of \cref{thm:ReLU-testing-0-OR-function-2-sided}%
    
, though the need to deal with multiple hidden layers makes the analysis more tedious and complex. The approach relies on two steps of the analysis, first proving part
\begin{inparaenum}[(i)]
\item about ReLU networks that are $(\epsilon,\delta)$-far from computing the constant $0$-function (\cref{subsubsec:output-mL,subsubsec:ReLU-testing-0-function-2-sided-far-mL}) and then part
\item about ReLU networks that compute the constant $0$-function (\cref{subsubsec:ReLU-testing-0-function-2-sided-compute-mL}).
\end{inparaenum}

The analysis of part (i) is in four steps, as described in \cref{lemma:sampling-mL,lemma:sampling-mL-existance,lemma:output-mL,lem:ReLU-testing-0-function-2-sided-far-mL}.
We first show in \cref{subsubsec:sampling-mL} that the use of randomly sampled nodes in each layer of the network results in a modified ReLU network that, after scaling, for any fixed input returns the value close to the value returned by the original network.
Next, in \cref{subsubsec:sampling-mL-delta}, we extend this claim to show that there is always a modified ReLU network that, after scaling, for all but a small fraction of the inputs, returns the value close to the value returned by the original network.
We then study ReLU networks $(W_0, \dots, W_{\ell})$ that are $(\epsilon, \delta)$-far from computing the constant $0$-function.
In that case, we first show in \cref{subsubsec:output-mL} that there is always an input $x_0 \in \{0,1\}^{m_0}$ on which the $(\epsilon, \delta)$-far network returns a value which is (sufficiently) large.
Then, using the result from \ref{lemma:sampling-mL-existance} that the sampled network approximates the output of the original network, we will argue in \cref{subsubsec:ReLU-testing-0-function-2-sided-far-mL} that if $(W_0, \dots, W_{\ell})$ is $(\epsilon, \delta)$-far from computing the constant $0$-function then randomly sampled network on input $x_0$ returns a (sufficiently) large.

The study of part (ii), as described in \cref{lem:ReLU-testing-0-function-2-sided-compute-mL}, relies on a similar approach as that in part (i). We use the fact that the sampled network computes a function which on a fixed input is (after scaling) close to the original function on that input (\cref{lemma:sampling-mL}), and hence that value cannot be too large. An important difference with part (i) is that one must prove the claim to \emph{hold for all inputs}, making the arguments a bit more challenging (the most straightforward arguments would require the size of the sample to depend on $m_0$, giving a tester with super-constant query complexity).


\subsubsection{Random sampling concentration (extending \cref{lemma:sampling} to multiple layers)}
\label{subsubsec:sampling-mL}

We begin with an extension of \cref{lemma:sampling} to study the impact of random sampling of the nodes used in the network on multiple hidden layers ReLU networks.
\junk{
\cref{lemma:sampling} shows that in a ReLU with $n$ input nodes and $m$ hidden layer nodes, and a single output, if we sample $s \ge \frac{512 \ln(4t/\lambda)}{\epsilon^2}$ input nodes at random and $t \ge \frac{512 \ln(4/\lambda)}{\epsilon^2}$ hidden rows at random, defining sampling matrices $S$ and $T$, then for any fixed $x \in \{0,1\}^n$ we have that the probability that $\frac{mn}{st} \cdot w^T \cdot \relu(TASx)$ and $w^T \cdot \relu(Ax)$ differ by more than $\frac18\epsilon nm$ is upper bounded by $\lambda$. The following lemma generalizes \cref{lemma:sampling} to multiple layers (and also provides the omitted proof of \cref{lemma:sampling-2L} which was considering only two hidden layers).
}
It argues that on any fixed input $x \in \{0,1\}^{m_0}$, for random selection of nodes at each layer of the network $\sset_0, \dots, \sset_{\ell}$, the original ReLU network $(W_0, \dots, W_{\ell})$ returns the value $f_{\ell+1}(x)$ which is very close to the value $h_{\ell+1}(x)$ of the sampled network $(W_0 \cdot S_0, \dots, W_{\ell} \cdot S_{\ell})$ after appropriate scaling.

\begin{lemma}
\label{lemma:sampling-mL}
Let $(W_0, \dots, W_{\ell})$ be a ReLU network with $m_0$ input nodes, $\ell \ge 1$ hidden layers with $m_1, \dots, m_{\ell}$ hidden layer nodes each, and a single output node.
Let $f:\{0,1\}^{m_0} \rightarrow \{0,1\}$ be the function computed by ReLU network $(W_0, \dots, W_{\ell})$, which is $f(x) := \sgn(\relu(f_{\ell+1}(x)))$, where
\begin{align*}
    f_i(x) &:=
        \begin{cases}
            W_{i-1} \cdot \relu(f_{i-1}(x)) & \text{ if } 1 \le i \le \ell+1, \\
            x & \text{ if } i=0.
        \end{cases}
\end{align*}

Let $S_0, \dots, S_{\ell}$ be sampling matrices corresponding to samples $\sset_0, \dots, \sset_{\ell}$ of size $s_0, \dots, s_{\ell}$, respectively. Let $h :\{0,1\}^{m_0} \rightarrow \{0,1\}$ be the function computed by ReLU network $(W_0 \cdot S_0, \dots, W_{\ell} \cdot S_{\ell})$, that is, $h$ is defined recursively as $h(x) := \sgn(\relu(h_{\ell+1}(x)))$, where
\begin{align*}
    h_i(x) &:=
        \begin{cases}
            W_{i-1} \cdot S_{i-1} \cdot \relu(h_{i-1}(x)) & \text{ if } 1 \le i \le \ell+1, \\
            x & \text{ if } i=0.
        \end{cases}
\end{align*}
Let $0 < \vartheta \le 1$ and $0 < \lambda < \frac{1}{\ell+1}$.
Assume that $s_k \ge \frac{2 (\ell+1)^2 \ln(2 \prod_{i=k+1}^{\ell} s_i/(\lambda(\ell+1)))}{\vartheta^2}$ for every $0 \le k \le \ell$.
Then for any fixed $x \in \{0,1\}^{m_0}$ we have (where the probability is with respect to random choices of $S_0, \dots, S_{\ell}$)
\begin{align*}
    \Pr\left[\left| \prod_{i=0}^{\ell} \frac{m_i}{s_i} \cdot h_{\ell+1}(x) - f_{\ell+1}(x) \right|
        >
    \vartheta \cdot \prod_{i=0}^{\ell} m_i \right]
    &\le \lambda.
\end{align*}
\end{lemma}

\begin{proof}
Our proof extends the analysis from \cref{lemma:sampling}, though the analysis and the notation are more complicated because of the need of dealing with multiple layers.

In order to compare the value $f_{\ell+1}(x)$ of ReLU network $(W_0, \dots, W_{\ell})$ with the scaled value $h_{\ell+1}(x)$ of ReLU network $(W_0 \cdot S_0, \dots, W_{\ell} \cdot S_{\ell})$, we will compare the scaled values in $\ell$ different ReLU networks: network $(W_0, \dots, W_{\ell})$ and network $(W_0, \dots, W_{\ell-1}, W_{\ell} \cdot S_{\ell})$, then network $(W_0, \dots, W_{\ell-1}, W_{\ell} \cdot S_{\ell})$ and network $(W_0, \dots, W_{\ell-2}, W_{\ell-1} \cdot S_{\ell-1}, W_{\ell} \cdot S_{\ell})$, and so on so forth, until network $(W_0, W_1 \cdot S_1, \dots, W_{\ell} \cdot S_{\ell})$ and network $(W_0 \cdot S_0, \dots, W_{\ell} \cdot S_{\ell})$. Observe that if (after appropriate scaling) all these values are similar, then we will able to argue that network $(W_0, \dots, W_{\ell})$ has a very similar value as that of the scaled value of the sampling network $(W_0 \cdot S_0, \dots, W_{\ell} \cdot S_{\ell})$.

Our analysis of the difference between the values of $(W_0, \dots, W_r, W_{r+1} \cdot S_{r+1}, \dots , W_{\ell} \cdot S_{\ell})$ and the scaled value of $(W_0, \dots, W_{r-1}, W_r \cdot S_r, W_{r+1} \cdot S_{r+1}, \dots , W_{\ell} \cdot S_{\ell})$ is done in two steps: we first compare the values of $(W_0, \dots, W_{r-1}, W_r)$ and $(W_0, \dots, W_{r-1}, W_r \cdot S_r)$ using the \emph{randomness of $S_r$}, and only then we extend the analysis to the entire pair $(W_0, \dots, W_r, W_{r+1} \cdot S_{r+1}, \dots , W_{\ell} \cdot S_{\ell})$ and the scaled $(W_0, \dots, W_{r-1}, W_r \cdot S_r, W_{r+1} \cdot S_{r+1}, \dots , W_{\ell} \cdot S_{\ell})$, which works for \emph{arbitrary $S_{r+1}, \dots, S_{\ell}$}.

In order to facilitate the analysis of such networks, for any $0 \le k \le r \le \ell+1$, for any ReLU network $(W_k \cdot S_k, W_{k+1} \cdot S_{k+1}, \dots, W_{r-1} \cdot S_{r-1})$, we let $\phi_k^r(x)$ denote the function%
\footnote{$\phi_k^r(x)$ is the function computed by the ReLU network $(W_k \cdot S_k, \dots, W_{r-1} \cdot S_{r-1})$ and thus, in particular,
\begin{inparaenum}[(i)]
\item $\phi_0^r(x) = h_r(x)$ for every $0 \le r \le \ell+1$,
\item $\phi_k^{\ell+1}(f_k(x))$ is the function computed by $(W_0 \dots W_{k-1}, W_k \cdot S_k, \dots, W_{\ell} \cdot S_{\ell})$ for every $0 \le k \le \ell+1$, and
\item $\phi_{\ell+1}^{\ell+1}(f_{\ell+1}(x)) = f_{\ell+1}(x)$.
\end{inparaenum}
}%
\junk
{\footnote{
To provide some intuitions, observe also that
\begin{align*}
    \phi_k^k(x) &= x; \\
    \phi_k^{k+1}(x) &= W_k \cdot S_k \cdot \relu(x); \\
    \phi_k^{k+2}(x) &= W_{k+1} \cdot S_{k+1} \cdot \relu(W_k \cdot S_k \cdot x); \\
    \phi_k^{k+3}(x) &= W_{k+2} \cdot S_{k+2} \cdot \relu(W_{k+1} \cdot S_{k+1} \cdot \relu(W_k \cdot S_k \cdot x)); \\
    \phi_{\ell+1}^{\ell+1}(x) &= x; \\
    \phi_{\ell}^{\ell+1}(x) &= W_{\ell} \cdot S_{\ell} \cdot \relu(x); \\
    \phi_{\ell-1}^{\ell+1}(x) &= W_{\ell} \cdot S_{\ell} \cdot \relu(W_{\ell-1} \cdot S_{\ell-1} \cdot x); \\
    \phi_{\ell-2}^{\ell+1}(x) &= W_{\ell} \cdot S_{\ell} \cdot \relu(W_{\ell-1} \cdot S_{\ell-1} \cdot \relu(W_{\ell-2} \cdot S_{\ell-2} \cdot x)); \\
    \phi_0^r(x) &= h_r(x).
\end{align*}
}}%
computed by such network on input $x$. Observe that function $\phi_k^r: \RR^{m_k} \rightarrow \RR^{m_r}$ is defined recursively as follows:
\begin{align*}
    \phi_k^r(x) &:=
    \begin{cases}
        W_{r-1} \cdot S_{r-1} \cdot \relu(\phi_k^{r-1}(x)) & \text { if } k < r, \\
        x & \text { if } k = r.
    \end{cases}
\end{align*}


\paragraph{Difference of having a single random sampling matrix.}
We begin with the study of the relation between ReLU networks $(W_k)$ and $(W_k \cdot S_k)$ for \emph{random $S_k$}, with our later objective of using this to understand the relation between networks $(W_0, \dots, W_k)$ and $(W_0, \dots, W_{k-1}, W_k \cdot S_k)$.

We will prove the following claim, where we later will use with $y = \relu(f_k(x))$.

\begin{claim}
\label{claim:ineq-mL-aux1}
Let $0 \le k \le \ell$ and $1 \le j \le m_k$. Let $0 < \kappa, \lambda < 1$ and assume that $s_k \ge \frac{2 \ln(2/\lambda)}{\kappa^2}$. Let $y$ be an arbitrary $m_k$-vector in $\RR^{m_k}$. If we use $W_k^{(j)}$ to denote the $j$-th row of $W_k$, then the following inequality holds (where the probability is with respect to the random choice of $S_k$):
\begin{align}
\label{ineq-mL-aux1}
    \Pr\left[
            \left|\frac{m_k}{s_k} \cdot W_k^{(j)} \cdot S_k \cdot y
                -
                W_k^{(j)} \cdot y
            \right|
        > \kappa \cdot \|y\|_{\infty} \cdot m_k
        \right]
    &\le \lambda.
\end{align}
\end{claim}

Notice that (\ref{ineq-mL-aux1}) holds for any single node $j$ from $V_k$, but not for all nodes at the same time.

\begin{proof}
We study the relation between $W_k^{(j)} \cdot S_k \cdot y$ and $W_k^{(j)} \cdot y$.  Since $S_k$ is an $m_k \times m_k$ matrix corresponding to sampling of $s_k$ nodes from $V_k$ uniformly at random without replacement, each of $W_k^{(j)} \cdot S_k \cdot y$ and $W_k^{(j)} \cdot y$ is a sum of terms obtained by multiplying the $j$-th row of $W_k^{(j)} \cdot S_k$ and $W_k^{(j)}$, respectively, by the $j$-th column of vector $y$. The difference is that while the latter sum takes the entire column of $y$, the former samples $s_k$ random terms from that column. Hence we view $W_k^{(j)} \cdot S_k \cdot y$ as the sum of $s_k$ random terms defining $W_k^{(j)} \cdot y$; notice that $\Ex[W_k^{(j)} \cdot S_k \cdot y] = \frac{s_k}{m_k} \cdot W_k^{(j)} \cdot y$. Since every term in $W_k^{(j)} \cdot y$ is in $[-\|y\|_{\infty}, \|y\|_{\infty}]$, rescaling by $\frac{1}{\|y\|_{\infty}}$ ensures that every term in $\frac{1}{\|y\|_{\infty}} \cdot W_k^{(j)} \cdot y$ is in $[-1,1]$. Thus Hoeffding's bound (\cref{lemma:Hoeffding}) yields for any $\kappa$:
\begin{align*}
    &\Pr\left[
        \left|\frac{m_k}{s_k} \cdot W_k^{(j)} \cdot S_k \cdot y
            -
        W_k^{(j)} \cdot y \right|
        > \kappa \cdot \|y\|_{\infty} \cdot m_k \right]
        =
        \\\notag
    &\Pr\left[
        \left|\frac{1}{s_k \cdot \|y\|_{\infty}} \cdot W_k^{(j)} \cdot S_k \cdot y
            -
        \frac{1}{m_k \cdot \|y\|_{\infty}} \cdot W_k^{(j)} \cdot y \right|
        > \kappa \right]
    \le 2e^{-2 s_k \kappa^2/4}
    \le \lambda,
\end{align*}
where the last inequality follows from our assumption that $s_k \ge \frac{2 \ln(2/\lambda)}{\kappa^2}$.
\end{proof}


\paragraph{Multiplying by sparse matrices.}
We consider the impact of using arbitrary sampling matrices with a small number of non-zero entries on the functions computed by ReLU networks. We have the following claim about multiplying by arbitrary (not necessarily random) sampling matrices%
\footnote{Let us recall that $\phi_k^r(x)$ is the function computed by ReLU network $(W_k \cdot S_k, W_{k+1} \cdot S_{k+1}, \dots, W_{r-1} \cdot S_{r-1})$ on input $x$, defined as $\phi_k^r(x) = W_{r-1} \cdot S_{r-1} \cdot \relu(\phi_k^{r-1}(x))$ for $k < r$ and $\phi_k^r(x) = x$ for $k=r$.}.

\begin{claim}
\label{claim:ineq-mL-aux2}
Let $0 \le k \le \ell$. Let $y$ and $z$ be arbitrary random vectors in $\RR^{m_k}$ such that for some $\zeta, \eta$, for any $1 \le i \le m_k$ holds $\Pr[|y_i - z_i| > \zeta] \le \eta$.
For every $0 \le k < r \le \ell+1$, if we consider arbitrary sampling matrices $S_k, \dots, S_{r-1}$, each matrix $S_i$ being diagonal with $s_i$ non-zero entries,~then
\begin{align}
\label{ineq-mL-aux2}
    \Pr\Big[
        \left\|
            \phi_k^r(y) -
            \phi_k^r(z)
        \right\|_{\infty}
        > \zeta \cdot \prod_{i=k}^{r-1} s_i
    \Big]
        & \le
        \eta \cdot \prod_{i=k}^{r-1} s_i
    \enspace.
\end{align}
\end{claim}

\begin{proof}
Let us fix any $0 \le k \le \ell$ and now we will prove the claim by induction on $r - k \ge 1$.

For the basis of induction, when $r = k + 1$, observe that $\phi_k^{k+1}(y) = W_k \cdot S_k \cdot \relu(\phi_k^k(y)) = W_k \cdot S_k \cdot \relu(y)$ and $\phi_k^{k+1}(z) = W_k \cdot S_k \cdot \relu(z)$, and therefore let us consider
\begin{align*}
    \left\|
        W_k \cdot S_k \cdot (y - z)
    \right\|_{\infty}
        &=
    \left\|
        W_k \cdot S_k \cdot y - W_k \cdot S_k \cdot z
    \right\|_{\infty}
    \enspace.
\end{align*}
Observe that the column vector $S_k \cdot (y - z)$ is equal to $y-z$ on exactly $s_k$ entries (corresponding to the non-zero rows in $S_k$) and is $0$ otherwise. If we take a union bound over the contribution of these $s_k$ entries and observe that for the remaining rows are always $0$, then we obtain:
\begin{align*}
    \Pr\left[
        \left\|
            S_k \cdot (y - z)
        \right\|_{\infty}
        > \zeta
    \right]
        &\le
        s_k \cdot \eta
    \enspace.
\end{align*}
If we multiply the terms in the bound above by $W_k$, then since each entry in $W_k$ is in $[-1,1]$ and since each row in $S_k \cdot (y - z)$ has at most $s_k$ non-zero entries, we have
\begin{align*}
    \left\|
        W_k \cdot S_k \cdot (y - z)
    \right\|_{\infty}
        &\le
    s_k \cdot
    \left\|
        S_k \cdot (y - z)
    \right\|_{\infty}
    \enspace.
\end{align*}
We can combine this bound with our probability bound above to obtain,
\begin{align}
\label{ineq-mL-aux3}
    \Pr\left[
        \left\|
            W_k \cdot S_k \cdot (y - z)
        \right\|_{\infty}
        > s_k \cdot \zeta
    \right]
        &\le
        s_k \cdot \eta
    \enspace.
\end{align}
Next, let us observe that the bounds stay unchanged if we apply the ReLU activation function to the terms. In particular, firstly notice that
\begin{align*}
    \left\|
        W_k \cdot S_k \cdot \relu(y)
            -
        W_k \cdot S_k \cdot \relu(z)
    \right\|_{\infty}
        &\le
    \left\|
        W_k \cdot S_k \cdot y
            -
        W_k \cdot S_k \cdot z
    \right\|_{\infty}
\end{align*}
and therefore we can combine this bound with (\ref{ineq-mL-aux3}) to obtain the following:
\begin{align}
\label{ineq-mL-aux4}
    \Pr\left[
        \left\|
            W_k \cdot S_k \cdot \relu(y) - W_k \cdot S_k \cdot \relu(z)
        \right\|_{\infty}
        > s_k \cdot \zeta
    \right]
        &\le
        s_k \cdot \eta
    \enspace,
\end{align}
which is exactly (\ref{ineq-mL-aux2}), completing the analysis for the base of induction $r = k+1$.

Next, let us consider $0 \le k \le r \le \ell$ and assume, by induction, that (\ref{ineq-mL-aux2}) holds for $k$ and $r$; we will then show that (\ref{ineq-mL-aux2}) holds for $k$ and $r+1$.
Observe that $\phi_k^{r+1}(y) = W_r \cdot S_r \cdot \relu(\phi_k^r(y))$ and $\phi_k^{r+1}(z) = W_r \cdot S_r \cdot \relu(\phi_k^r(z))$. We will apply (\ref{ineq-mL-aux4}) with new $y' := \phi_k^r(y)$ and $z' := \phi_k^r(z)$. By induction, $\Pr[\|\phi_k^r(y) - \phi_k^r(z)\|_{\infty} > \zeta \cdot \prod_{i=k}^{r-1} s_i] \le \eta \cdot \prod_{i=k}^{r-1} s_i$, and hence we will use new $\zeta' := \zeta \cdot \prod_{i=k}^{r-1} s_i$ and $\eta' := \eta \cdot \prod_{i=k}^{r-1} s_i$. Then, (\ref{ineq-mL-aux4}) implies the following:
\begin{eqnarray*}
    \lefteqn{
    \Pr\left[
        \left\|
            \phi_k^{r+1}(y) -
            \phi_k^{r+1}(z)
        \right\|_{\infty}
        > s_r \cdot \left(\zeta \cdot \prod_{i=k}^{r-1} s_i\right)
    \right]
        =}
        \\
    &&
    \Pr\Big[
        \left\|
            W_r \cdot S_r \cdot \relu(\phi_k^r(y)) -
            W_r \cdot S_r \cdot \relu(\phi_k^r(z))
        \right\|_{\infty}
        > s_r \cdot \left(\zeta \cdot \prod_{i=k}^{r-1} s_i\right)
    \Big]
        \le
        \\
    &&
    s_r \cdot \left(\eta \cdot \prod_{i=k}^{r-1} s_i\right)
    \enspace,
\end{eqnarray*}
which is exactly the bound in  (\ref{ineq-mL-aux2}). This completes the proof of \cref{claim:ineq-mL-aux2}.
\end{proof}


\paragraph{Completing the proof of \cref{lemma:sampling-mL}.}
In what follows, we will consider an arbitrary \emph{fixed} $x \in \{0,1\}^{m_0}$.
In order to prove \cref{lemma:sampling-mL}, we study $\prod_{i=0}^{\ell} \frac{m_i}{s_i} \cdot h_{\ell+1}(x) - f_{\ell+1}(x)$.
Observe that the definition of $\phi$ implies that $h_{\ell+1}(x) = \phi_0^{\ell+1}(x)$ and $f_{\ell+1}(x) = \phi_{\ell+1}^{\ell+1}(f_{\ell+1}(x))$, and hence,
\begin{align*}
    \prod_{i=0}^{\ell} \frac{m_i}{s_i} \cdot h_{\ell+1}(x) - f_{\ell+1}(x) &=
    \prod_{i=0}^{\ell} \frac{m_i}{s_i} \cdot \phi_0^{\ell+1}(f_0(x)) - \phi_{\ell+1}^{\ell+1}(f_{\ell+1}(x))
    \enspace.
\end{align*}
Next, using that $\prod_{i=\ell+1}^{\ell} \frac{m_i}{s_i} = 1$, we can use telescoping sum to obtain
\junk
{\Artur{We have the following:
\begin{align*}
    \sum_{k=0}^{\ell} \left(
        \prod_{i=k}^{\ell} \frac{m_i}{s_i} \cdot \phi_k^{\ell+1}(f_k(x)) -
        \prod_{i=k+1}^{\ell} \frac{m_i}{s_i} \cdot \phi_{k+1}^{\ell+1}(f_{k+1}(x))
    \right)
        &=
    \sum_{k=0}^{\ell} \prod_{i=k}^{\ell} \frac{m_i}{s_i} \cdot \phi_k^{\ell+1}(f_k(x)) -
    \sum_{k=0}^{\ell} \prod_{i=k+1}^{\ell} \frac{m_i}{s_i} \cdot \phi_{k+1}^{\ell+1}(f_{k+1}(x))
        \\&=
    \sum_{k=0}^{\ell} \prod_{i=k}^{\ell} \frac{m_i}{s_i} \cdot \phi_k^{\ell+1}(f_k(x)) -
    \sum_{k=1}^{\ell+1} \prod_{i=k}^{\ell} \frac{m_i}{s_i} \cdot \phi_k^{\ell+1}(f_k(x))
        \\&=
    \prod_{i=0}^{\ell} \frac{m_i}{s_i} \cdot \phi_0^{\ell+1}(f_k(x)) -
    \prod_{i=\ell+1}^{\ell} \frac{m_i}{s_i} \cdot \phi_{\ell+1}^{\ell+1}(f_k(x))
        \\&=
    \prod_{i=0}^{\ell} \frac{m_i}{s_i} \cdot \phi_0^{\ell+1}(f_k(x)) -
    \phi_{\ell+1}^{\ell+1}(f_{\ell+1}(x))
    \enspace.
\end{align*}
}}
\begin{align*}
    \prod_{i=0}^{\ell} \frac{m_i}{s_i} \cdot h_{\ell+1}(x) - f_{\ell+1}(x) &=
    \sum_{k=0}^{\ell} \left(
        \prod_{i=k}^{\ell} \frac{m_i}{s_i} \cdot \phi_k^{\ell+1}(f_k(x)) -
        \prod_{i=k+1}^{\ell} \frac{m_i}{s_i} \cdot \phi_{k+1}^{\ell+1}(f_{k+1}(x))
    \right)
    \enspace,
\end{align*}
and hence
\begin{align}
    \notag
    \left|
        \prod_{i=0}^{\ell} \frac{m_i}{s_i} \cdot h_{\ell+1}(x) - f_{\ell+1}(x)
    \right|
        &\le
    \sum_{k=0}^{\ell}
        \left|
            \prod_{i=k}^{\ell} \frac{m_i}{s_i} \cdot \phi_k^{\ell+1}(f_k(x)) -
            \prod_{i=k+1}^{\ell} \frac{m_i}{s_i} \cdot \phi_{k+1}^{\ell+1}(f_{k+1}(x))
    \right|
        \\
        \label{ineq-mL-aux6}
        &=
    \sum_{k=0}^{\ell}
        \left(\prod_{i=k+1}^{\ell} \frac{m_i}{s_i} \cdot
        \left|
            \frac{m_k}{s_k} \cdot \phi_k^{\ell+1}(f_k(x)) -
            \phi_{k+1}^{\ell+1}(f_{k+1}(x))
        \right|
        \right)
    \enspace.
\end{align}

In view of (\ref{ineq-mL-aux6}), we want to bound the probability that $|\frac{m_k}{s_k} \cdot \phi_k^{\ell+1}(f_k(x)) - \phi_{k+1}^{\ell+1}(f_{k+1}(x))|$ is large for all $0 \le k \le \ell$; for that, we will use \cref{claim:ineq-mL-aux2}.
In order to incorporate \cref{claim:ineq-mL-aux1}, we will use \cref{claim:ineq-mL-aux2} with $y = \frac{m_k}{s_k} \cdot W_k \cdot S_k \cdot \relu(f_k(x))$ and $z = W_k \cdot \relu(f_k(x))$.
Observe that $\|\relu(f_k(x))\|_{\infty} \le \prod_{i=0}^{k-1} m_i$.
Therefore \cref{claim:ineq-mL-aux1} implies that for any $0 \le k \le \ell$ and $1 \le j \le m_k$, for any $0 < \kappa_k, \lambda_k < 1$ with $s_k \ge \frac{2 \ln(2/\lambda_k)}{\kappa_k^2}$, we have the following:
\begin{align*}
    \Pr\left[
        \left|
            \frac{m_k}{s_k} \cdot W_k^{(j)} \cdot S_k \cdot \relu(f_k(x)) -
            W_k^{(j)} \cdot \relu(f_k(x))
        \right|
        > \kappa_k \cdot \prod_{i=0}^k m_i
    \right]
    &\le
    \lambda_k
    \enspace.
\end{align*}
\junk{
\begin{align*}
    \Pr\left[
        \left\|
            \frac{1}{s_k} \cdot \relu\left(W_k \cdot S_k \cdot \relu(f_k(x)) \right) -
            \frac{1}{m_k} \cdot \relu\left(W_k \cdot \relu(f_k(x)) \right)
        \right\|_{\infty}
        > \kappa_k \cdot \prod_{i=0}^k m_i
    \right]
    &\le
    \lambda_k
    \enspace.
\end{align*}
}

We combine this with \cref{claim:ineq-mL-aux2} for $\phi_{k+1}^{\ell+1}(\cdot)$ with $y = \frac{m_k}{s_k} \cdot W_k \cdot S_k \cdot \relu(f_k(x))$, $z = W_k \cdot \relu(f_k(x))$, $\zeta = \kappa_k \cdot \prod_{i=0}^k m_i$, and $\eta = \lambda_k$. Since with our notation $\phi_{k+1}^{\ell+1}(W_k S_k \relu(f_k(x))) = \phi_k^{\ell+1}(f_k(x))$ and $\phi_{k+1}^{\ell+1}(W_k \relu(f_k(x))) = \phi_{k+1}^{\ell+1}(f_{k+1}(x))$, by \cref{claim:ineq-mL-aux2} we obtain the following:
\begin{eqnarray}
    \notag
    \lefteqn{
    \Pr\left[
        \left|
            \frac{m_k}{s_k} \cdot \phi_k^{\ell+1}(f_k(x)) -
            \phi_{k+1}^{\ell+1}(f_{k+1}(x))
        \right|
        > \kappa_k \cdot \prod_{i=0}^k m_i \cdot \prod_{i=k+1}^{\ell} s_i
    \right]
        =
    }
        \\
    \notag
        &&
    \Pr\left[
        \left|
            \phi_{k+1}^{\ell+1}\left(\frac{m_k}{s_k} \cdot W_k \cdot S_k \cdot \relu(f_k(x))\right) -
            \phi_{k+1}^{\ell+1}(W_k \cdot \relu(f_k(x)))
        \right|
        > \kappa_k \cdot \prod_{i=0}^k m_i \cdot \prod_{i=k+1}^{\ell} s_i
    \right]
        \\
        \label{ineq-mL-aux7}
        && =
    \Pr\left[
        \left|
            \phi_{k+1}^{\ell+1}(y) -
            \phi_{k+1}^{\ell+1}(z)
        \right|
        > \zeta \cdot \prod_{i=k+1}^{\ell} s_i
    \right]
        \le
        \lambda_k \cdot \prod_{i=k+1}^{\ell} s_i
    \enspace.
\end{eqnarray}
(Since the values of $\phi_{k+1}^{\ell+1}$ are real, we replaced the $\|\cdot\|_{\infty}$-norm from \cref{claim:ineq-mL-aux2} by $|\cdot|$, and since $\phi_{k+1}^{\ell+1}$ is homogenous, we have $\frac{m_k}{s_k} \cdot \phi_{k+1}^{\ell+1}(W_k \cdot S_k \cdot \relu(f_k(x))) = \phi_{k+1}^{\ell+1}(\frac{m_k}{s_k} \cdot W_k \cdot S_k \cdot \relu(f_k(x)))$.)

Let $\kappa_0, \dots, \kappa_{\ell}, \lambda_0, \dots, \lambda_{\ell} > 0$ and assume that $s_k \ge \frac{2 \ln(2/\lambda_k)}{\kappa_k^2}$ for every $0 \le k \le \ell$. Now, we plug the bound from (\ref{ineq-mL-aux7}) into (\ref{ineq-mL-aux6}), and after applying union bound with $\vartheta = \kappa_0 + \dots + \kappa_k$, we obtain
\begin{eqnarray*}
    \lefteqn{
    \Pr\left[\left| \prod_{i=0}^{\ell} \frac{m_i}{s_i} \cdot h_{\ell+1}(x) - f_{\ell+1}(x) \right|
        >
    \vartheta \cdot \prod_{i=0}^{\ell} m_i \right]
        \le
    }
        \\&&
    \sum_{k=0}^{\ell}
        \Pr\left[
            \left|
                \frac{m_k}{s_k} \cdot \phi_k^{\ell+1}(f_k(x)) -
                \phi_{k+1}^{\ell+1}(f_{k+1}(x))
            \right|
                >
            \kappa_k \cdot \prod_{i=0}^{\ell} m_i \cdot \prod_{i=k+1}^{\ell} \frac{s_i}{m_i}
        \right]
        =
        \\&&
    \sum_{k=0}^{\ell}
        \Pr\left[
            \left|
                \frac{m_k}{s_k} \cdot \phi_k^{\ell+1}(f_k(x)) -
                \phi_{k+1}^{\ell+1}(f_{k+1}(x))
            \right|
                >
            \kappa_k \cdot \prod_{i=0}^k m_i \cdot \prod_{i=k+1}^{\ell} s_i
        \right]
        \le
        \\&&
    \sum_{k=0}^{\ell}
        \left(\lambda_k \cdot \prod_{i=k+1}^{\ell} s_i\right)
    \enspace.
\end{eqnarray*}
Following the approach for single hidden layer ReLU networks, we take all $\kappa_k$ to be the same and all terms $\lambda_k \cdot \prod_{i=k+1}^{\ell} s_i$ to be the same. That is, for any $\vartheta, \lambda > 0$, for every $0 \le k \le \ell$, we set
$\kappa_k := \frac{\vartheta}{\ell+1}$ and $\lambda_k := \frac{\lambda (\ell+1)}{\prod_{i=k+1}^{\ell} s_i}$ with constraint $s_k \ge
\frac{2 \ln(2/\lambda_k)}{\kappa_k^2} = \frac{2 (\ell+1)^2 \ln(2 \prod_{i=k+1}^{\ell} s_i/(\lambda(\ell+1)))}{\vartheta^2}$, to get
\begin{align*}
    \Pr\left[\left| \prod_{i=0}^{\ell} \frac{m_i}{s_i} \cdot h_{\ell+1}(x) - f_{\ell+1}(x) \right|
        >
    \vartheta \cdot \prod_{i=0}^{\ell} m_i \right]
        &\le
    \lambda \enspace.
    \qedhere
\end{align*}
\end{proof}


\subsubsection{Extending \cref{lemma:sampling-mL} to hold (\emph{existentially}) for many inputs}
\label{subsubsec:sampling-mL-delta}


\cref{lemma:sampling-mL} shows that for any single input, the sub-network of any ReLU network with randomly selected nodes in each layer (with further conditions specified in the lemma) will return a scaled value which is very similar to the value returned by the original network. In this section, we transform that claim to hold for \emph{most of the inputs at the same time}, by showing an auxiliary \emph{existential} claim that there is at least one choice of $S_0, \dots, S_{\ell}$ such that for at least a $(1-\lambda)$-fraction of the inputs 
a scaled value is very similar to the value returned by the original network.

\begin{lemma}
\label{lemma:sampling-mL-existance}
Let $(W_0, \dots, W_{\ell})$ be a ReLU network with $m_0$ input nodes, $\ell \ge 1$ hidden layers with $m_1, \dots, m_{\ell}$ hidden layer nodes each, and a single output node.
Let $f:\{0,1\}^{m_0} \rightarrow \{0,1\}$ be the function computed by ReLU network $(W_0, \dots, W_{\ell})$, which is $f(x) := \sgn(\relu(f_{\ell+1}(x)))$, where
\begin{align*}
    f_i(x) &:=
        \begin{cases}
            W_{i-1} \cdot \relu(f_{i-1}(x)) & \text{ if } 1 \le i \le \ell+1, \\
            x & \text{ if } i=0.
        \end{cases}
\end{align*}
For any sampling matrices $S_0, \dots, S_{\ell}$ corresponding to samples $\sset_0, \dots, \sset_{\ell}$ of size $s_0, \dots, s_{\ell}$, respectively, let $h :\{0,1\}^{m_0} \rightarrow \{0,1\}$ be the function computed by ReLU network $(W_0 \cdot S_0, \dots, W_{\ell} \cdot S_{\ell})$, that is, $h$ is defined recursively as $h(x) := \sgn(\relu(h_{\ell+1}(x)))$, where
\begin{align*}
    h_i(x) &:=
        \begin{cases}
            W_{i-1} \cdot S_{i-1} \cdot \relu(h_{i-1}(x)) & \text{ if } 1 \le i \le \ell+1, \\
            x & \text{ if } i=0.
        \end{cases}
\end{align*}
Let $0 < \vartheta \le 1$ and $0 < \lambda < \frac{1}{\ell+1}$. Assume that $s_k \ge \frac{2 (\ell+1)^2 \ln(2 \prod_{i=k+1}^{\ell} s_i/(\lambda(\ell+1)))}{\vartheta^2}$ for every $0 \le k \le \ell$. Then there is at least one choice of sampling matrices $S_0, \dots, S_{\ell}$ such that for at most a $\lambda$-fraction of the inputs $x \in \{0,1\}^{m_0}$ it holds the following:
\begin{align*}
    \left| \prod_{i=0}^{\ell} \frac{m_i}{s_i} \cdot h_{\ell+1}(x) - f_{\ell+1}(x) \right|
        &>
    \vartheta \cdot \prod_{i=0}^{\ell} m_i
    \enspace.
\end{align*}
\end{lemma}

\begin{proof}
Let $\IN := \{0,1\}^{m_0}$ be the set of all input instances.
Let $\mathcal{S}$ be the set of all possible selections of the nodes in sets $\sset_0, \dots, \sset_{\ell}$, that is, $\mathcal{S} := \{\langle \sset_0, \dots, \sset_{\ell}\rangle: \sset_i \subseteq V_i  \wedge |\sset_i| = s_i \text{ for all } 0 \le i \le \ell\}$.

Next, we introduce some \emph{bad} sets and pairs, and \emph{obstructive} tuples. Let us call tuple $\langle \sset_0, \dots, \sset_{\ell}\rangle \in \mathcal{S}$ to be \emph{bad for input $x_0 \in \IN$} if after fixing sets $\sset_0, \dots, \sset_{\ell}$ to define function $h(x_0)$, we have
\begin{align*}
    \left| \prod_{i=0}^{\ell} \frac{m_i}{s_i} \cdot h_{\ell+1}(x_0) - f_{\ell+1}(x_0) \right|
        >
    \vartheta \cdot \prod_{i=0}^{\ell} m_i
    \enspace.
\end{align*}
For any $x \in \IN$ and $\mathfrak{s} \in \mathcal{S}$, we call pair $\langle x, \mathfrak{s} \rangle$ \emph{bad} if $\mathfrak{s}$ is bad for input $x$. Let us call a tuple $\mathfrak{s} \in \mathcal{S}$ \emph{$\lambda$-obstructive} if for more than a $\lambda$-fraction of inputs $x \in \IN$, pairs $\langle x, \mathfrak{s} \rangle$ are bad.

Observe that one way of reading \cref{lemma:sampling-mL} is that (under the assumptions of \cref{lemma:sampling-mL}) for every $x \in \IN$, at most a $\lambda$ fraction of the tuples in $\mathcal{S}$ are bad, that is, at most $\lambda \cdot |\mathcal{S}|$ tuples are bad. This immediately implies that at most $\lambda \cdot |\IN| \cdot |\mathcal{S}|$ pairs $\langle x, \mathfrak{s} \rangle$ with $x \in \IN$ and $\mathfrak{s} \in \mathcal{S}$ are bad.

Next, observe that if there are $t$ tuples in $\mathcal{S}$ that are $\lambda$-obstructive, then there are more than $t \cdot \lambda \cdot |\IN|$ pairs $\langle x, \mathfrak{s} \rangle$ with $x \in \IN$ and $\mathfrak{s} \in \mathcal{S}$ that are bad. However, since there are at most $\lambda \cdot |\IN| \cdot |\mathcal{S}|$ pairs $\langle x, \mathfrak{s} \rangle$ with $x \in \IN$ and $\mathfrak{s} \in \mathcal{S}$ that are bad, we must have $t \cdot \lambda \cdot |\IN| < \lambda \cdot |\IN| \cdot |\mathcal{S}|$, and thus $t < |\mathcal{S}|$. This means that there are less than $|\mathcal{S}|$ tuples $\mathfrak{s} \in \mathcal{S}$ that are $\lambda$-obstructive, or equivalently, there is at least one tuple $\mathfrak{s} \in \mathcal{S}$ that is \emph{not} $\lambda$-obstructive. This completes the proof.
\end{proof}


\subsubsection{If $(\epsilon, \delta)$-far then there is a bad input (extending \cref{lemma:output} to multiple layers)}
\label{subsubsec:output-mL}

To analyze the correctness of algorithm \textsc{AllZeroTesterMHL}, as in \cref{sec:ReLU-testing-0-OR-function-2-sided}, we prove a structural lemma (cf. \cref{lemma:output}), which argues that if a ReLU network is $(\epsilon, \delta)$-far from computing the constant $0$-function then there is an input $x$ on which the network computes a large value. (Observe that \cref{lemma:output-mL} is existential, and unlike our other arguments later in \cref{subsubsec:ReLU-testing-0-function-2-sided-far-mL} and \cref{subsubsec:ReLU-testing-0-function-2-sided-compute-mL}, it is independent of \textsc{AllZeroTester2HL} and of the sampling matrices $S_0, \dots, S_{\ell}$).

\begin{lemma}
\label{lemma:output-mL}
Let $(W_0, \dots, W_{\ell})$ be a ReLU network with $m_0$ input nodes, $\ell \ge 1$ hidden layers with $m_1, \dots, m_{\ell}$ hidden layer nodes each, and a single output node.
Let $0 < \epsilon < \frac12$.
Further, suppose that for an arbitrary parameter $0 < \probp < \frac{1}{\ell+1}$, it holds that $m_0 \ge 128 (\ell+1)^2 (2/\epsilon)^{2\ell} \ln(\frac{1}{\probp} \prod_{i=1}^{\ell} m_i)$ and that for every $1 \le k \le \ell$ holds $m_k \ge 171 (\ell+1)^2 (2/\epsilon)^{2\ell} \ln(\frac{1}{\probp} \prod_{i=k+1}^{\ell} m_i)$.
Let $\delta \ge \probp + e^{-m_0/16}$.

Let $f:\{0,1\}^{m_0} \rightarrow \{0,1\}$ be the function computed by ReLU network $(W_0, \dots, W_{\ell})$, which is defined recursively as $f(x) := \sgn(\relu(f_{\ell+1}(x)))$, where
\begin{align*}
    f_i(x) &:=
        \begin{cases}
            W_{i-1} \cdot \relu(f_{i-1}(x)) & \text{ if } 1 \le i \le \ell, \\
            x & \text{ if } i=0.
        \end{cases}
\end{align*}
If $(W_0, \dots, W_{\ell})$ is $(\epsilon, \delta)$-far from computing the constant $0$-function then there exists an input $x \in \{0,1\}^{m_0}$ such that
\begin{align*}
    f_{\ell+1}(x) & > \frac18 \cdot (\epsilon/2)^{\ell} \cdot \prod_{i=0}^{\ell} m_i
    \enspace.
\end{align*}
\end{lemma}

\begin{remark}\rm
    \Artur{Do we need both \cref{remark:auxiliary-to-lemma:output-mL-short} and \cref{remark:auxiliary-to-lemma:output-mL}?}
\label{remark:auxiliary-to-lemma:output-mL-short}
The auxiliary parameter $\probp$ in our analysis has two uses: on one hand, it puts some constraints on $m_0, \dots, m_{\ell}$ and on the other hand, it allows us to parameterize $\delta$. In general, one would hope for $\probp$ to be close to $2^{-\Theta(m_0)}$, and in particular, the arguments used in \cref{claim:auxiliary-to-lemma:output-mL} allow us to obtain $\probp$ as low as $e^{-\Theta((\epsilon/2)^{2\ell} \cdot \min_{0 \le k \le \ell}\{m_k\}/\ell^2)}$; for example, if $m_0 = \dots = m_{\ell}$, then we have $\probp \ge e^{-\Theta(m_0 \cdot (\epsilon/2)^{2\ell^2}/\ell^2)}$.
\end{remark}

\begin{proof}[Proof of \cref{lemma:output-mL}]
While some high-level arguments used in the proof can be seen as an extensions of the proof of \cref{lemma:output}(1), rather surprisingly, some of the key arguments used in the
    proof of \cref{lemma:output}
are insufficient for the proof of \cref{lemma:output-mL} for $\ell \ge 3$, and so some new arguments are required. In particular, our analysis relies on \cref{lemma:sampling-mL-existance}.

For the purpose of contradiction, we will assume that $(W_0, \dots, W_{\ell})$ is $(\epsilon, \delta)$-far from computing the constant $0$-function and that for every input $x \in \{0,1\}^{m_0}$ the output value is at most $\tfrac12 \cdot (\epsilon/2)^{\ell} \cdot \prod_{i=0}^{\ell} m_i$, and then we will show that the ReLU network is $(\epsilon, \delta)$-close to computing the constant $0$-function, yielding contradiction.

Our approach is to modify the network in at most an $\epsilon$-fraction of edge weights at every layer, by splitting it into two separate parts by selecting some sets of nodes $\IN_1, \dots, \IN_{\ell}$ at hidden layers and then considering the ReLU network induced by these nodes and another network which ignores these nodes. We will be able to ensure that the network induced by the appropriately selected nodes contributes to a very low value, and the network induced by the remaining nodes contributes to a value close to the value returned by the original network. Since such a construction can be done using changes of at most an $\epsilon$-fraction of edge weights at every layer, and since the original network is $(\epsilon, \delta)$-far from computing the constant $0$-function, we will be able to argue that for some input $x$, the value returned by the original network is larger than $\tfrac18 \cdot (\epsilon/2)^{\ell} \cdot \prod_{i=0}^{\ell} m_i$.

Let us first introduce some useful notation (see also Figure~\ref{fig:ReLU-network-mL}).
Let $V_0$ be the set of input nodes, let $V_k$ the set of hidden layer nodes at layer $1 \le k \le \ell$, and let $V_{\ell+1}$ be the set of output nodes; if we set $m_{\ell+1} := 1$ then for every $0 \le k \le \ell+1$, we have $|V_k| = m_k$.
Let $E_0$ denote the set of edges between the input nodes and the first hidden layer, for any $1 \le k < \ell$, let $E_k$ denote the set of edges between the $k$-th hidden layer and the $(k+1)$-st hidden layer, and $E_{\ell}$ denote the set of edges between the $\ell$-th hidden layer and the output node.

Next, we will use \cref{lemma:sampling-mL-existance}. 
Let $s^*_0 := m_0$ and $s^*_k := (1-\frac12 \epsilon) m_k$ for $1 \le k \le \ell$. Our initial conditions of \cref{lemma:output-mL} ensure that with setting $\vartheta := \frac18 \cdot (\epsilon/2)^{\ell}$, we have $s^*_k \ge \frac{2 (\ell+1)^2 \ln(2 \prod_{i=k+1}^{\ell} s^*_i/(\probp(\ell+1)))}{\vartheta^2}$ for every $0 \le k \le \ell$.%
\footnote{Observe that conditions  $s^*_k \ge \frac{2 (\ell+1)^2 \ln(2 \prod_{i=k+1}^{\ell} s^*_i/(\probp(\ell+1)))}{\vartheta^2}$ for every $0 \le k \le \ell$ imply with our setting of $s^*_0, \dots, s^*_{\ell}$, the following constraints on $m_0, \dots, m_{\ell}$: $m_0 \ge \frac{2 (\ell+1)^2 \ln(2 (1-\frac12 \epsilon)^{\ell} \prod_{i=1}^{\ell} m_i/(\probp(\ell+1)))}{\vartheta^2}$ and that for $1 \le k \le \ell$ holds $m_k \ge \frac{2 (\ell+1)^2 \ln(2 (1-\frac12 \epsilon)^{\ell-k} \prod_{i=k+1}^{\ell} m_i/(\probp(\ell+1)))}{(1-\frac12 \epsilon) \vartheta^2}$. And thus our setting $\vartheta := \frac18 \cdot (\epsilon/2)^{\ell}$ combined with the lemma's assumptions that $m_0 \ge 128 (\ell+1)^2 (2/\epsilon)^{2\ell} \ln(\frac{1}{\probp} \prod_{i=1}^{\ell} m_i)$ and that for every $1 \le k \le \ell$ holds
$m_k \ge 171 (\ell+1)^2 (2/\epsilon)^{2\ell} \ln(\frac{1}{\probp} \prod_{i=k+1}^{\ell} m_i)$, imply these conditions.}
Therefore, by \cref{lemma:sampling-mL-existance}, there is a sequence of sets $\sset^*_0, \dots, \sset^*_{\ell}$ such that
\begin{itemize}
\item $s_k = |\sset^*_k|$ for every $0 \le k \le \ell$, and
\item if we take sampling matrices $S^*_0, \dots, S^*_{\ell}$ corresponding to the nodes from $\sset^*_0, \sset^*_1, \dots, \sset^*_{\ell}$, then for the function from $\{0,1\}^{m_0}$ to $\RR^{m_i}$:
    \begin{align*}
        h^*_i(x) :=
            \begin{cases}
                W_{i-1} \cdot S^*_{i-1} \cdot \relu(h^*_{i-1}(x)) & \text{ if } 1 \le i \le \ell+1, \\
                x & \text{ if } i=0,
            \end{cases}
    \end{align*}
for any $0 < \vartheta \le 1$, $0 < \delta < \frac{1}{\ell+1}$, for at most a $\probp$-fraction of the inputs $x \in \{0,1\}^{m_0}$ we have,
\begin{align*}
    \left| \prod_{i=0}^{\ell} \frac{m_i}{s^*_i} \cdot h^*_{\ell+1}(x) - f_{\ell+1}(x) \right|
        &>
    \vartheta \cdot \prod_{i=0}^{\ell} m_i
    \enspace.
\end{align*}
Since $s^*_0 := m_0$ and $s^*_k := (1-\frac12 \epsilon) m_k$ for $1 \le k \le \ell$, we can simplify this bound to obtain
\begin{align}
\label{ineq:combo-lemma:output-mL+lemma:sampling-mL-existance}
    \left|h^*_{\ell+1}(x) - (1-\tfrac12 \epsilon)^{\ell} \cdot f_{\ell+1}(x) \right|
        &>
    \vartheta \cdot (1-\tfrac12 \epsilon)^{\ell} \cdot \prod_{i=0}^{\ell} m_i
    \enspace.
\end{align}
\end{itemize}
Let $\IN_k := V_k \setminus \sset^*_k$ for every $1 \le k \le \ell$. Now, we will modify some edge weights in the ReLU network.

\paragraph{Construction of the modified ReLU network.}
\begin{itemize}
\item For every $1 \le k \le \ell$, we change to $+1$ the weights of all edges connecting the nodes from $V_{k-1}$ to the nodes from $\IN_k$;
    \begin{itemize}[$\blacktriangleright$]
    \item the value of every node in $\IN_k$ is at least $\|x\|_1 \cdot \prod_{i=1}^{k-1} |\IN_i|$;
        \begin{itemize}[$\diamond$]
        \item
        This can be proven by induction. First notice that since every node in $\IN_1$ is connected by edges with weight $+1$ to every input node, the value of each node in $\IN_1$ is exactly $\|x\|_1$, proving the base case of induction for $k=1$. Next, suppose that for $2 \le k \le \ell-1$, the claim holds for the nodes from $\IN_{k-1}$ (so that the value of every node in $\IN_{k-1}$ is at least $\|x\|_1 \cdot \prod_{i=1}^{k-2} |\IN_i|$). Since any node $i_k$ from $\IN_k$ is connected to the nodes from $V_{k-1}$ by edges with weight $1$, every node from $\IN_{k-1}$ contributes $\|x\|_1 \cdot \prod_{i=1}^{k-2} |\IN_i|$ to the value of $i_k$, and all other nodes from $V_{k-1} \setminus \IN_{k-1}$ have non-negative contributions to the value of $i_k$, giving the value of node $i_k$ from $\IN_k$ to be at least $|\IN_k| \cdot \|x\|_1 \cdot \prod_{i=1}^{k-2} |\IN_i| = \|x\|_1 \cdot \prod_{i=1}^{k-1} |\IN_i|$.
        \end{itemize}
    \item this involves $m_{k-1} \cdot |\IN_k| \le \frac12 \epsilon m_{k-1} m_k$ changes of weights of edges from $E_{k-1}$, $1 \le k \le \ell$.
    \end{itemize}
\item For the nodes in $\IN_{\ell}$, we change to $-1$ the weights of all edges connecting nodes from $\IN_{\ell}$ to the output node;
    \begin{itemize}[$\blacktriangleright$]
    \item the contribution to the output of every node in $\IN_{\ell}$ is negative, at most $- \|x\|_1 \cdot \prod_{i=1}^{\ell-1} |\IN_i|$;
    \item this involves $|\IN_{\ell}| \le \frac12 \epsilon m_{\ell}$ changes of weights of $E_{\ell}$ edges.
    \end{itemize}
\end{itemize}
The construction above ensures that the contribution of the nodes from $\IN_{\ell}$ to the final output is at most $- \|x\|_1 \cdot \prod_{i=1}^{\ell} |\IN_i| = - (\epsilon/2)^{\ell} \cdot \|x\|_1 \cdot \prod_{i=1}^{\ell} m_i$. However, we also want to ensure that the contributions of other nodes from the final hidden layer does not change much with respect to their original contribution. For that, we ``remove'' the contributions of the selected nodes $\bigcup_{i=1}^{\ell-1} \IN_i$.
\begin{itemize}
\item For every $1 \le k \le \ell-1$, we change to $0$ the weights of all edges connecting the nodes from $\IN_k$ to the nodes from $\sset^*_{k+1}$;
    \begin{itemize}[$\blacktriangleright$]
    \item this involves $|\IN_k| \cdot |\sset^*_{k+1}| \le \frac12 \epsilon m_k m_{k+1}$ changes of weights of $E_k$ edges.
    \end{itemize}
\end{itemize}

\paragraph{Properties of the modified ReLU network.}
Observe that in total, we modified
\begin{inparaenum}[(i)]
\item $m_0 \cdot |\IN_1| \le \frac12 \epsilon m_0 m_1$ weights of $E_0$ edges,
\item for every $1 \le k \le \ell-1$, $m_k \cdot |\IN_{k+1}| + |\IN_k| \cdot |\sset^*_{k+1}| \le \epsilon m_k m_{k+1}$ weights of $E_k$ edges, and
\item $|\IN_{\ell}| = \frac12 \epsilon m_{\ell}$ weights of $E_{\ell}$ edges.
\end{inparaenum}
Therefore, since the original ReLU network is $(\epsilon, \delta)$-far from computing the constant $0$-function, our network obtained after modifying at most an $\epsilon$-fraction of weights at each level cannot return $0$ for all but a $\delta$-fraction of the inputs.

Let us now look at the value returned by our modified network, focusing on the contributions of all $\ell$-th hidden layer nodes from $V_{\ell}$.
\begin{itemize}
\item[\textsf{Nodes in $\IN_{\ell}$:}] By our arguments above, the contribution of the nodes from $\IN_{\ell}$ is negative, and it is at most $- \|x\|_1 \cdot \prod_{i=1}^{\ell} |\IN_i| = - (\epsilon/2)^{\ell} \cdot \|x\|_1 \cdot \prod_{i=1}^{\ell} m_i$.
\item[\textsf{Nodes in $\sset^*_{\ell}$:}] Our construction ensures that the contribution of all nodes from $\sset^*_{\ell}$ is exactly the output value $h^*_{\ell+1}(x)$ of a ReLU network $(W_0 \cdot S^*_0, \dots, W_{\ell} \cdot S^*_{\ell})$ with the sampling matrices $S^*_0, \dots, S^*_{\ell}$ corresponding to the nodes $\sset^*_0, \sset^*_1, \dots, \sset^*_{\ell}$. Therefore, by (\ref{ineq:combo-lemma:output-mL+lemma:sampling-mL-existance}), for all but a $\probp$-fraction of the inputs it holds $\left|h^*_{\ell+1}(x) - (1-\tfrac12 \epsilon)^{\ell} \cdot f_{\ell+1}(x) \right| \le \vartheta \cdot (1-\tfrac12 \epsilon)^{\ell} \cdot \prod_{i=0}^{\ell} m_i$, and thus, $h^*_{\ell+1}(x) \le (1-\tfrac12 \epsilon)^{\ell} \cdot f_{\ell+1}(x) + \vartheta \cdot (1-\tfrac12 \epsilon)^{\ell} \cdot \prod_{i=0}^{\ell} m_i$ for all but a $\probp$-fraction of the~inputs.
\end{itemize}

Let $\mathfrak{I}$ be the set of inputs $x \in \{0,1\}^{m_0}$ for which $h^*_{\ell+1}(x) \le (1-\tfrac12 \epsilon)^{\ell} \cdot f_{\ell+1}(x) + \vartheta \cdot (1-\tfrac12 \epsilon)^{\ell} \cdot \prod_{i=0}^{\ell} m_i$ \emph{and} that $\|x\|_1 \ge \frac14 m_0$. By our arguments above and by \cref{claim:NumberOfOnes}, $|\mathfrak{I}| \ge (1 - \probp - e^{-m_0/16}) \cdot 2^{m_0}$.

Our construction ensures that for every input $x \in \{0,1\}^{m_0}$, the modified network returns the value of at most $- (\epsilon/2)^{\ell} \cdot \|x\|_1 \cdot \prod_{i=1}^{\ell} m_i + h^*_{\ell+1}(x)$, and therefore our definition of set $\mathfrak{I}$ implies that for every input $x \in \mathfrak{I}$, the modified network returns the value of at most
\begin{align}
\label{ineq:lemma:output-mL-bound-output-modified-network}
        \notag
    - \frac14 \cdot (\epsilon/2)^{\ell} \cdot \prod_{i=0}^{\ell} m_i + (1-\tfrac12 \epsilon)^{\ell} \cdot f_{\ell+1}(x) + \vartheta \cdot (1-\tfrac12 \epsilon)^{\ell} \cdot \prod_{i=0}^{\ell} m_i
        =
        \\
    (1-\tfrac12 \epsilon)^{\ell} \cdot
    \left(
        f_{\ell+1}(x) + \left(\vartheta - \frac14 \cdot \left(\frac{\epsilon}{2-\epsilon}\right)^{\ell} \cdot \right)\cdot \prod_{i=0}^{\ell} m_i
    \right)
    \enspace.
\end{align}
Since the original ReLU network $(W_0, \dots, W_{\ell})$ is $(\epsilon, \delta)$-far from computing the constant $0$-function, our network obtained after modifying at most an $\epsilon$-fraction of weights at each layer cannot return $0$ for more than a $\delta$-fraction of the inputs. Observe that since we have set $\delta \ge \probp + e^{-m_0/16}$, there must be at least one input $x_0 \in \mathfrak{I}$ for which the modified network returns positive value. However, since the value returned by $x_0$ is upper bounded in (\ref{ineq:lemma:output-mL-bound-output-modified-network}), we obtain that
\begin{align*}
    (1-\tfrac12 \epsilon)^{\ell} \cdot
    \left(
        f_{\ell+1}(x) + \left(\vartheta - \frac14 \cdot \left(\frac{\epsilon}{2-\epsilon}\right)^{\ell} \cdot \right)\cdot \prod_{i=0}^{\ell} m_i
    \right)
        &>
    0
    \enspace.
\end{align*}
Observe that since $\frac{\epsilon}{2-\epsilon} \ge \frac12 \epsilon$, we obtain the following final inequality:
\begin{align*}
    f_{\ell+1}(x_0)
        &>
    \left(\frac14 \cdot (\epsilon/2)^{\ell} - \vartheta\right)
        \cdot \prod_{i=0}^{\ell} m_i
        \enspace,
\end{align*}
which completes the proof of \cref{lemma:output-mL} after setting $\vartheta := \frac18 \cdot (\epsilon/2)^{\ell}$.
\end{proof}


\begin{remark}\rm
\label{remark:auxiliary-to-lemma:output-mL}
The auxiliary parameter $\probp$ has two uses in our analysis: on one hand, it puts some constraints on $m_0, \dots, m_{\ell}$ and on the other hand, it allows us to parameterize $\delta$. In our use of \cref{lemma:output-mL} in \cref{lem:ReLU-testing-0-function-2-sided-far-mL} in the next \cref{subsubsec:ReLU-testing-0-function-2-sided-far-mL}, we will want to have $\delta$ as small as possible, and with that, we will want to have $\probp$ as small as possible, possibly close to $2^{-\Theta(m_0)}$. In order to make our setting concrete, let us state an auxiliary claim which shows that under some mild conditions about $m_0, \dots, m_{\ell}$, we could get $\delta \sim e^{-m_0/16} + e^{-\Theta(\min m_k \cdot (\epsilon/2)^{2\ell}/\ell^2)}$ and $\probp \sim e^{-\Theta(\min m_k \cdot (\epsilon/2)^{2\ell}/\ell^2)}$.

\begin{claim}
\label{claim:auxiliary-to-lemma:output-mL}
Let $0 < \epsilon < \frac12$ and $\ell \ge 1$. Let $0 < \probp < \frac{1}{\ell+1}$ be allowed to depend on $m_0, \dots, m_{\ell}$. If
\begin{inparaenum}[(i)]
\item\label{aux-claim-cond-2} $\prod_{i=1}^{\ell} m_i \le \frac{1}{\probp}$ and
\item\label{aux-claim-cond-1} $\probp \ge e^{-\frac{(\epsilon/2)^{2\ell}}{342 (\ell+1)^2} \cdot \min_{0 \le k \le \ell}\{m_k\}}$,
\end{inparaenum}
then
\begin{itemize}
\item $m_0 \ge 128 \cdot (\ell+1)^2 \cdot (2/\epsilon)^{2\ell} \cdot \ln\left(\frac{1}{\probp} \prod_{i=1}^{\ell} m_i\right)$, and
\item $m_k \ge 171 \cdot (\ell+1)^2 \cdot (2/\epsilon)^{2\ell} \cdot \ln\left(\frac{1}{\probp} \prod_{i=k+1}^{\ell} m_i\right)$ for every $1 \le k \le \ell$.
\end{itemize}
In particular, if $n = m_0 = \dots = m_{\ell}$, then for a sufficiently large $n$ (since we require that $n^{\ell} \le e^{(\epsilon/2)^{\ell} n /(342 (\ell+1)^2)}$), the conditions on $m_0, \dots, m_{\ell}$ are satisfied as long as $\probp \ge e^{- (\epsilon/2)^{\ell} n /(342 (\ell+1)^2)}$.
\end{claim}

\begin{proof}
Let us focus first on $m_0$. Since $\prod_{i=1}^{\ell} m_i \le \frac{1}{\probp}$ (by assumption \emph{(\ref{aux-claim-cond-1})}), we obtain the following
\begin{align*}
    128 \cdot (\ell+1)^2 \cdot (2/\epsilon)^{2\ell} \cdot
        \ln\left(\frac{1}{\probp} \prod_{i=1}^{\ell} m_i\right)
        &\le
    128 \cdot (\ell+1)^2 \cdot (2/\epsilon)^{2\ell} \cdot
        \ln\left(\frac{1}{\probp^2}\right)
        \\
        &=
    256 \cdot (\ell+1)^2 \cdot (2/\epsilon)^{2\ell} \cdot
        \ln\left(\frac{1}{\probp}\right)
        \\
        &\le
    m_0
    \enspace,
\end{align*}
where the last inequality follows from (assumption \emph{(\ref{aux-claim-cond-2})}).

Similarly, for $m_k$ with $1 \le k \le \ell$, by assumption \emph{(\ref{aux-claim-cond-1})} we have
\begin{align*}
    171 \cdot (\ell+1)^2 \cdot (2/\epsilon)^{2\ell} \cdot
        \ln\left(\frac{1}{\probp} \prod_{i=k+1}^{\ell} m_i\right)
        &\le
    171 \cdot (\ell+1)^2 \cdot (2/\epsilon)^{2\ell} \cdot \ln\left(\frac{1}{\probp^2}\right)
        \\
        &=
    342 \cdot (\ell+1)^2 \cdot (2/\epsilon)^{2\ell} \cdot \ln\left(\frac{1}{\probp}\right)
        \\
        &\le
    m_k
    \enspace,
\end{align*}
where the last inequality follows from (assumption \emph{(\ref{aux-claim-cond-2})}).

Finally, if all $m_k$ are the same, $n = m_0 = \dots = m_{\ell}$, then assumptions \emph{(\ref{aux-claim-cond-1})} and \emph{(\ref{aux-claim-cond-2})} simplify to $n^{\ell} \le \frac{1}{\probp}$ and $\probp \ge e^{-(\epsilon/2)^{2\ell}/(342 (\ell+1)^2)}$, yielding the promised claim.
\end{proof}

\end{remark}


\subsubsection{If $(\epsilon, \delta)$-far then we reject (extending \cref{lem:ReLU-testing-0-OR-function-2-sided-far} to multiple layers)}
\label{subsubsec:ReLU-testing-0-function-2-sided-far-mL}

We can now incorporate \cref{lemma:output-mL,lemma:sampling-mL} to finalize the analysis of our testers for ReLU networks that are $(\epsilon,\delta)$-far from computing the constant $0$-function.

\begin{lemma}
\label{lem:ReLU-testing-0-function-2-sided-far-mL}
Let $(W_0, \dots, W_{\ell})$ be a ReLU network with $m_0$ input nodes, $\ell \ge 1$ hidden layers with $m_1, \dots, m_{\ell}$ hidden layer nodes each, and a single output node.
Let $\frac{1}{\min_{0 \le k \le \ell}\{m_k\}} < \epsilon < \frac12$ and $0 < \lambda < \frac{1}{\ell+1}$.
Let $0 < \probp < \frac{1}{\ell+1}$ be an arbitrary parameter and let $\delta \ge \probp + e^{-m_0/16}$.
Suppose that $m_0, \dots, m_{\ell}$ and $s_0, \dots, s_{\ell}$ satisfy the following:
\begin{itemize}
\item $m_0 \ge 128 \cdot (\ell+1)^2 \cdot (2/\epsilon)^{2\ell} \cdot \ln(\frac{1}{\probp} \prod_{i=1}^{\ell} m_i)$,
\item $m_k \ge 171 \cdot (\ell+1)^2 \cdot (2/\epsilon)^{2\ell} \cdot \ln(\frac{1}{\probp} \prod_{i=k+1}^{\ell} m_i)$ for every $1 \le k \le \ell$, and
\item $s_k \ge 512 \cdot (\ell+1)^2 \cdot (2/\epsilon)^{2\ell} \cdot \ln\left(\frac{2 \prod_{i=k+1}^{\ell} s_i}{\lambda(\ell+1)}\right)$ for every $0 \le k \le \ell$.
\end{itemize}
If ReLU network $(W_0, \dots, W_{\ell})$ is $(\epsilon, \delta)$-far from computing the constant $0$-function then with probability at least $1 - \lambda$ (over the random choice of $S_0, \dots, S_{\ell}$) there is an $x \in \{0,1\}^{m_0}$ for which
\begin{align*}
    \prod_{i=0}^{\ell} \frac{m_i}{s_i} \cdot h_{\ell+1}(x) -
        \frac{1}{16} \cdot (\epsilon/2)^{\ell} \cdot \prod_{i=0}^{\ell} m_i
    &> 0
    \enspace,
\end{align*}
and so algorithm \textsc{AllZeroTester2HL} (page \pageref{alg:AllZeroTester-mL}) rejects with probability at least $1 - \lambda$.
\end{lemma}

\begin{proof}
Let $f_i(x)$, $f(x)$, $h_i(x)$, and $h(x)$ be defined as in \cref{lemma:output-mL}. 
By \cref{lemma:output-mL}, since $(W_0, \dots, W_{\ell})$ is $(\epsilon,\delta)$-far from computing the constant $0$-function, there exists an $x_0 \in \{0,1\}^n$ such that on input $x_0$, the value in the output node of the network is at least $\frac18 \cdot (\epsilon/2)^{\ell} \cdot \prod_{i=0}^{\ell} m_i$, that is, that $f_{\ell+1}(x_0) \ge \frac18 \cdot (\epsilon/2)^{\ell} \cdot \prod_{i=0}^{\ell} m_i$. Next, by \cref{lemma:sampling-mL}, with probability at least $1 - \lambda$ (over the random choice of $S_0, \dots, S_{\ell}$ in the algorithm) we have for $\vartheta := \frac{1}{16} \cdot (\epsilon/2)^{\ell}$:
\begin{align*}
    \Pr\left[
    \left| \prod_{i=0}^{\ell} \frac{m_i}{s_i} \cdot h_{\ell+1}(x_0) - f_{\ell+1}(x_0) \right|
        >
    \frac{1}{16} \cdot (\epsilon/2)^{\ell} \cdot \prod_{i=0}^{\ell} m_i
    \right]
        & \ge
    1- \lambda
    \enspace.
\end{align*}
(Here the use of \cref{lemma:sampling-mL} requires $s_k \ge 512 (\ell+1)^2 (2/\epsilon)^{2\ell} \ln\left(\frac{2 \prod_{i=k+1}^{\ell} s_i}{\lambda(\ell+1)}\right)$ for every $0 \le k \le \ell$.)

This implies that with probability at least $1 - \lambda$ we obtain
\begin{align*}
    \prod_{i=0}^{\ell} \frac{m_i}{s_i} \cdot h_{\ell+1}(x_0)
        & >
    \frac{1}{16} \cdot (\epsilon/2)^{\ell} \cdot \prod_{i=0}^{\ell} m_i
    \enspace,
\end{align*}
and so \textsc{AllZeroTester2HL} rejects with probability at least $1 - \lambda$.
\end{proof}


\subsubsection{If network computes 0 then accept (extending \cref{lem:ReLU-testing-0-OR-function-2-sided-compute} to multiple layers)}
\label{subsubsec:ReLU-testing-0-function-2-sided-compute-mL}

Our final lemma here summarizes the second part of the analysis of \textsc{AllZeroTester2HL} for ReLU networks and considers the networks that correctly compute the constant $0$-function.

\begin{lemma}
\label{lem:ReLU-testing-0-function-2-sided-compute-mL}
Let $0 < \lambda < \frac{1}{\ell+1}$, $\delta \ge e^{-m_0/16}$, and $\frac{1}{\min_{0 \le k \le \ell}\{m_k\}} < \epsilon < \frac12$.
Let $(W_0, \dots, W_{\ell})$ be a ReLU network with $m_0$ input nodes, $\ell \ge 1$ hidden layers with $m_1, \dots, m_{\ell}$ hidden layer nodes each, and a single output node.
Suppose that $s_k \ge 512 \cdot \ell^2 \cdot (2/\epsilon)^{2\ell} \cdot \ln\left(\frac{2^{s_0+1} \cdot \ell \cdot \prod_{i=k+1}^{\ell} s_i}{\lambda}\right)$ for any $1 \le k \le \ell$.
If $(W_0, \dots, W_{\ell})$ computes the constant $0$-function, then it holds with probability at least $1 - \lambda$ (over the random choice of the sampling matrices $S_0, \dots, S_{\ell}$) that for all $x \in \{0,1\}^{m_0}$
\begin{align*}
    \prod_{i=0}^{\ell} \frac{m_i}{s_i} \cdot h_{\ell+1}(x) -
        \frac{1}{16} \cdot (\epsilon/2)^{\ell} \cdot \prod_{i=0}^{\ell} m_i
    &\le 0
    \enspace,
\end{align*}
and so algorithm {\sc AllZeroTester2HL} (page \pageref{alg:AllZeroTester-mL}) accepts with probability at least $1 - \lambda$.
\end{lemma}

\begin{proof}
Before we proceed with the proof, let us first notice the main difference between \cref{lem:ReLU-testing-0-function-2-sided-far-mL} and \cref{lem:ReLU-testing-0-function-2-sided-compute-mL}: the key point of \cref{lem:ReLU-testing-0-function-2-sided-compute-mL} (and the challenge in proving it) is that we have to prove that \emph{for every input} $x \in \{0,1\}^{m_0}$ we have $\prod_{i=0}^{\ell} \frac{m_i}{s_i} \cdot h_{\ell+1}(x) - \frac{1}{16} \cdot (\epsilon/2)^{\ell} \cdot \prod_{i=0}^{\ell} m_i \le 0$, whereas in \cref{lem:ReLU-testing-0-function-2-sided-far-mL}, we only have to find a single input $x$ for which the inequality above fails.

Let $(W_0, \dots, W_{\ell})$ be a ReLU network that computes the constant $0$-function. Let ${\sset_0}$ be the $s_0$ nodes sampled by our algorithm from the input layer and let $X_{\sset_0} = \{x = (x_1,\dots,x_n)^T \in\{0,1\}^{m_0} : x_i = 0, \forall i \notin \sset_0\}$ be the set of different inputs to the sampled network. In what follows, we will condition on the choice of $\sset_0$ being fixed (and the randomness comes from sampling $\sset_1, \dots, \sset_{\ell}$).

Observe that
\begin{align*}
    \prod_{i=0}^{\ell} \frac{m_i}{s_i} \cdot h_{\ell+1}(x) -
        \frac{1}{16} \cdot (\epsilon/2)^{\ell} \cdot \prod_{i=0}^{\ell} m_i \le 0
\end{align*}
holds for all $x\in \{0,1\}^{m_0}$ if and only if this inequality holds for all $x \in X_{\sset_0}$. Therefore, in what follows, we will show that the output value on a fixed input $x \in X_{\sset_0}$ will be approximated with high probability and then we will apply a union bound over all $x \in X_{\sset_0}$ to conclude the analysis.

\footnote{Notice that \emph{if we could use the randomness of $\sset_0$}, then by \cref{lemma:sampling-mL}, for any $x \in X_{\sset_0}$ we \emph{would have}
\begin{align*}
    \Pr\left[\left| \prod_{i=0}^{\ell} \frac{m_i}{s_i} \cdot h_{\ell+1}(x) - f_{\ell+1}(x) \right|
        >
    \frac{1}{16} \cdot (\epsilon/2)^{\ell} \cdot \prod_{i=0}^{\ell} m_i \right]
    &\le \lambda
    \enspace.
\end{align*}
However, since we are conditioning on fixed $\sset_0$, we have to use a bit more elaborate arguments.}%
Observe that for any $x \in X_{\sset_0}$, since $x$ has only $0$s outside set $\sset_0$, ReLU networks $(W_0, W_1 \cdot S_1, \dots, W_{\ell} \cdot S_{\ell})$ and $(W_0 \cdot S_0, W_1 \cdot S_1, \dots, W_{\ell} \cdot S_{\ell})$ are outputting the same value on any $x \in X_{\sset_0}$, and hence $h_{\ell+1}(x) = \phi_1^{\ell+1}(f_1(x))$. For random sampling sets $\sset_1, \dots, \sset_{\ell}$ (\emph{independently of the sampling set $\sset_0$}) we studied the relation between $\phi_1^{\ell+1}(f_1(x))$ and $f_{\ell+1}(x) = \phi_{\ell+1}^{\ell+1}(f_{\ell+1}(x))$ already in the proof of \cref{lemma:sampling-mL}, indirectly in (\ref{ineq-mL-aux7}). Let us consider (\ref{ineq-mL-aux7}) with $\kappa_k := \frac{(\epsilon/2)^{\ell}}{16 \ell}$ and $\lambda_k := \frac{\lambda}{\ell \cdot 2^{s_0} \cdot \prod_{i=k+1}^{\ell} s_i}$. Then for any $1 \le k \le \ell$, if $s_k \ge 512 \ell^2 (2/\epsilon)^{2\ell} \ln(2^{s_0+1} \ell (\prod_{i=k+1}^{\ell} s_i)/\lambda)$ then (\ref{ineq-mL-aux7}) yields the following
\begin{align*}
    \Pr\left[
        \left|
            \frac{m_k}{s_k} \cdot \phi_k^{\ell+1}(f_k(x)) -
            \phi_{k+1}^{\ell+1}(f_{k+1}(x))
        \right|
        > \frac{(\epsilon/2)^{\ell}}{16 \ell} \cdot \prod_{i=0}^k m_i \cdot \prod_{i=k+1}^{\ell} s_i
        \; \Big| \; \sset_0
    \right]
        &\le
        \frac{\lambda}{2^{s_0} \cdot \ell}
    \enspace.
\end{align*}
If we combine this inequality with the telescoping sum (\ref{ineq-mL-aux6}), then for any $x \in X_{\sset_0}$:
\begin{eqnarray*}
    \lefteqn{
    \Pr\left[\left| \prod_{i=0}^{\ell} \frac{m_i}{s_i} \cdot h_{\ell+1}(x) - f_{\ell+1}(x) \right|
        >
        \frac{1}{16} \cdot (\epsilon/2)^{\ell} \cdot \prod_{i=0}^{\ell} m_i \; \Big| \; \sset_0
    \right]
        =
    }
        \\&&
    \Pr\left[\left| \prod_{i=1}^{\ell} \frac{m_i}{s_i} \cdot \phi_1^{\ell+1}(f_1(x)) - \phi_{\ell+1}^{\ell+1}(f_{\ell+1}(x)) \right|
        >
        \frac{1}{16} \cdot (\epsilon/2)^{\ell} \cdot \prod_{i=0}^{\ell} m_i
    \right]
        \le
        \\&&
    \sum_{k=1}^{\ell}
        \Pr\left[
            \left|
                \frac{m_k}{s_k} \cdot \phi_k^{\ell+1}(f_k(x)) -
                \phi_{k+1}^{\ell+1}(f_{k+1}(x))
            \right|
                >
            \frac{(\epsilon/2)^{\ell}}{16 \ell} \cdot \prod_{i=0}^k m_i \cdot \prod_{i=k+1}^{\ell} s_i
        \right]
        \le
        \\&&
    \sum_{k=1}^{\ell} \frac{\lambda}{2^{s_0} \cdot \ell}
        = \frac{\lambda}{2^{s_0}}
    \enspace.
\end{eqnarray*}
(Observe that since $f_1(x)$ does not depend on $\sset_0$, there is no conditioning in the second term, and since we consider functions $f_k(x)$ with $k \ge 1$, the summation in the third term starts only from $1$.)

Therefore, by the union bound over the $|X_{\sset_0}| = 2^{s_0}$ vectors in $X_{\sset_0}$, this implies that the following inequality holds for all $x \in \{0,1\}^{m_0}$:
\begin{align*}
    \Pr\left[\left| \prod_{i=0}^{\ell} \frac{m_i}{s_i} \cdot h_{\ell+1}(x) - f_{\ell+1}(x) \right|
        >
        \frac{1}{16} \cdot (\epsilon/2)^{\ell} \cdot \prod_{i=0}^{\ell} m_i \; \Big| \; \sset_0
    \right]
        &\le
    2^{s_0} \cdot \frac{\lambda}{2^{s_0}}
        =
    \lambda
    \enspace.
\end{align*}
Now the lemma follows from the fact that for every $x \in \{0,1\}^{m_0}$ we have $f_{\ell+1}(x) \le 0$, and hence,
\begin{align*}
    \lefteqn{
    \Pr\left[\prod_{i=0}^{\ell} \frac{m_i}{s_i} \cdot h_{\ell+1}(x)
        >
        \frac{1}{16} \cdot (\epsilon/2)^{\ell} \cdot \prod_{i=0}^{\ell} m_i
    \right]
        \le
    }
        \\&
    \sum_{\sset_0}
    \left(
        \Pr[\sset_0] \cdot
        \Pr\left[\prod_{i=0}^{\ell} \frac{m_i}{s_i} \cdot h_{\ell+1}(x) >
            \frac{1}{16} \cdot (\epsilon/2)^{\ell} \cdot \prod_{i=0}^{\ell} m_i
            \; \Big| \; \sset_0 \right]
    \right)
        \le
    \lambda
    \enspace.
    \qedhere
\end{align*}
\end{proof}


\subsubsection{Combining \cref{lem:ReLU-testing-0-function-2-sided-far-mL,lem:ReLU-testing-0-function-2-sided-compute-mL} to complete the proof of \cref{thm:ReLU-testing-0-function-2-sided-mL}}
\label{subsubsec:proving-thm:ReLU-testing-0-function-2-sided-mL}

Let us first observe that (assuming $\ell$ is constant) the query complexity of Algorithm \textsc{AllZeroTesterMHL} (page \pageref{alg:AllZeroTester-mL}) is $s_{\ell} + \sum_{k=0}^{\ell-1} s_k \cdot s_{k+1} = \Theta(\sum_{k=0}^{\ell-1} s_k \cdot s_{k+1}) = \Theta(\epsilon^{-8\ell} \cdot \ln^4(1/\lambda\epsilon))$, since we have to consider only the original ReLU network induced by the nodes from $\sset_0, \dots, \sset_{\ell}$ and the single output node, and such network has exactly $s_{\ell} + \sum_{k=0}^{\ell-1} s_k \cdot s_{k+1}$ edges.%
\footnote{
Another way of arguing about this, is to observe that $h_{\ell+1}(x)$ can be also computed as $h_{\ell+1}(x) = g_{\ell+1}(x)$, where after defining $S_{\ell+1}$ to be the identity $1 \times 1$ matrix $1$, $g_k(x)$ is defined recursively as follows:
\begin{align*}
    g_k(x) &:=
        \begin{cases}
            S_k \cdot W_{k-1} \cdot S_{k-1} \cdot \relu(g_{k-1}(x)) & \text{ if } 1 \le k \le \ell+1, \\
            S_0 \cdot x & \text{ if } k=0 \enspace.
        \end{cases}
\end{align*}
(This follows since the ReLU function is positive homogeneous and $S_0, \dots, S_{\ell}$ have only non-negative entries, and thus $S_{k-1} \cdot \relu(y) = S_{k-1} \cdot \relu(S_{k-1} \cdot y)$.)
Since $S_k$ is a diagonal matrix with $s_k$ non-zero entries, $S_k \cdot W_{k-1} \cdot S_{k-1}$ can be computed by taking the $s_k$ rows (with indices from $\sset_k$) and $s_{k-1}$ columns (with indices from $\sset_{k-1}$) of matrix $W_{k-1}$. 
Therefore, in order to compute $g_{\ell+1}(x)$, and hence $h_{\ell+1}(x)$, we take only $s_k \cdot s_{k-1}$ entries of any matrix $W_{k-1}$, implying our claim.
}

Next, observe that for any constant $\ell \ge 1$, there exist positive constants $c_0, \dots, c_{\ell}$ such that if
\begin{itemize}
\item we set $s_0 := c_0 \cdot \epsilon^{-2\ell} \cdot \ln (1/\lambda\epsilon)$ and
\item for every $1 \le k \le \ell$ we set $s_k := c_k \cdot \epsilon^{-4\ell} \cdot \ln^2(1/\lambda\epsilon)$,
\end{itemize}
then $s_0, \dots, s_{\ell}$ satisfy the assumptions of \cref{lem:ReLU-testing-0-function-2-sided-far-mL,lem:ReLU-testing-0-function-2-sided-compute-mL}, that is,
\begin{itemize}
\item $s_k \ge 512 \cdot (\ell+1)^2 \cdot (2/\epsilon)^{2\ell} \cdot \ln\left(\frac{2 \prod_{i=k+1}^{\ell} s_i}{\lambda(\ell+1)}\right)$ for every $0 \le k \le \ell$, and
\item $s_k \ge 512 \cdot \ell^2 \cdot (2/\epsilon)^{2\ell} \cdot \ln\left(\frac{2^{s_0+1} \cdot \ell \cdot \prod_{i=k+1}^{\ell} s_i}{\lambda}\right)$ for every $1 \le k \le \ell$.
\end{itemize}

Therefore, with our conditions on $m_0, \dots, m_{\ell}$ and our choice of $s_0, \dots, s_{\ell}$, \cref{thm:ReLU-testing-0-function-2-sided-mL} follows immediately from \cref{lem:ReLU-testing-0-function-2-sided-far-mL,lem:ReLU-testing-0-function-2-sided-compute-mL} and from our arguments above that Algorithm \textsc{AllZeroTesterMHL} (page \pageref{alg:AllZeroTester-mL}) samples only $s_{\ell} + \sum_{k=0}^{\ell-1} s_k = \Theta(\epsilon^{-8\ell} \cdot \ln^4(1/\lambda\epsilon))$ entries of the input matrices $W_0, \dots, W_{\ell}$.
\qed


\subsection{Testing the \emph{OR-function} of a ReLU network with multiple layers}
\label{sec:ReLU-testing-OR-function-2-sided-mL}

It is straightforward to slightly modify Algorithm \textsc{AllZeroTesterMHL} and \cref{thm:ReLU-testing-0-function-2-sided-mL} to obtain a tester for the OR-function of a ReLU network with multiple layers.

Let us begin with the modification of Algorithm \textsc{AllZeroTesterMHL} to a tester for the OR-function. We follow the same approach as for single hidden layer ReLU networks, where we extended Algorithm \ref{alg:AllZeroTester} \textsc{AllZeroTester} for testing the constant $0$-function to Algorithm \ref{alg:ORTester} \textsc{ORTester} testing the OR-function. We randomly sample some subset of input nodes and of nodes at every hidden layer, and compute the function returned by the ReLU network induced by the sampled nodes \emph{adjusted by some small bias}. For the resulting network with bias we check whether it computes the OR-function. If this is the case, we accept. Otherwise, we reject.


\begin{algorithm}[h]
\SetAlgoLined\DontPrintSemicolon
\caption{\textsc{ORTesterMHL}$(\epsilon, \lambda, m_0, m_1, \dots, m_{\ell})$}
\label{alg:ORTester-mL}

\KwIn{ReLU network $(W_0, \dots, W_{\ell})$ with $m_0$ input nodes, $\ell \ge 1$ hidden layers with $m_1, \dots, m_{\ell}$ hidden layer nodes each, and a single output node;
parameters $\epsilon$ and $\lambda$}

set $s_0 := c_0 \cdot \epsilon^{-2\ell} \cdot \ln (1/\lambda\epsilon)$ for an appropriate positive constant $c_0$

for every $1 \le k \le \ell$, set $s_k := c_k \cdot \epsilon^{-4\ell} \cdot \ln (1/\lambda\epsilon)$ for an appropriate positive constant $c_k$

\For{$i=0$ \KwTo $\ell$}{
    sample $s_i$ nodes $\sset_i$ from $V_i$ uniformly at random without replacement

    let $S_i \in \NN_0^{m_i \times m_i}$ be the corresponding sampling matrix
}

let $h :\{0,1\}^{m_0} \rightarrow \{0,1\}$ be defined recursively as $h(x) := \sgn(\relu(h_{\ell+1}(x)))$, where
\begin{align*}
    h_i(x) &:=
        \begin{cases}
            W_{i-1} \cdot S_{i-1} \cdot \relu(h_{i-1}(x)) & \text{ if } 1 \le i \le \ell+1, \\
            x & \text{ if } i=0.
        \end{cases}
\end{align*}

\eIf{\emph{there is $x \in \{0,1\}^n$ with $\prod_{i=0}^{\ell} \frac{m_i}{s_i} \cdot h_{\ell+1}(x) + \frac14 \cdot (\epsilon/2)^{\ell} \cdot \prod_{i=0}^{\ell} m_i < 0$}}{\textbf{reject}}
{\textbf{accept}}
\end{algorithm}


The analysis 
essentially identical to that in the proof of \cref{thm:ReLU-testing-0-function-2-sided-mL} gives the following theorem.

\begin{theorem}
\label{thm:ReLU-testing-OR-function-2-sided-mL}
Let $0 < \lambda < \frac{1}{\ell+1}$ and $\frac{1}{\min_{0 \le k \le \ell}\{m_k\}} < \epsilon < \frac12$.
Let \NE be a ReLU network $(W_0, \dots, W_{\ell})$ with $m_0$ input nodes, $\ell \ge 1$ hidden layers with $m_1, \dots, m_{\ell}$ hidden layer nodes each, and a single output node. Further, suppose that for an arbitrary parameter $0 < \probp < \frac{1}{\ell+1}$, it holds that $m_0 \ge 128 \cdot (\ell+1)^2 \cdot (2/\epsilon)^{2\ell} \cdot \ln(\frac{1}{\probp} \prod_{i=1}^{\ell} m_i)$ and that $m_k \ge 171 \cdot (\ell+1)^2 \cdot (2/\epsilon)^{2\ell} \cdot \ln(\frac{1}{\probp} \prod_{i=k+1}^{\ell} m_i)$ for every $1 \le k \le \ell$.
Let $\delta \ge \probp + e^{-m_0/16}$.
Assuming that $\ell$ is constant, tester \textsc{ORTesterMHL} queries $\Theta(\epsilon^{-8\ell} \cdot \ln^4(1/\lambda\epsilon))$ entries from $W_0, \dots, W_{\ell}$, and
\begin{enumerate}[(i)]
\item rejects with probability at least $1 - \lambda$, if \NE is $(\epsilon,\delta)$-far from computing the OR-function, and
\item accepts with probability at least $1 - \lambda$, if \NE computes the OR-function.
\end{enumerate}
\end{theorem}

\begin{proof}
The arguments mimic the proof of \cref{thm:ReLU-testing-0-function-2-sided-mL}. The only difference is in the use of \cref{lemma:output-mL}, where essentially identical arguments as in the proof of \cref{lemma:output-mL} can be applied to show that if $(W_0, \dots, W_{\ell})$ is $(\epsilon, \delta)$-far from computing the OR-function then there exists an input $x_0 \in \{0,1\}^{m_0}$ such that $f_{\ell+1}(x) < - \frac18 \cdot (\epsilon/2)^{\ell} \cdot \prod_{i=0}^{\ell} m_i$. Once we have such claim, then the arguments used in the proofs of \cref{lem:ReLU-testing-0-function-2-sided-far-mL} and \cref{lem:ReLU-testing-0-function-2-sided-compute-mL} directly yield \cref{thm:ReLU-testing-OR-function-2-sided-mL}.
\end{proof}

Similarly as for \cref{thm:ReLU-testing-0-function-2-sided-mL}, in the special case when $m_0 = \dots = m_{\ell}$, we can simplify the bounds and obtain the following extension of \cref{thm:ReLU-testing-0-function-2-sided-mL-n} to the OR function.

\begin{theorem}
\label{thm:ReLU-testing-OR-function-2-sided-mL-m0=n}
Let $0 < \lambda < \frac{1}{\ell+1}$, $\frac{1}{n} < \epsilon < \frac12$, and $\delta \ge e^{-n/16} + e^{- (\epsilon/2)^{\ell} n /(342 (\ell+1)^2)}$.
Let \NE be a ReLU network $(W_0, \dots, W_{\ell})$ with $n$ input nodes, $\ell \ge 1$ hidden layers with $n$ hidden layer nodes each, and a single output node. Suppose that
$n \ge 171 \cdot (\ell+1)^2 \cdot (2/\epsilon)^{2\ell} \cdot (\ln(2/\delta) + \ell \ln n)$.
Assuming that~$\ell$ is constant, tester \textsc{ORTesterMHL} queries $\Theta(\epsilon^{-8\ell} \cdot \ln^4(1/\lambda\epsilon))$ entries from $W_0, \dots, W_{\ell}$, and
\begin{enumerate}[(i)]
\item rejects with probability at least $1 - \lambda$, if \NE is $(\epsilon,\delta)$-far from computing the OR-function, and
\item accepts with probability at least $1 - \lambda$, if \NE computes the OR-function.
\qed
\end{enumerate}
\end{theorem}


\section{Testing ReLU networks with\emph{ multiple layers and outputs}}
\label{sec:ReLU-testing-constant-function-2-sided-mL-mo}

Let us finally demonstrate how \cref{thm:ReLU-testing-0-function-2-sided-mL,thm:ReLU-testing-OR-function-2-sided-mL} from \cref{sec:multiple-layers}
can be combined with the approach from \cref{sec:multiple-outputs-reduction} to obtain a tester of near constant functions in ReLU networks with multiple layers and multiple outputs. In this section, we will prove the following theorem that generalizes \cref{ReLU-testing-constant-function-2-sided-mo} (\cref{sec:multiple-outputs-applications}) to an arbitrary number of hidden layers.

\MultipleOutputsLayers*

\begin{remark}\rm
\label{remark:auxiliary-to-theorem:ReLU-testing-constant-function-2-sided-mL-mo}
The auxiliary parameter $\probp$ in our analysis, on one hand, puts some constraints on $m_0, \dots, m_{\ell}$ and on the other hand, allows us to minimize $\delta$. In general, one would hope for $\probp$ to be close to $2^{-\Theta(m_0)}$, and in particular, the arguments used in \cref{claim:auxiliary-to-lemma:output-mL} allow us to obtain $\probp$ as low as $e^{-\Theta((\epsilon/2)^{2\ell^2} \cdot \min_{0 \le k \le \ell}\{m_k\}/\ell^2)}$; for example, if $m_0 = \dots = m_{\ell}$, then we have $e^{-\Theta(m_0 \cdot (\epsilon/2)^{2\ell^2}/\ell^2)}$.
\end{remark}


As in \cref{sec:multiple-outputs-reduction} (for ReLU networks with a single hidden layer and multiple outputs) the proof of \cref{ReLU-testing-constant-function-2-sided-mL-mo} relies on a reduction (in \cref{thm:reduction-constant-ReLU-mL-mo}, \cref{sec:proof-thm:reduction-constant-ReLU-mL-mo}) of the problem to testing a near constant function to testing the constant $0$-function (\cref{thm:ReLU-testing-0-function-2-sided-mL}) and the OR-function (\cref{thm:ReLU-testing-OR-function-2-sided-mL}).

Before we will present our framework in \cref{sec:proof-thm:reduction-constant-ReLU-mL-mo}, let us first introduce some auxiliary tools for the study of ReLU networks with multiple layers and multiple outputs.


\subsection{Auxiliary tools in ReLU networks with multiple layers and \mbox{multiple outputs}}
\label{sec:ReLU-testing-constant-function-2-sided-mL-aux-tools}

Let us start with a useful claim that extends our arguments used in the proof of \cref{thm:reduction-constant-ReLU-mo} and shows a simple bound for the output value\footnote{The \emph{output value} of ReLU network $(W_0, \dots, W_{\ell})$ is equal to $f_{\ell+1}(x)$, using the terminology from \cref{def:ReLU-mol}.} of any ReLU network (with a single output node) that is $(\epsilon,\delta)$-close to computing a constant function.

\begin{claim}
\label{claim:if-close-then-not-too-large}
Let $(W_0, \dots, W_{\ell})$ be a ReLU network with $m_0$ input nodes, $\ell \ge 1$ hidden layers with $m_1, \dots, m_{\ell}$ hidden layer nodes each, and a single output node. Then,
\begin{itemize}
\item if $(W_0, \dots, W_{\ell})$ is $(\epsilon,\delta)$-close to computing the constant $0$-function, then for at least a $(1-\delta)$-fraction of the inputs the output value is at most $2 \cdot (\ell+1) \cdot \epsilon \cdot \prod_{i=0}^{\ell} m_i$, and
\item if $(W_0, \dots, W_{\ell})$ is $(\epsilon,\delta)$-close to computing the OR-function, then for at least a $(1-\delta)$-fraction of the inputs the output value is at least $- 2 \cdot (\ell+1) \cdot \epsilon \cdot \prod_{i=0}^{\ell} m_i$.
\end{itemize}
\end{claim}

\begin{proof}
Observe that a change of a weight of an edge in $E_0$ changes the value of the output node by at most $2$; similarly, one edge weight change in $E_1$ may change it by at most $2m_0$, and in general, one edge weight change in $E_k$ may change the value of the output node by at most $2 \prod_{i=0}^{k-1} m_i$. Therefore, if $(W_0, \dots, W_{\ell})$ is $(\epsilon,\delta)$-close to computing constant $0$-function, then for at least a $(1-\delta)$-fraction of the inputs the output value is at most
\begin{align*}
    \sum_{k=0}^{\ell-1} \left( (\epsilon m_k m_{k+1}) \cdot 2 \prod_{i=0}^{k-1} m_i \right) +
        (\epsilon m_{\ell}) \cdot 2 \prod_{i=0}^{\ell-1} m_i
\end{align*}
as otherwise, the output node could not be $(\epsilon,\delta)$-close to computing the constant $0$-function. Now, the first part of the claim follows since the value above is upper bounded by $2 \cdot (\ell+1) \cdot \epsilon \cdot \prod_{i=0}^{\ell} m_i$.

Identical arguments imply that if $(W_0, \dots, W_{\ell})$ is $(\epsilon,\delta)$-close to computing the OR-function, then for at least a $(1-\delta)$-fraction of the inputs the output value is at least
\begin{align*}
    - \sum_{k=0}^{\ell-1} \left( (\epsilon m_k m_{k+1}) \cdot 2 \prod_{i=0}^{k-1} m_i \right) +
        (\epsilon m_{\ell}) \cdot 2 \prod_{i=0}^{\ell-1} m_i
        \enspace,
\end{align*}
for otherwise the output node could not be $(\epsilon,\delta)$-close to computing the OR-function.
\end{proof}

(Let us also observe that if $m_i \ge 2$ for every $0 \le i \le \ell-1$, then in fact the sum used can be upper bounded by $6 \cdot \epsilon \cdot \prod_{i=0}^{\ell} m_i$, independently of $\ell$.)

\medskip

Our next lemma is a straightforward generalization of \cref{lemma:sampling-mL-existance} to multiple outputs.

\begin{lemma}
\label{lemma:sampling-mL-existance-mo}
Let $(W_0, \dots, W_{\ell})$ be a ReLU network with $m_0$ input nodes, $\ell \ge 1$ hidden layers with $m_1, \dots, m_{\ell}$ hidden layer nodes each, and $m_{\ell+1}$ output nodes.
Let $f:\{0,1\}^{m_0} \rightarrow \{0,1\}^{m_{\ell+1}}$ be the function computed by ReLU network $(W_0, \dots, W_{\ell})$, which is $f(x) := \sgn(\relu(f_{\ell+1}(x)))$, where
\begin{align*}
    f_i(x) &:=
        \begin{cases}
            W_{i-1} \cdot \relu(f_{i-1}(x)) & \text{ if } 1 \le i \le \ell+1, \\
            x & \text{ if } i=0.
        \end{cases}
\end{align*}
For any sampling matrices $S_0, \dots, S_{\ell}$ corresponding to samples $\sset_0, \dots, \sset_{\ell}$ of size $s_0, \dots, s_{\ell}$, respectively, let $h :\{0,1\}^{m_0} \rightarrow \{0,1\}^{m_{\ell+1}}$ be the function computed by ReLU network $(W_0 \cdot S_0, \dots, W_{\ell} \cdot S_{\ell})$, that is, $h$ is defined recursively as $h(x) := \sgn(\relu(h_{\ell+1}(x)))$, where
\begin{align*}
    h_i(x) &:=
        \begin{cases}
            W_{i-1} \cdot S_{i-1} \cdot \relu(h_{i-1}(x)) & \text{ if } 1 \le i \le \ell+1, \\
            x & \text{ if } i=0.
        \end{cases}
\end{align*}
Let $0 < \vartheta \le 1$ and $0 < \lambda < \frac{1}{\ell+1}$. Assume that $s_k \ge \frac{2 (\ell+1)^2 \ln(2 \prod_{i=k+1}^{\ell} s_i/(\lambda(\ell+1)))}{\vartheta^2}$ for every $0 \le k \le \ell$. Then there is at least one choice of sampling matrices $S_0, \dots, S_{\ell}$ such that for at most a $(\lambda \cdot m_{\ell+1})$-fraction of the inputs $x \in \{0,1\}^{m_0}$, for some $1 \le j \le m_{\ell+1}$ holds the following:
\begin{align*}
    \left| \prod_{i=0}^{\ell} \frac{m_i}{s_i} \cdot h^{(j)}_{\ell+1}(x) - f^{(j)}_{\ell+1}(x) \right|
        &>
    \vartheta \cdot \prod_{i=0}^{\ell} m_i
    \enspace,
\end{align*}
where $h^{(j)}_{\ell+1}(x)$ and $f^{(j)}_{\ell+1}(x)$ denote the $j$-th rows of $h_{\ell+1}(x)$ and $f_{\ell+1}(x)$, respectively.
\end{lemma}

\begin{proof}
This is a straightforward generalization of the arguments used in the proof of \cref{lemma:sampling-mL-existance}.

Let $\IN := \{0,1\}^{m_0}$ be the set of all input instances.
Let $\mathcal{S}$ be the set of all possible selections of the nodes in sets $\sset_0, \dots, \sset_{\ell}$.
Let us call tuple $\langle \sset_0, \dots, \sset_{\ell}\rangle \in \mathcal{S}$ to be \emph{bad for input $x_0 \in \IN$} if after fixing sets $\sset_0, \dots, \sset_{\ell}$ to define function $h(x_0)$, for some $1 \le j \le m_{\ell+1}$ we have
\begin{align*}
    \left| \prod_{i=0}^{\ell} \frac{m_i}{s_i} \cdot h^{(j)}_{\ell+1}(x_0) - f^{(j)}_{\ell+1}(x_0) \right|
        >
    \vartheta \cdot \prod_{i=0}^{\ell} m_i
    \enspace.
\end{align*}
For any $x \in \IN$ and $\mathfrak{s} \in \mathcal{S}$, we call pair $\langle x, \mathfrak{s} \rangle$ \emph{bad} if $\mathfrak{s}$ is bad for input $x$. Let us call a tuple $\mathfrak{s} \in \mathcal{S}$ \emph{$\alpha$-obstructive} if for more than a $\alpha$-fraction of inputs $x \in \IN$, pairs $\langle x, \mathfrak{s} \rangle$ are bad.

By \cref{lemma:sampling-mL} (and under the assumptions of \cref{lemma:sampling-mL}) for every $x \in \IN$ and every fixed output (corresponding to row $1 \le j \le m_{\ell+1}$ in $h_{\ell+1}(x)$ and $f_{\ell+1}(x)$), at most a $\lambda$ fraction of the tuples in $\mathcal{S}$ are bad, and hence, when we extend this over all $m_{\ell+1}$ outputs, we obtain that at most $\lambda \cdot m_{\ell+1} \cdot |\mathcal{S}|$ tuples are bad. Therefore at most $\lambda \cdot m_{\ell+1} \cdot |\IN| \cdot |\mathcal{S}|$ pairs $\langle x, \mathfrak{s} \rangle$ with $x \in \IN$ and $\mathfrak{s} \in \mathcal{S}$ are bad.

Next, observe that if there are $t$ tuples in $\mathcal{S}$ that are $(\lambda \cdot m_{\ell+1})$-obstructive, then there are more than $t \cdot \lambda \cdot m_{\ell+1} \cdot |\IN|$ pairs $\langle x, \mathfrak{s} \rangle$ with $x \in \IN$ and $\mathfrak{s} \in \mathcal{S}$ that are bad. However, since there are at most $\lambda \cdot m_{\ell+1} \cdot |\IN| \cdot |\mathcal{S}|$ pairs $\langle x, \mathfrak{s} \rangle$ with $x \in \IN$ and $\mathfrak{s} \in \mathcal{S}$ that are bad, we must have $t < |\mathcal{S}|$. This means that there are less than $|\mathcal{S}|$ tuples $\mathfrak{s} \in \mathcal{S}$ that are $(\lambda \cdot m_{\ell+1})$-obstructive, or equivalently, there is at least one tuple $\mathfrak{s} \in \mathcal{S}$ that is \emph{not} $(\lambda \cdot m_{\ell+1})$-obstructive. This completes the proof.
\end{proof}


\subsection{
Reducing multiple outputs networks to single output networks}
\label{sec:proof-thm:reduction-constant-ReLU-mL-mo}

In this section, we show that testing if a multiple-layer ReLU network with multiple outputs computes a near constant function \bb can be reduced to testing if some multiple-layer ReLU network with a single output compute constant $0$- or $1$-functions.
We begin with an auxiliary definition.

\begin{definition}
\label{def:network-restricted-to-output-node-j}
Let $\NE$ be an arbitrary ReLU network $(W_0, \dots, W_{\ell})$ with $m_0$ input nodes, $\ell \ge 1$ hidden layers with $m_1, \dots, m_{\ell}$ hidden layer nodes each, and $m_{\ell+1}$ output nodes. For any output node $j$, we define \textbf{ReLU network $\NE_j$ restricted to the output node $j$} to be the sub-network of $\NE$ restricted to the output node $i$ containing all input and all hidden layers. That is, $\NE_j$ is ReLU network $(W_0, \dots, W_{\ell-1}, W_{\ell}^{(j)})$, where $W_{\ell}^{(j)}$ is the $m_{\ell}$-vector corresponding to the $j$-th column in~$W_{\ell}$.
\end{definition}

The following theorem extends \cref{thm:reduction-constant-ReLU-mo} and shows that if a ReLU network with $m_{\ell+1}$ output bits is $(\epsilon, \delta)$-far from computing near constant function \bb, then for at least an $\frac12\epsilon$ fraction of the output bits $j$, the ReLU network restricted to the output node $j$ is uniformly $(\epsilon',\delta')$-far from computing an appropriate constant function, where $\epsilon' \sim (\epsilon/2)^{\ell}/\ell$, $\delta' \sim \delta/m_{\ell+1} - \probp$, and $\probp$ is a parameter bounding some inter-dependencies between $m_0, \dots, m_{\ell}$ (and one would hope for $\probp$ to be close to $2^{\Theta(m_0)}$).

\begin{theorem}
\label{thm:reduction-constant-ReLU-mL-mo}
Let \NE be a ReLU network with $m_0$ input nodes, $\ell \ge 1$ hidden layers with $m_1, \dots, m_{\ell}$ hidden layer nodes each, and $m_{\ell+1} \ge 2$ output nodes.
Let $\bb = (\ob_1, \dots, \ob_{m_{\ell+1}}) \in \{0,1\}^{m_{\ell+1}}$ and $\frac{1}{\min_{0 \le k \le \ell}\{m_k\}} < \epsilon < \frac12$.
Further, suppose that for an arbitrary parameter $0 < \probp < \frac{1}{\ell+1}$, it holds that $m_k \ge 145 \cdot (\ell+1)^2 \cdot (2/\epsilon)^{2\ell} \cdot \ln(\frac{1}{\probp} \prod_{i=k+1}^{\ell} m_i)$ for every $0 \le k \le \ell$.
If $\delta \ge \probp m_{\ell+1} + e^{-m_0/16}$ and \NE is $(\epsilon, \delta)$-far from computing near constant function \bb then there are more than $\frac12 \epsilon m_{\ell+1}$ output nodes such that for any such output node $j$, the ReLU network $\NE_j$ restricted to the output node $j$ is $(\epsilon',\delta')$-far from computing constant function $\ob_i$, where $\epsilon' := \frac{1}{17 \cdot (\ell+1)} \cdot \left(\frac{\epsilon}{2-\epsilon}\right)^{\ell}$ and $\delta' := \frac{\delta - e^{-m_0/16} - \probp \cdot m_{\ell+1}}{m_{\ell+1}}$.
\end{theorem}


\begin{proof}
Fix $\bb \in \{0,1\}^{m_\ell+1}$ and recall that $\epsilon' := \frac{1}{17 \cdot (\ell+1)} \cdot \left(\frac{\epsilon}{2-\epsilon}\right)^{\ell}$ and $\delta' := \frac{\delta - e^{-m_0/16} - \probp \cdot m_{\ell+1}}{m_{\ell+1}}$; hence $\delta = (\probp + \delta') \cdot m_{\ell+1} + e^{-m_0/16}$. Let us set $\vartheta := 2 \cdot (\ell+1) \cdot \epsilon' = \frac{2}{17} \cdot \left(\frac{\epsilon}{2-\epsilon}\right)^{\ell}$, and observe that then,
\begin{align}
\label{ineq:ReLU-testing-constant-function-2-sided-mL-mo}
    \left(2 \cdot (\ell+1) \cdot \epsilon' + \vartheta\right) \cdot (1-\tfrac12 \epsilon)^{\ell} - \tfrac14 \cdot (\epsilon/2)^{\ell}
        &=
    - \frac{1}{58} \cdot \left(\frac{\epsilon}{2 - \epsilon}\right)^{\ell} \cdot (1-\tfrac12 \epsilon)^{\ell}
        < 0
        \enspace.
\end{align}

Let $\OV_0$ be the set of output nodes $j$ with $\ob_j = 0$ and $\OV_1$ be the set of remaining output nodes, that is, the output nodes $j$ with $\ob_j = 1$.
Let $\FN$ be the set of output nodes $j$ such that the restricted network $\NE_j$ is $(\epsilon',\delta')$-far from computing constant function $\ob_j$.
Let $\FN_0 := \FN \cap \OV_0 = \{ j \in \FN: \ob_j = 0 \}$ and let $\FN_1 := \FN \cap \OV_1 = \{ j \in \FN: \ob_j = 1 \} = \FN \setminus \FN_0$.

\medskip

The proof is by contradiction. Assume that there are at most $\frac12 \epsilon r$ output nodes $i$ such that the restricted network $\NE_i$ is $(\epsilon',\delta')$-far from computing constant function $\ob_i$, that is, that $|\FN| \le \frac12 \epsilon r$. We will show that then we can modify weights of at most an $\epsilon$-fraction of edges in every level to obtain a ReLU network that computes \bb for at least a $(1-\delta)$-fraction of the inputs. This would contradict the assumption that $\NE$ is $(\epsilon,\delta)$-far from computing constant function \bb, and thus, would complete the proof by contradiction.

\paragraph{Framework.}
Let us first define an auxiliary construction that we will use. We will rely on the approach presented earlier in the proof of \cref{lemma:output-mL} combined with \cref{lemma:sampling-mL-existance-mo}. In order to modify weights of some edges of the original ReLU network to obtain a  ReLU network that computes \bb for at least a $(1-\delta)$-fraction of the inputs, we will split the ReLU network \NE into two independent parts, one in which the nodes in the final hidden layer have a large value, and one in which the nodes in the final hidden layer have almost the same value as in the original network.

Consider applying \cref{lemma:sampling-mL-existance-mo} with $s^*_0 := m_0$ and $s^*_k := (1 - \frac12\epsilon) m_k$ for $1 \le k \le \ell$. The assumptions about the values of $m_0, \dots, m_{\ell}$ in \cref{thm:reduction-constant-ReLU-mL-mo} ensure%
that with setting $\vartheta := 
\frac{2}{17} \cdot \left(\frac{\epsilon}{2-\epsilon}\right)^{\ell}$, we have $s^*_k \ge \frac{2 (\ell+1)^2 \ln(2 \prod_{i=k+1}^{\ell} s^*_i/(\probp(\ell+1)))}{\vartheta^2}$ for every $0 \le k \le \ell$. Therefore, by \cref{lemma:sampling-mL-existance-mo}, there is a sequence of sets $\sset^*_0, \dots, \sset^*_{\ell}$ such that
\begin{itemize}
\item $s_k = |\sset^*_k|$ for every $0 \le k \le \ell$, and
\item if we take sampling matrices $S^*_0, \dots, S^*_{\ell}$ corresponding to the nodes from $\sset^*_0, \sset^*_1, \dots, \sset^*_{\ell}$, then for the function from $\{0,1\}^{m_0}$ to $\RR^{m_i}$:
    \begin{align*}
        h^*_i(x) :=
            \begin{cases}
                W_{i-1} \cdot S^*_{i-1} \cdot \relu(h^*_{i-1}(x)) & \text{ if } 1 \le i \le \ell+1, \\
                x & \text{ if } i=0,
            \end{cases}
    \end{align*}
for any $0 < \vartheta \le 1$, $0 < \delta < \frac{1}{\ell+1}$, for at most a $(\probp \cdot m_{\ell+1})$-fraction of the inputs $x \in \{0,1\}^{m_0}$, for some $1 \le j \le m_{\ell+1}$, we have,
\begin{align*}
    \left| \prod_{i=0}^{\ell} \frac{m_i}{s^*_i} \cdot h^{*(j)}_{\ell+1}(x) - f{(j)}_{\ell+1}(x) \right|
        &>
    \vartheta \cdot \prod_{i=0}^{\ell} m_i
    \enspace,
\end{align*}
where $h^{*(j)}_{\ell+1}(x)$ and $f^{(j)}_{\ell+1}(x)$ denote the $j$-th rows of $h^*_{\ell+1}(x)$ and $f_{\ell+1}(x)$, respectively.

Since $s^*_0 := m_0$ and $s^*_k := (1-\frac12 \epsilon) m_k$ for $1 \le k \le \ell$, we can simplify this bound to obtain that for all but a $(\probp \cdot m_{\ell+1})$-fraction of the inputs $x \in \{0,1\}^{m_0}$, for all $1 \le j \le m_{\ell+1}$,
\begin{align}
\label{ineq:combo-lemma:output-mL+lemma:sampling-mL-existance-mo}
    \left|h^{*(j)}_{\ell+1}(x) - (1-\tfrac12 \epsilon)^{\ell} \cdot f^{(j)}_{\ell+1}(x) \right|
        &\le
    \vartheta \cdot (1-\tfrac12 \epsilon)^{\ell} \cdot \prod_{i=0}^{\ell} m_i
    \enspace.
\end{align}
\end{itemize}

The following simple claim follows directly from (\ref{ineq:combo-lemma:output-mL+lemma:sampling-mL-existance-mo}) and \cref{claim:NumberOfOnes}.
\begin{claim}
\label{claim:defining-I1}
Let $\mathfrak{I}_1$ be the set of all inputs $x \in \{0,1\}^{m_0}$ such that $\|x\|_1 \ge \frac14 m_0$ and
\begin{itemize}
\item $h^{*(j)}_{\ell+1}(x) \le (1-\tfrac12 \epsilon)^{\ell} \cdot f^{(j)}_{\ell+1}(x) + \vartheta \cdot (1-\tfrac12 \epsilon)^{\ell} \cdot \prod_{i=0}^{\ell} m_i$ for every $j \in \OV_0$, and
\item $h^{*(j)}_{\ell+1}(x) \ge (1-\tfrac12 \epsilon)^{\ell} \cdot f^{(j)}_{\ell+1}(x) - \vartheta \cdot (1-\tfrac12 \epsilon)^{\ell} \cdot \prod_{i=0}^{\ell} m_i$ for every $j \in \OV_1$.
\end{itemize}
Then $|\mathfrak{I}_1| \ge (1 - \probp \cdot m_{\ell+1} - e^{-m_0/16}) \cdot 2^{m_0}$.
\qed
\end{claim}

Let us also state a similar claim that follows directly from \cref{claim:if-close-then-not-too-large}.
\begin{claim}
\label{claim:defining-I2}
Let $\mathfrak{I}_2$ be the set of all inputs $x \in \{0,1\}^{m_0}$ for which for all output nodes $1 \le j \le m_{\ell+1}$,
\begin{itemize}
\item if $\NE_j$ is $(\epsilon',\delta')$-close to computing the $0$-function, then $f^{(j)}_{\ell+1}(x) \le 2 \cdot (\ell+1) \cdot \epsilon' \cdot \prod_{i=0}^{\ell} m_i$, and
\item if $\NE_j$ is $(\epsilon',\delta')$-close to computing the $1$-function, then $f^{(j)}_{\ell+1}(x) \ge - 2 \cdot (\ell+1) \cdot \epsilon' \cdot \prod_{i=0}^{\ell} m_i$.
\end{itemize}
Then $|\mathfrak{I}_2| \ge (1 - m_{\ell+1} \cdot \delta') \cdot 2^{m_0}$.
\qed
\end{claim}

We can combine Claims \ref{claim:defining-I1} and \ref{claim:defining-I2}, to immediately obtain the following.
\begin{claim}
\label{claim:defining-I}
Let $\mathfrak{I} := \mathfrak{I}_1 \cap \mathfrak{I}_2$. Then
\begin{itemize}
\item $|\mathfrak{I}| \ge (1 - (\probp + \delta') \cdot m_{\ell+1} - e^{-m_0/16}) \cdot 2^{m_0}$
\item $\|x\|_1 \ge \frac14 m_0$ for every $x \in \mathfrak{I}$, and
\item for all output nodes $1 \le j \le m_{\ell+1}$,
\begin{itemize}[$\blacktriangleright$]
\item if $\NE_j$ is $(\epsilon',\delta')$-close to computing the constant $0$-function, then
    \begin{align*}
        h^{*(j)}_{\ell+1}(x) &\le
            \left(2 \cdot (\ell+1) \cdot \epsilon' + \vartheta\right) \cdot (1-\tfrac12 \epsilon)^{\ell} \cdot \prod_{i=0}^{\ell} m_i
            \enspace,
    \end{align*}
\item if $\NE_j$ is $(\epsilon',\delta')$-close to computing the constant $1$-function, then
    \begin{align*}
        h^{*(j)}_{\ell+1}(x) &\ge
        - \left(2 \cdot (\ell+1) \cdot \epsilon' + \vartheta\right) \cdot (1-\tfrac12 \epsilon)^{\ell} \cdot \prod_{i=0}^{\ell} m_i
        \enspace.
    \end{align*}
\end{itemize}
\end{itemize}
\end{claim}

\paragraph{Construction of the modified ReLU network.}
Let $\IN_k := V_k \setminus \sset^*_k$ for every $1 \le k \le \ell$. Now, we will modify weights of some edges in the ReLU network.
\begin{itemize}
\item For every $1 \le k \le \ell$, we change to $+1$ the weights of all edges connecting the nodes from $V_{k-1}$ to the nodes from $\IN_k$;
    \begin{itemize}[$\blacktriangleright$]
    \item the value of every node in $\IN_k$ is at least $\|x\|_1 \cdot \prod_{i=1}^{k-1} |\IN_i|$;
    \item this involves $m_{k-1} \cdot |\IN_k| \le \frac12 \epsilon m_{k-1} m_k$ changes of edge weights from $E_{k-1}$, $1 \le k \le \ell$.
    \end{itemize}
\item For every $1 \le k \le \ell-1$, we change to $0$ the weights of all edges connecting the nodes from $\IN_k$ to the nodes from $\sset^*_{k+1}$;
    \begin{itemize}[$\blacktriangleright$]
    \item this involves $|\IN_k| \cdot |\sset^*_{k+1}| \le \frac12 \epsilon m_k m_{k+1}$ changes of weights of $E_k$ edges.
    \end{itemize}
\item If there are at most $\frac12 \epsilon m_{\ell+1}$ output nodes $i$ such that the restricted network $\NE_i$ is $(\epsilon',\delta')$-far from computing near constant function \bb (i.e., when $|\FN| \le \frac12 \epsilon m_{\ell+1}$), then for each $i \in \FN$, we set to $\ob_i$ weights of all edges from $E_{\ell}$ that are incident to the output node $i$;
    \begin{itemize}[$\blacktriangleright$]
    \item each output node in $\FN_0$ computes the constant $0$-function (since all edges incident to the output node $i$ have weight $0$);
    \item each output node in $\FN_1$ computes the OR-function (since all edges incident to the output node $i$ have weight $1$ and the nodes from $\IN_{\ell}$ have value at least $\|x\|_1 \cdot \prod_{i=1}^{\ell-1} |\IN_i|$); 
    \item this involves $|\FN| m_{\ell} \le \frac12 \epsilon m_{\ell} m_{\ell+1}$ changes of weights of edges from $E_{\ell}$.
    \end{itemize}
\item For the nodes in $\IN_{\ell}$, we change to $-1$ the weights of all edges connecting nodes from $\IN_{\ell}$ to the output nodes in $\OV_0 \setminus \FN$, and change to $1$ the weights of all edges connecting nodes from $\IN_{\ell}$ to the output nodes in $\OV_1 \setminus \FN$;
    \begin{itemize}[$\blacktriangleright$]
    \item the contribution of every node in $\IN_{\ell}$ to any output node in $\OV_0 \setminus \FN$ is negative, and is at most $- \|x\|_1 \cdot \prod_{i=1}^{\ell-1} |\IN_i|$;
    \item the contribution of every node in $\IN_{\ell}$ to any output node in $\OV_1 \setminus \FN$ is positive, and is at least $\|x\|_1 \cdot \prod_{i=1}^{\ell-1} |\IN_i|$;
    \item this involves $|\IN_{\ell}| \cdot (|V_{\ell+1} \setminus \FN|) \le \frac12 \epsilon m_{\ell} m_{\ell+1}$ changes of weights of $E_{\ell}$ edges.
    \end{itemize}
\end{itemize}

Observe that in total, we modified
\begin{inparaenum}[(i)]
\item $m_0 \cdot |\IN_1| \le \frac12 \epsilon m_0 m_1$ weights of $E_0$ edges,
\item for every $1 \le k \le \ell-1$, $m_k \cdot |\IN_{k+1}| + |\IN_k| \cdot |\sset^*_{k+1}| \le \epsilon m_k m_{k+1}$ weights of $E_k$ edges, and
\item $|\FN| \cdot m_{\ell} + |\IN_{\ell}| \cdot (|V_{\ell+1} \setminus \FN|) = \frac12 \epsilon m_{\ell} m_{\ell+1}$ weights of $E_{\ell}$ edges.
\end{inparaenum}
Therefore, since the original ReLU network is $(\epsilon, \delta)$-far from computing near constant function \bb, our network obtained after modifying at most an $\epsilon$-fraction of weights at each level cannot return \bb for all but a $\delta$-fraction of the inputs.

Let us look at the output values returned by our modified network on an arbitrary input $x \in \mathfrak{I}$ (as defined in \cref{claim:defining-I}); we will show that for inputs from $\mathfrak{I}$, it computes near constant function~\bb.
\begin{itemize}
\item each output node in $\FN_0$ computes the constant $0$-function;
    \begin{itemize}[$\blacktriangleright$]
    \item follows directly from our construction and arguments above;
    \end{itemize}
\item each output node in $\FN_1$ computes the constant $1$-function;
    \begin{itemize}[$\blacktriangleright$]
    \item each output node in $\FN_1$ computes the OR-function but since for any input $x \in \mathfrak{I}$ we have $\|x\|_1 \ge \frac14 m_0$, the computed function is the constant $1$-function for any $x \in \mathfrak{I}$;
    \end{itemize}
\item for any other output node $j \in \OV_0 \setminus \FN_0$, the modified network returns the value which is at most $- (\epsilon/2)^{\ell} \cdot \|x\|_1 \cdot \prod_{i=1}^{\ell} m_i + h^{*(j)}_{\ell+1}(x)$;
    \begin{itemize}[$\blacktriangleright$]
    \item by \cref{claim:defining-I}, for any output node $j \in \OV_0 \setminus \FN_0$, for every input $x \in \mathfrak{I}$, the modified network returns the value at most
        \begin{align*}
            \left(\left(2 \cdot (\ell+1) \cdot \epsilon' + \vartheta\right) \cdot (1-\tfrac12 \epsilon)^{\ell}
                - \frac14 \cdot (\epsilon/2)^{\ell}\right)
                \cdot \prod_{i=0}^{\ell} m_i
            \enspace,
        \end{align*}
        which is negative by (\ref{ineq:ReLU-testing-constant-function-2-sided-mL-mo}), and hence $\NE_j$ returns $0$.
    \end{itemize}
\item for any other output node $j \in \OV_1 \setminus \FN_1$, the modified network returns the value which is at least $(\epsilon/2)^{\ell} \cdot \|x\|_1 \cdot \prod_{i=1}^{\ell} m_i + h^{*(j)}_{\ell+1}(x)$;
    \begin{itemize}[$\blacktriangleright$]
    \item by \cref{claim:defining-I}, for any output node $j \in \OV_1 \setminus \FN_1$, for every input $x \in \mathfrak{I}$, the modified network returns the value at least
        \begin{align*}
            \left(\frac14 \cdot (\epsilon/2)^{\ell}
                - \left(2 \cdot (\ell+1) \cdot \epsilon' + \vartheta\right) \cdot (1-\tfrac12 \epsilon)^{\ell}\right)
                \cdot \prod_{i=0}^{\ell} m_i
            \enspace,
        \end{align*}
        which is positive by (\ref{ineq:ReLU-testing-constant-function-2-sided-mL-mo}), and hence $\NE_j$ returns $1$.
    \end{itemize}
\end{itemize}

Let us summarize the discussion above. We took a ReLU network \NE that is $(\epsilon, \delta)$-far from computing near constant function \bb and then assumed, for the purpose of contradiction, that there are at most $\frac12 \epsilon r$ output nodes $j$ such that the restricted network $\NE_j$ is $(\epsilon',\delta')$-far from computing constant function $\ob_j$, that is, that $|\FN| \le \frac12 \epsilon r$. Then, we have modified the original ReLU network \NE in at most an $\epsilon$-fraction of weights at each level and obtained a network that for all inputs $x \in \mathfrak{I}$ computes near constant function \bb:
\begin{inparaenum}[(i)]
\item every output node in $\FN_0$ always computes the constant $0$-function,
\item every output node in $\FN_1$ always computes the constant $1$-function,
\item every output node in $\OV_0 \setminus \FN_0$ computes the constant $0$-function for all inputs $x \in \mathfrak{I}$, and
\item every output node in $\OV_1 \setminus \FN_1$ computes the constant $1$-function for all inputs $x \in \mathfrak{I}$.
\end{inparaenum}
However, since $\delta' := \frac{\delta - e^{-m_0/16} - \probp \cdot m_{\ell+1}}{m_{\ell+1}}$, by \cref{claim:defining-I} we have $|\mathfrak{I}| \ge (1 - (\probp + \delta') \cdot m_{\ell+1} - e^{-m_0/16}) \cdot 2^{m_0} = (1 - \delta) \cdot 2^{m_0}$, this is a contradiction to the assumption that the original ReLU network is $(\epsilon, \delta)$-far from computing near constant function \bb. This implies that there are more than $\frac12 \epsilon r$ output nodes $i$ such that the restricted network $\NE_i$ is $(\epsilon',\delta')$-far from computing constant function $\ob_i$, completing the proof of \cref{ReLU-testing-constant-function-2-sided-mL-mo}.
\end{proof}


\subsection{Testing a near constant function: Proof of \cref{ReLU-testing-constant-function-2-sided-mL-mo}}
\label{sec:proof-of-ReLU-testing-constant-function-2-sided-mL-mo}

With \cref{thm:reduction-constant-ReLU-mL-mo} at hand, \cref{ReLU-testing-constant-function-2-sided-mL-mo} of testing a constant function in ReLU networks with multiple hidden layers and multiple outputs follows easily from our testing of the constant $0$-function (\cref{thm:ReLU-testing-0-function-2-sided-mL}) and the OR-function (\cref{thm:ReLU-testing-OR-function-2-sided-mL}) in ReLU networks with multiple hidden layers and a single output.

By \cref{thm:reduction-constant-ReLU-mL-mo}, if \NE is $(\epsilon,\delta)$-far from computing constant function \bb then there are more than $\frac12 \epsilon m_{\ell+1}$ output nodes such that for any such output node $j$, the ReLU network $\NE_j$ restricted to the output node $j$ is $\left(\frac{1}{17 (\ell+1)} \left(\frac{\epsilon}{2-\epsilon}\right)^{\ell}, \frac{\delta - e^{-m_0/16} - \probp \cdot m_{\ell+1}}{m_{\ell+1}}\right)$-far from computing constant function $\ob_j$.

Therefore \cref{ReLU-testing-constant-function-2-sided-mL-mo} follows by sampling $O(1/\epsilon)$ output nodes uniformly at random and then, for each sampled output node $j$, testing whether the ReLU network $\NE_j$ restricted to the output node $j$ computes constant function $\ob_j$ using the tester for the constant $0$-function (\cref{thm:ReLU-testing-0-function-2-sided-mL}) or the OR-function (\cref{thm:ReLU-testing-OR-function-2-sided-mL}) in ReLU networks with multiple hidden layers and a single output. By outputting the majority outcome of the calls to these testers, we obtain the promised property testing algorithm.
\junk{
\Artur{Remember the key constraints in \cref{thm:ReLU-testing-0-function-2-sided-mL,thm:ReLU-testing-OR-function-2-sided-mL}:
$0 < \probp < \frac{1}{\ell+1}$,
$\delta' \ge \probp + e^{-m_0/16}$,
$\frac{1}{\min_{0 \le k \le \ell}\{m_k\}} < \epsilon'$,
$m_0 \ge 128 \cdot (\ell+1)^2 \cdot (2/\epsilon')^{2\ell} \cdot \ln(\frac{1}{\probp} \prod_{i=1}^{\ell} m_i)$ and
$m_k \ge 171 \cdot (\ell+1)^2 \cdot (2/\epsilon')^{2\ell} \cdot \ln(\frac{1}{\probp} \prod_{i=k+1}^{\ell} m_i)$ for $1 \le k \le \ell$.
}
}

Observe that the complexity of the tester is $O(1/\epsilon)$ times the complexity of the tester of the constant $0$-function (\cref{thm:ReLU-testing-0-function-2-sided-mL}) or the OR-function (\cref{thm:ReLU-testing-OR-function-2-sided-mL}) with parameters $\epsilon' := \frac{1}{17 \cdot (\ell+1)} \cdot \left(\frac{\epsilon}{2-\epsilon}\right)^{\ell}$ and $\delta' := \frac{\delta - e^{-m_0/16} - \probp \cdot m_{\ell+1}}{m_{\ell+1}}$, and hence, for a constant $\ell$, it is $\Theta(\epsilon^{-(8\ell^2+1)} \cdot \ln^4(1/\lambda \epsilon))$.

Notice that the parameters $\epsilon' := \frac{1}{17 \cdot (\ell+1)} \cdot \left(\frac{\epsilon}{2-\epsilon}\right)^{\ell}$ and $\delta' := \frac{\delta - e^{-m_0/16} - \probp \cdot m_{\ell+1}}{m_{\ell+1}}$ used in \cref{thm:ReLU-testing-0-function-2-sided-mL,thm:ReLU-testing-OR-function-2-sided-mL} impose constraints $\frac{1}{\min_{0 \le k \le \ell}\{m_k\}} < \epsilon'$ and $\delta' \ge \probp + e^{-m_0/16}$, which we incorporate into the constraints in \cref{ReLU-testing-constant-function-2-sided-mL-mo} as $\epsilon > 2 \sqrt[\ell]{\frac{17(\ell+1)}{\max_{0 \le k \le \ell} \{m_k\}}}$ and $\delta \ge 2 m_{\ell+1} (\probp + e^{-m_0/16})$. Finally, the use of $\epsilon'$ and $\delta'$ in \cref{thm:ReLU-testing-0-function-2-sided-mL,thm:ReLU-testing-OR-function-2-sided-mL} impose constraints on $m_0, \dots, m_{\ell}$, which we incorporate in \cref{ReLU-testing-constant-function-2-sided-mL-mo} as
$m_k \ge 171 \cdot 34^{2\ell} \cdot (\ell+1)^{2(\ell+1)} \cdot (2/\epsilon)^{2\ell^2} \cdot \ln(\frac{1}{\probp} \prod_{i=k+1}^{\ell} m_i)$ for all $0 \le k \le \ell$.
\qed




\section{Conclusions}
\label{sec:conclusions}

In this paper, we introduce a property testing model to study computational networks used as computational devices and illustrate our setting on a case study of simple ReLU neural networks. We believe that the property testing model developed in this paper establishes a novel and useful theoretical framework for the study of testing algorithms and offers a new and interesting viewpoint with a potential of advancing our understanding of the structure of neural networks (and, in general, of computational networks).

In the context of neural networks, given our poor understanding of the computational processes in modern deep networks (leading to serious questions on the trust and transparency of the underlying learning algorithms), we hope that the lens of property testing as developed in the framework considered in this paper will contribute to a better understanding of the relation of the network structure and the computed function or properties of that function. Indeed, since the analysis of property testing algorithms relies on establishing close links between the object's local structure and the tested property, we believe the study of simple yet fundamental models of ReLU networks may provide a new insight into such networks.

Although our paper makes the first step in the analysis of ReLU networks in the context of property testing, there is still a large number of open questions (from the property testing point of view) regarding such networks, including, among others, the following:
\begin{itemize}
\item What is the query complexity of testing dictatorship and juntas (see \cite{Blais2010} for a survey on testing juntas)?
\item Classify the constant time testable properties with one- and two-sided error.
\item Study tolerant testing \cite{PRR06} in our model.
\item Extend the results to inputs over the reals.
\end{itemize}
These are just a small number of immediate questions one can ask about our model. We believe that there is a wide variety of interesting questions and are confident to see further research in this direction.

\section{Acknowledgements}

We would like to thank Gereon Frahling for his valuable discussions regarding the computational model. We would like to thank Nithin Varma for helpful discussions during the initial phase of research that lead to this paper.




\phantomsection\addcontentsline{toc}{section}{References}
\bibliography{references}
\bibliographystyle{alpha}

\appendix

\newpage


\phantomsection\addcontentsline{toc}{section}{Appendix}
\centerline{\Huge\textbf{Appendix}}


\section{Concentration bounds}
\label{app:concentration-bounds}

For the sake of completeness, let us recall three auxiliary concentration bounds used in the paper, one by Hoeffding, another one by McDiarmid, and a straightforward application of Chernoff bound.

\subsection{Hoeffding inequality for random variables sampled without replacement}
\label{app:subsec-Hoeffind}
We begin with a version of the classical Hoeffding inequality \cite{Hoeffding63} applied to random variables sampled without replacement (we use here a formulation from Bardenet and Maillard \cite[Proposition~1.2]{BM15}).

\begin{lemma}
\label{lemma:Hoeffding}
Let $\mathcal{X} = (x_1, \dots, x_N)$ be a finite population of $N$ real points and $X_1, \dots, X_n$ be a random sample drawn without replacement from $\mathcal{X}$. Let
\begin{align*}
    a = \min_{1 \le i \le N} x_i
        \;\;\;\;\;\;\;\;\text{ and }\;\;\;\;\;\;\;
    b = \max_{1 \le i \le N} x_i.
\end{align*}
Then, for all $\epsilon > 0$,
\begin{align*}
    \Pr\left[\frac1n \sum_{i=1}^n X_i - \mu \ge \epsilon\right] &\le
    \exp\left(- \frac{2n\epsilon^2}{(b-a)^2}\right).
\end{align*}
where $\mu = \frac1N \sum_{i=1}^N x_1$ is the mean of $\mathcal{X}$.
\end{lemma}
We remark that a simple transformation implies
\begin{align*}
    \Pr\left[\left|\frac1n \sum_{i=1}^n X_i - \mu \right| \ge \epsilon\right] &\le 2
    \exp\left(- \frac{2n\epsilon^2}{(b-a)^2}\right).
\end{align*}

\subsection{McDiarmid's method of bounded differences}
\label{app:subsec-McDiarmid}
Our next concentration bound is the so-called method of bounded differences by McDiarmid (Theorem~3.1 and inequalities (3.2--3.4) in  \cite{McDiarmid98}).

\begin{lemma}[\textbf{McDiarmid's method of bounded differences}]
\label{lemma:McDiarmid}
Let $\mathbf{X} = (X_1, \dots, X_n)$ be a family of independent random variables with $X_k$ taking values in a set $A_k$ for each $k$. Suppose that the real-valued function $f$ defined on $\prod A_k$ satisfies
\begin{align*}
    |f(\bx) - f(\bx')| & \le c_k
\end{align*}
whenever the vectors $\bx$ and $\bx'$ differ only in the $k$th coordinate. Let $\mu$ be the expected value of the random variable $f(\bx)$. Then for any $t \ge 0$,
\begin{align*}
    \Pr\left[|f(\mathbf{X}) - \mu| \ge t \right] &\le 2 e^{-2t^2/\sum_{k=1}^nc_k^2} \enspace.
\end{align*}
There are also ``one-sided'' versions of the inequality above that holds for any $t \ge 0$:
\begin{align*}
    \Pr\left[f(\mathbf{X}) - \mu \ge t \right] &\le e^{-2t^2/\sum_{k=1}^nc_k^2}  \enspace,
        \\
    \Pr\left[f(\mathbf{X}) - \mu \le -t \right] &\le e^{-2t^2/\sum_{k=1}^nc_k^2}  \enspace.
\end{align*}
\end{lemma}

\subsection{Bounding the number of ones in a random vector}
\label{app:NumberOfOnes}

In our analysis, we will also use the following simple (and well-known) auxiliary claim.

\begin{claim}
\label{claim:NumberOfOnes}
Let $x \in \{0,1\}^n$ be chosen uniformly at random. Then with probability at least $e^{-n/16}$ the vector $x$ contains more than $\frac14n$ ones.
\end{claim}

\begin{proof}
Let $x \in \{0,1\}^n$ be chosen uniformly at random. Let $Y$ denote the random variable for the number of ones in $x$. We have $\Ex[Y] = \frac12 n$ and using Chernoff bound we obtain
\begin{align*}
    \Pr[Y \le \tfrac14 n] &= \Pr[Y \le \tfrac12 \cdot \Ex[Y]] \le e^{-\Ex[Y]/8} = e^{-n/16}. \qedhere
\end{align*}
\end{proof}

   




\end{document}